\documentclass[final,12pt]{elsarticle}

\makeatletter
\def\ps@pprintTitle{%
 \let\@oddhead\@empty
 \let\@evenhead\@empty
 \def\@oddfoot{}%
 \let\@evenfoot\@oddfoot}
\makeatother

\usepackage{tabularx}
\newcolumntype{Y}{>{\raggedright\arraybackslash}X}
\newcolumntype{Z}{>{\centering\arraybackslash}X}
\usepackage{graphicx}

\newlength\Colsep
\newtheorem{theorem}{Theorem}
\newtheorem{definition}[theorem]{Definition}
\newtheorem{proposition}[theorem]{Proposition}
\newtheorem{lemma}[theorem]{Lemma}
\newtheorem{corollary}[theorem]{Corollary}

\newenvironment{proof}{\textit{Proof.}}{\hfill $\ \square \ $ \\}

\usepackage{amsmath,amssymb,amsfonts,latexsym,xspace,xcolor,tikz}
\usetikzlibrary{shapes,shapes.geometric,arrows,fit,calc,positioning,automata,arrows}
\usetikzlibrary{patterns,hobby}

\tikzstyle{line}=[draw]

\definecolor{colornodo}{RGB}{102,194,165}
\definecolor{colornodotexto}{RGB}{0,0,0}
\definecolor{coloredge}{RGB}{41,41,255}

\definecolor{colorinter1}{RGB}{255,0,0}
\definecolor{colorinter2}{RGB}{0,200,0}
\definecolor{colorinter3}{RGB}{0,0,255}

\definecolor{cyan}{HTML}{007079}
\definecolor{darkcyan}{HTML}{333333}
\definecolor{orange}{rgb}{0.8,0,0}
\definecolor{darkyellow}{HTML}{333333}

\definecolor{colornododiag}{rgb}{0.8,0,0}
\definecolor{colornododiag2}{RGB}{0,0,255}
\definecolor{colornododiag3}{HTML}{666666}
\definecolor{colornododiag4}{RGB}{255,0,255}

\tikzset{ n1/.style={circle,scale=1.0},
n2/.style={circle,fill=colornododiag,scale=0.5},
n3/.style={circle,draw,fill=colornodo,draw=colornodo,text=colornodotexto,scale=0.9},
e1/.style={line width=0.3mm}, e3/.style={draw=coloredge,line
width=0.7mm}, inter/.style={line width=0.85mm},
e2/.style={draw=black,line width=0.5mm}, c1/.style={line
width=0.3mm}, c2/.style={line width=0.2mm} }

\newcommand{\vertex}[4][black]{
    \draw[#1, fill=#1, inner sep=0pt] (#2, #3) circle (0.12) node(#4){};
}

\newcommand{\wvertex}[3]{
    \draw[black, fill=white, inner sep=0pt] (#1, #2) circle (0.12) node(#3){};
}

\newcommand{\vertexLabel}[3][above]{
    \path (#2) node[#1]{#3};
}

\newcommand{\edge}[3][]{
    \draw[#1] (#2) -- (#3);
}

 \DeclareMathOperator{\pw}{pw}
 \DeclareMathOperator{\bw}{bw}
 \DeclareMathOperator{\thin}{thin}
\DeclareMathOperator{\pthin}{pthin}
\DeclareMathOperator{\indthin}{thin_{ind}}

\DeclareMathOperator{\indpthin}{pthin_{ind}}

 \DeclareMathOperator{\diam}{diam}

\newcommand{\Lr}{\textsf{L}}

\newcommand{\Lxy}{\textsf{\raisebox{\depth}{\scalebox{-1}[-1]{L}}}}

\date{}
\begin{document}

\begin{frontmatter}
  \title{Intersection models and forbidden pattern characterizations for 2-thin and proper 2-thin
  graphs}

\author[UBA,ICC]{Flavia Bonomo-Braberman}
\ead{fbonomo@dc.uba.ar}

\author[UBA]{Gast\'on Abel Brito}

\address[UBA]{Universidad de Buenos Aires. Facultad de Ciencias Exactas y Naturales. Departamento de Computaci\'on. Buenos Aires,
Argentina.}
\address[ICC]{CONICET-Universidad de Buenos Aires. Instituto de
Investigaci\'on en Ciencias de la Computaci\'on (ICC). Buenos
Aires, Argentina.}

\begin{abstract}
 The \emph{thinness} of a graph is a width parameter that generalizes some
properties of interval graphs, which are exactly the graphs of
thinness one. Graphs with thinness at most two include, for
example, bipartite convex graphs. Many NP-complete problems can be
solved in polynomial time for graphs with bounded thinness, given
a suitable representation of the graph. \emph{Proper thinness} is
defined analogously, generalizing proper interval graphs, and a
larger family of NP-complete problems are known to be polynomially
solvable for graphs with bounded proper thinness.

The complexity of recognizing 2-thin and proper 2-thin graphs is
still open. In this work, we present characterizations of 2-thin
and proper 2-thin graphs as intersection graphs of rectangles in
the plane, as vertex intersection graphs of paths on a grid (VPG
graphs), and by forbidden ordered patterns. We also prove that
independent 2-thin graphs are exactly the interval bigraphs, and
that proper independent 2-thin graphs are exactly the bipartite
permutation graphs.

Finally, we take a step towards placing the thinness and its
variations in the landscape of width parameters, by upper bounding
the proper thinness in terms of the bandwidth.
\end{abstract}
\begin{keyword}
bipartite permutation graphs, forbidden patterns, intersection
graphs of rectangles, interval bigraphs, thinness, VPG graphs
\end{keyword}
\end{frontmatter}

\section{Introduction}\label{sec:intro}

\nocite{BBOSSS-thin-coloring-arxiv, Mixed-thin-wg,
Raf-rec-int-big}

A large family of graph width parameters have been studied since
the introduction of treewidth in the 80's~\cite{R-S-minors3-pltw},
some of them very recently defined, like
twin-width~\cite{BKTW20twinw} in 2020. For each width parameter, a
growing family of NP-complete problems are known to be
polynomial-time solvable on graph classes of bounded width. Thus,
it is interesting to find structural properties of a graph class
that ensure bounded width.

These structural properties can be described, for instance, by
forbidden induced subgraphs, or by the existence of vertex
orderings avoiding a family of patterns, or by the existence of
certain intersection models.

It is also useful to know whether a width parameter $\rho$ can be
bounded by a function of another width parameter $\rho'$. In that
case, every class of graphs $\mathcal{C}$ that has bounded
$\rho'$, has also bounded $\rho$. As a trade-off, every problem
that admits an efficient algorithm parameterized by $\rho$, admits
also an efficient algorithm parameterized by $\rho'$.

\subsection{Our focus}

In this work we will focus on a width parameter called the
\emph{thinness} of a graph. Intuitively speaking, the thinness of
a graph $G$ is a measure of how close $G$ is to an interval graph,
with the interval graphs being the class of graphs with
thinness~1. Similarly, proper thinness measures how close $G$ is
to a proper interval graph, with the proper interval graphs being
the class of graphs with proper thinness~1.

This similarity allows to generalize techniques for (proper)
interval graphs to (proper) $k$-thin graphs. It is worth
mentioning that thinness and mim-width~\cite{VatshelleThesis} are
two of the few width parameters that are bounded on interval
graphs and allow such algorithmic generalizations. In the case of
thinness, a large family of combinatorial optimization problems
become polynomial-time solvable for graphs of bounded thinness
(given a suitable
representation)~\cite{B-D-thinness,B-M-O-thin-tcs,M-O-R-C-thinness}.
They can be described shortly as optimization versions of list
matrix partition problems with the possibility of adding some
cardinality constraints. This family generalizes, for instance,
maximum weighted clique (whose unweighted version is NP-complete
for graphs of mim-width at most~6, even when the representation is
given~\cite{VatshelleThesis}), maximum weighted independent set
(whose unweighted version is NP-complete on boxicity 2 graphs,
even when the rectangle model is
given~\cite{F-P-T-2-box-indep,I-A-2-box-indep}), and list
$t$-coloring with constant $t$. For graphs of bounded proper
thinness (given a suitable representation), some domination-like
constraints can be added to the problem
formulation~\cite{B-D-thinness}.  However, classes of bounded
thinness are rich enough that the coloring problem (number of
colors being part of the input) is NP-complete even for proper
2-thin graphs~\cite{BBOSSS-thin-coloring-arxiv}.


\subsection{Thinness, proper thinness, and (proper) independent thinness}

A graph $G=(V,E)$ is \emph{$k$-thin} if there exist an ordering
$v_1, \dots , v_n$ of $V$ and a partition of $V$ into $k$ classes
$(V^1,\dots,V^k)$ such that, for each triple $(r,s,t)$ with
$r<s<t$, if $v_r$, $v_s$ belong to the same class and $v_t v_r \in
E$, then $v_t v_s \in E$. Such an ordering and partition are
called \emph{consistent}. The minimum $k$ such that $G$ is
$k$-thin is called the \emph{thinness} of $G$ and is denoted by
$\thin(G)$. The thinness is unbounded on the class of all graphs,
and graphs with bounded thinness were introduced by Mannino,
Oriolo, Ricci and Chandran in~\cite{M-O-R-C-thinness} as a
generalization of \emph{interval graphs} (intersection graphs of
intervals of the real line), which are exactly the $1$-thin
graphs~\cite{Ola-interval}. Graphs of thinness at most two
include, for example, convex bipartite
graphs~\cite{B-B-M-P-convex-wads}.

In~\cite{B-D-thinness}, the concept of \emph{proper thinness} is
defined in order to obtain an analogous generalization of
\emph{proper interval graphs} (intersection graphs of intervals of
the real line such that no interval properly contains another). A
graph $G=(V,E)$ is \emph{proper $k$-thin} if there exist an
ordering and a partition of $V$ into $k$ classes such that both
the ordering and its reverse are consistent with the partition.
Such an ordering and partition are called \emph{strongly
consistent}.  The minimum $k$ such that $G$ is proper $k$-thin is
called the \emph{proper thinness} of $G$ and is denoted by
$\pthin(G)$. Proper interval graphs are exactly the proper
$1$-thin graphs~\cite{Rob-uig}, and in~\cite{B-D-thinness} it is
proved that the proper thinness is unbounded on the class of
interval graphs. Examples of thin representations of graphs are
shown in Figure~\ref{fig:thinness-ex}.

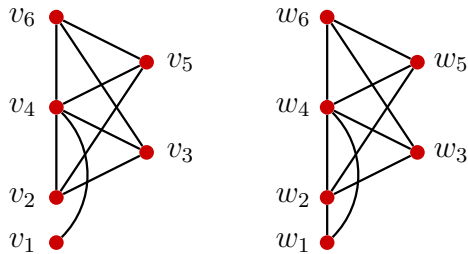
\begin{figure}
\begin{center}
\begin{tikzpicture}[scale=0.6]

\node(v1) at (1, 1)  [n2] {}; \node(v2) at (1, 2)  [n2] {};
\node(v3) at (3, 3)  [n2] {}; \node(v4) at (1, 4)  [n2] {};
\node(v5) at (3, 5)  [n2] {}; \node(v6) at (1, 6)  [n2] {};

\node(w1) at (7, 1)  [n2] {}; \node(w2) at (7, 2)  [n2] {};
\node(w3) at (9, 3)  [n2] {}; \node(w4) at (7, 4)  [n2] {};
\node(w5) at (9, 5)  [n2] {}; \node(w6) at (7, 6)  [n2] {};

\node(lv1) at (0.25, 1) {$v_1$}; \node(lv2) at (0.25, 2) {$v_2$};
\node(lv3) at (3.75, 3) {$v_3$}; \node(lv4) at (0.25, 4) {$v_4$};
\node(lv5) at (3.75, 5) {$v_5$}; \node(lv6) at (0.25, 6) {$v_6$};

\node(lw1) at (6.25, 1)  {$w_1$}; \node(lw2) at (6.25, 2)
{$w_2$}; \node(lw3) at (9.75, 3)  {$w_3$}; \node(lw4) at (6.25, 4)
{$w_4$}; \node(lw5) at (9.75, 5)  {$w_5$}; \node(lw6) at (6.25, 6)
{$w_6$};

\path (v1) edge [e1,bend right=45] (v4); \path (v4) edge [e1]
(v2); \path (v4) edge [e1] (v3); \path (v4) edge [e1] (v5); \path
(v4) edge [e1] (v6); \path (v2) edge [e1] (v3); \path (v2) edge
[e1] (v5); \path (v3) edge [e1] (v6); \path (v5) edge [e1] (v6);

\path (w1) edge [e1,bend right=45] (w4); \path (w4) edge [e1]
(w2); \path (w4) edge [e1] (w3); \path (w4) edge [e1] (w5); \path
(w4) edge [e1] (w6); \path (w2) edge [e1] (w3); \path (w2) edge
[e1] (w5); \path (w3) edge [e1] (w6); \path (w5) edge [e1] (w6);
\path (w1) edge [e1] (w2);

\end{tikzpicture}
\end{center}

\caption{A 2-thin graph (whose proper thinness is 3) and a proper
2-thin graph. The vertices are ordered increasingly by their
$y$-coordinate, and the classes correspond to the vertical lines,
i.e., $\{v_1,v_2,v_4,v_6\}$ and $\{v_3,v_5\}$ are the classes of
the first graph, $\{w_1,w_2,w_4,w_6\}$ and $\{w_3,w_5\}$ are the
classes of the second graph. }\label{fig:thinness-ex}
\end{figure}

In~\cite{BGOSS-thin-oper}, the concept of \emph{(proper)
independent thinness} was introduced in order to bound the
(proper) thinness of the lexicographical and direct products of
two graphs. In this case it is required, additionally, the classes
of the partition to be independent sets. These concepts are
denoted by $\indthin(G)$ and $\indpthin(G)$, respectively.

\subsection{Algorithmic aspects of thinness}

The recognition problem for $k$-thin graphs is
NP-complete~\cite{Shitov-thin}, and the recognition problem for
proper $k$-thin graphs is open. Both problems remain open when $k
\geq 2$ is fixed. Some related algorithmic problems were also
studied: partition into a minimum number of classes (strongly)
consistent with a given vertex
ordering~\cite{B-D-thinness,B-M-O-thin-tcs}, existence of a vertex
ordering (strongly) consistent with a given vertex
partition~\cite{B-D-thinness}.

In this work, in order to prove the intersection model characterizations, we will deal with the problem of the existence of a vertex ordering (strongly) consistent with a given vertex partition and that extends a partial order of the vertices that is a total order when restricted to each of the parts. We will call this problem \textsc{(Strongly) Consistent Extending Order}, or \textsc{(S)CEO}, and we will solve it in Section~\ref{sec:algo}. \\

\noindent {\textsc{(Strongly) Consistent Extending Order -- (S)CEO}}\\
\textbf{Instance:} A graph $G$, a partition $\Pi =
\{V^1,\dots,V^k\}$ and a partial order ${<}$ of $V(G)$ that is
total and (strongly) consistent restricted to each $V^j$, $1 \leq j \leq k$. \\
\textbf{Question:} Does there exist a total ordering of $V(G)$ extending $<$ and (strongly) consistent with $\Pi$?\\

Also, as a corollary of characterization theorems for independent
2-thin graphs and proper independent 2-thin graphs in
Section~\ref{sec:pat}, it follows that both classes can be
recognized in polynomial time.

We summarize in Table~\ref{tab:algo} the computational complexity
of the different algorithmic problems related to (proper)
(independent) thinness.

\begin{table}[tp]
\small
    \centering
\renewcommand{\tabularxcolumn}[1]{>{\arraybackslash}m{#1}}
    \begin{tabularx}{\textwidth}{|Y|Z|Z|Z|}
    \hline
        \textbf{Question}  & \textbf{Consistency} & \textbf{Strong Consistency} & \textbf{References} \\ \hline
    Existence of $k$-partition and order & NP-c & open & \cite{Shitov-thin} \\ \hline
    Existence of $k$-partition and order, fixed $k$ & open for $k \geq 2$, \linebreak P for $k=1$  & open for $k \geq 2$,  \linebreak P for $k=1$  & \cite{B-L-PQtrees,F-M-M-prop-int} \\ \hline
        Minimum partition, order given & P & P & \cite{B-M-O-thin-tcs,B-D-thinness} \\ \hline
        Existence of order, $k$-partition given & NP-c & NP-c & \cite{B-D-thinness} \\ \hline
        Existence of order, $k$-partition given, fixed $k$ & open for $k \geq 2$, \linebreak P for $k=1$ & open for $k \geq 2$, \linebreak P for $k=1$ & \cite{B-L-PQtrees,F-M-M-prop-int} \\ \hline
        \textsc{(S)CEO} & P & P & Corollary~\ref{cor:sceo} \\ \hline
            Existence of independent $k$-partition and order &  open for input $k$ \linebreak or fixed $k \geq 3$,  \linebreak P for $k=2$, \linebreak trivial for $k=1$ & open for input $k$ \linebreak or fixed $k \geq 3$,  \linebreak P for $k=2$, \linebreak trivial for $k=1$ & Cors~\ref{cor:indep-2-thin-rec} and~\ref{cor:indep-2-pthin-rec} \\ \hline
                Minimum independent partition, order given  &  P & P & Corollary~\ref{cor:ind-cocomp} \\ \hline
    \end{tabularx}
    \caption{Survey of the computational complexity of algorithmic questions related to thinness, proper thinness, independent thinness, and proper independent thinness.}
    \label{tab:algo}
\end{table}


\subsection{Graph intersection models}

Several graph classes are defined by means of a geometrical
intersection model, being the most prominent of such classes the
class of interval graphs, introduced by Haj\'os in
1957~\cite{Haj-int}. Moreover, many of these classes are
generalizations of interval graphs, like circular-arc
graphs~\cite{Gav-circ-arc}, vertex and edge intersection graphs of
paths on a grid~\cite{asinowski,Golumbic-epg}, and graphs with
bounded boxicity. The \emph{boxicity} of a graph, introduced by
Roberts in 1969~\cite{Rob-box}, is the minimum dimension in which
a given graph can be represented as an intersection graph of
axis-parallel boxes. Chandran, Mannino, and Oriolo
in~\cite{C-M-O-thinness-man} proved that boxicity 2 graphs have
unbounded thinness, but $k$-thin graphs have boxicity at most $k$.
In particular, 2-thin graphs are a subclass of boxicity 2 graphs,
i.e., intersection graphs of axis-parallel rectangles in the
plane. Their proof is constructive, and in their boxicity 2 model
for 2-thin graphs, the upper-right corners of the rectangles lie
in two diagonals, according to the class the corresponding vertex
belongs to. We call this a 2-diagonal model, and we show in
Propositions~\ref{prop:rep} and~\ref{prop:th3} that there are
graphs having a 2-diagonal model which are not 2-thin. Thus, we do
in Section~\ref{sec:box} a slight modification to the model
in~\cite{M-O-R-C-thinness} for 2-thin graphs, needed to obtain,
together with the 2-diagonal property, a characterization. Namely,
we modify it to satisfy a further property that we call
\emph{blocking}, and we prove in Theorem~\ref{thm:char-2-thin}
that a graph is 2-thin if and only if it has a blocking 2-diagonal
model. We obtain in Theorem~\ref{thm:char-2-pthin} a similar
characterization for proper 2-thin graphs.

Notice that when restricting the upper-right corners of the
rectangles to lie in one diagonal, we obtain the class of interval
graphs, i.e., 1-thin graphs. The definition of \emph{p-box
graphs}~\cite{S-T-p-box} ``looks'' similar, since they are the
intersection graphs of rectangles whose lower-right corners lie in
a diagonal, but the classes of interval graphs and p-box graphs
are very different. The models with this kind of restrictions,
like endpoints, corners or sides of the geometrical objects lying
on a line, are known in the literature as \emph{grounded models}
(see Figure~\ref{fig:classes}).

\begin{figure}
\begin{center}
\includegraphics[width=.9\textwidth]{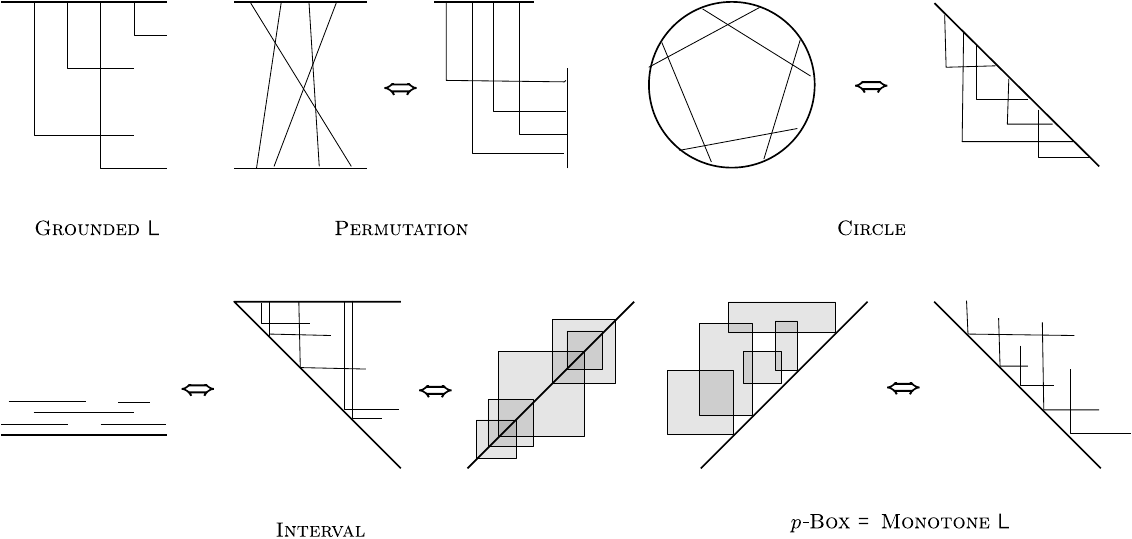}
\caption{Other graph classes defined or characterized by grounded
box models or grounded \textsf{L}-models.}\label{fig:classes}
\end{center}
\end{figure}

A graph is \emph{$B_k$-VPG} (resp. \emph{$B_k$-EPG}) if it is the
vertex (resp. edge) intersection graph of paths with at most
\emph{$k$ bends} in a grid~\cite{asinowski,Golumbic-epg}. VPG
graphs, without bounds in the number of bends, are also known as
\emph{string graphs}. A subclass of $B_1$-VPG graphs is the class
of \Lr-graphs, in which all the paths have the shape \Lr. Many
classes can be characterized by different grounded \Lr-models. For
instance, circle graphs are exactly the doubly grounded \Lr-graphs
where both endpoints of the paths belong to an inverted
diagonal~\cite{asinowski}, and p-box graphs are also characterized
as \emph{monotone} \Lr-graphs~\cite{monot-L} (\Lr-graphs such that
the bends of the \Lr \ shapes lie on an inverted diagonal). A very
nice survey on this kind of models can be found
in~\cite{J-T--L-graphs}. We present in Section~\ref{sec:vpg}
another grounded rectangle model for 2-thin graphs that gives rise
to a grounded \Lr-model for them and a grounded $B_0$-VPG model
for independent 2-thin graphs. Based on these models, we can also
obtain a $B_3$-VPG representation for 3-thin graphs and a
$B_1$-VPG representation for independent 3-thin graphs. We
furthermore show that $B_0$-VPG graphs have unbounded thinness,
and that not every 4-thin graph is a VPG graph.

The \Lr-model can be modified to prove that 2-thin graphs are
monotone \Lr-graphs. On the one hand, the inclusion is proper,
since the class of monotone \Lr-graphs, equivalent to p-box
graphs, contains all trees~\cite{S-T-p-box}, which have unbounded
thinness~\cite{B-D-thinness}. On the other hand, the result is
tight, in the sense that the octahedron $\overline{3K_2}$ is an
example of a graph of thinness~3 which is not p-box since it has
boxicity~3~\cite{Rob-box}.

The bend number of 2-thin graphs as \emph{edge} intersection
graphs of paths on a grid (EPG graphs) is unbounded, since already
proper independent 2-thin graphs contain the class of complete
bipartite graphs, that has unbounded EPG bend
number~\cite{Suk-epg}.

The results relating thinness to VPG and EPG models are summarized
in Table~\ref{tab:vpg}. Other width parameters are analyzed for
VPG and EPG graphs in~\cite{G-M-width-vpg}.


\begin{table}[t]
\small
    \centering
    \begin{tabular}{|rcll|}
    \hline
   independent 2-thin & $\subseteq$ & $B_0$-VPG & Prop~\ref{prop:b1-b0}  \\ \hline
    2-thin & $\subseteq$ & \Lr-graphs $\ \subseteq \ $ $B_1$-VPG & Prop~\ref{prop:b1-b0}  \\
    2-thin & $=$ & blocking monotone \Lr-graphs & Thm~\ref{thm:forb-pat-2-thin}  \\
    2-thin & $\not \subseteq$ & $B_0$-VPG & Prop~\ref{prop:not-b0}  \\ \hline
   independent 3-thin & $\subseteq$ & $B_1$-VPG & Prop~\ref{prop:b3-b1}  \\
   independent 3-thin & $\not \subseteq$ & $B_0$-VPG & Prop~\ref{prop:not-b0}  \\
   independent 3-thin & $\not \subseteq$ & monotone \Lr-graphs & Prop~\ref{prop:not-L} \\ \hline
   3-thin & $\subseteq$ & $B_3$-VPG & Prop~\ref{prop:b3-b1} \\ \hline
   4-thin & $\not \subseteq$ & VPG & Prop~\ref{prop:4t-not-vpg}  \\ \hline
   $B_0$-VPG & $ \not \subseteq$ & $k$-thin, $\forall k$ & Prop~\ref{prop:b0-unbound}  \\ \hline
   monotone \Lr-graphs & $ \not \subseteq$ & $k$-thin, $\forall k$ & \cite{B-D-thinness,S-T-p-box}  \\ \hline
   proper independent 2-thin & $ \not \subseteq$ & $B_k$-EPG, $\forall k$ & Prop~\ref{prop:epg}  \\ \hline
    \end{tabular}
    \caption{Summary of the results in Section~\ref{sec:vpg}. }
    \label{tab:vpg}
\end{table}

\subsection{Forbidden patterns}

Many classic graph classes, such as interval, proper interval,
chordal, comparability, co-comparability, and bipartite graphs,
can be characterized by the existence of an ordering of the
vertices avoiding some ordered subgraphs, called \emph{patterns}.
Very recently, all the classes corresponding to patterns on three
vertices (including the ones mentioned above) have been listed,
and proved to be efficiently recognizable~\cite{Habib-patterns}.
Less is known about patterns on four vertices. One of the few
graph classes characterized by a pattern on four vertices is the
class of monotone \Lr-graphs~\cite{Chap-MPT}. A recent paper
studies systematically forbidden pattern characterizations of
graph classes defined by grounded intersection
models~\cite{Habib-patterns-ground-arx}.

In the literature, there are two additional ways of defining
patterns for subclasses of bipartite graphs. In both cases, an
explicit bipartition is given. In the case of \emph{bicolored
patterns}~\cite{H-H-bigraphs}, a total order of the vertices is
defined, while in the case of \emph{bipartite
patterns}~\cite{H-M-R-orders}, each part of the bipartition is
linearly ordered.

In Section~\ref{sec:pat}, we present characterizations for 2-thin
graphs, independent 2-thin graphs and proper independent 2-thin
graphs in terms of forbidden patterns. We also characterize
independent 2-thin graphs and proper independent 2-thin graphs in
terms of forbidden bicolored patterns and forbidden bipartite
patterns. These latter characterizations lead also to equivalences
with two well known subclasses of bipartite graphs. Namely, we
proved that independent 2-thin graphs are equivalent to interval
bigraphs, and proper independent 2-thin graphs are equivalent to
bipartite permutation graphs.

\subsection{Other width parameters}

The relation between thinness and other well known width
parameters is surveyed
in~\cite{Mixed-thin-wg,B-D-thinness,B-B-M-P-convex-arxiv}.


The \emph{pathwidth} (resp. \emph{bandwidth}) of a graph $G$ can
be defined as one less than the maximum clique size of an interval
(resp. \emph{proper} interval) supergraph of $G$, chosen to
minimize its maximum clique size~\cite{K-S-bw,Moh-bw}.
In~\cite{M-O-R-C-thinness} it is proved that the thinness of a
graph is at most its pathwidth plus one but, indeed, the proof
shows that the \emph{independent} thinness of a graph is at most
its pathwidth plus one. Combining the ideas behind that proof and
a characterization of the bandwidth of a graph as a \emph{proper
pathwidth}, due to Kaplan and Shamir~\cite{K-S-bw}, it can be
proved that the proper independent thinness of a graph is at most
its bandwidth plus one. We will furthermore prove that the proper
thinness of a graph is at most its bandwidth (unless the graph is
edgeless; in that case, the bandwidth is zero but the thinness is
one).

\subsection{Outline}

In the remaining of this section we will introduce some basic
definitions. In Section~\ref{sec:algo}, we present algorithmic
results related to the recognition of (proper) (independent)
$k$-thin graphs. In Section~\ref{sec:box}, we characterize some
models for 2-thin graphs as intersection graphs of rectangles in
the plane. In Section~\ref{sec:vpg}, we relate the thinness and
the number of bends per path in representations as vertex
intersection graphs of paths on a grid (VPG graphs). In
Section~\ref{sec:pat}, we obtain forbidden pattern
characterizations for 2-thin graphs, independent 2-thin graphs,
and proper independent 2-thin graphs, and relate these classes to
well known subclasses of bipartite graphs. Finally, in
Section~\ref{sec:width}, we bound the proper thinness and proper
independent thinness of a graph is terms of its bandwidth, and the
independent thinness in terms of its pathwidth.

An extended abstract of this work was presented at LAGOS~2021 and
appears in~\cite{B-B-lagos21}.

\subsection{Basic definitions}\label{sec:def}

All graphs in this work are finite, have no loops or multiple
edges, and are undirected unless we say explicitly
\emph{digraphs}. For a graph $G$, denote by $V(G)$ its vertex set
and by $E(G)$ its edge set. For a subset $A$ of $V(G)$, denote by
$G[A]$ the subgraph of $G$ induced by $A$.

A \emph{digraph} is a graph $D = (V,A)$ such that $A$ consists of
ordered pairs of $V$, called \emph{arcs}. A \emph{directed cycle}
of a digraph $D$ is a sequence $v_1,v_2, \ldots,v_n$ of vertices
of $V(D)$ such that $v_1 = v_{n}$ and, for all $1 \leq j<n$,
$v_jv_{j+1} \in A(D)$. We may omit the word ``directed'' if it is
clear from context that the cycle is directed and not just a cycle
of the underlying graph. A digraph is \emph{acyclic} if it has no
directed cycles. A \emph{topological ordering} of a digraph $D$ is
an ordering~$<$ of its vertices such that for each arc $vw \in
A(D)$, $v < w$. A digraph admits a topological ordering if and
only if it is acyclic, and such an ordering can be computed in
$O(|V|+|A|)$ time~\cite{Knuth1}.

A \emph{clique} or \emph{complete set} (resp.\ \emph{independent
set}) is a set of pairwise adjacent (resp.\ nonadjacent) vertices.
The \emph{clique number} of a graph is the size of a maximum
clique. Let $X, Y \subseteq V(G)$. We say that $X$ is
\textit{complete to} $Y$ if every vertex in $X$ is adjacent to
every vertex in $Y$, and that $X$ is \textit{anticomplete to} $Y$
if no vertex of $X$ is adjacent to a vertex of $Y$. A graph is
\emph{complete} if its vertex set is a complete set. A graph is
\emph{bipartite} if its vertex set can be partitioned into two
independent sets, and \emph{complete bipartite} if those sets are
complete to each other.

Given a graph $G$ and two disjoint subsets $A$, $B$ of $V(G)$, the
bipartite graph $G[A,B]$ is defined as the subgraph of $G$ formed
by the vertices $A \cup B$ and the edges of $G$ that have one
endpoint in $A$ and one in $B$. Notice that $G[A,B]$ is not
necessarily an induced subgraph of $G$.

A \emph{$t$-coloring} of a graph is a partition of its vertices
into $t$ independent sets. The smallest $t$ such that a graph $G$
admits a $t$-coloring is called the \emph{chromatic number} of
$G$. A graph is \emph{perfect} if for every induced subgraph of
it, the chromatic number equals the clique number.

A graph $G(V,E)$ is a \emph{comparability graph} if there exists
an ordering $v_1, \dots , v_n$ of $V$ such that, for each triple
$(r,s,t)$ with $r<s<t$, if $v_r v_s$ and $v_s v_t$ are edges of
$G$, then so is $v_r v_t$.\ Such an ordering is a
\emph{comparability ordering}. A graph is a \emph{co-comparability
graph} if its complement is a comparability graph. A graph is a
\emph{permutation graph} if it is both a comparability and a
co-comparability graph, and a \emph{bipartite permutation graph}
if it is, moreover, bipartite.

A bipartite graph with bipartition $(X,Y)$ is an \emph{interval
bigraph} if every vertex can be assigned an interval on the real
line such that for all $x \in X$ and $y \in Y$, $x$ is adjacent to
$y$ if and only if the corresponding intervals intersect. A
\emph{proper interval bigraph} is an interval bigraph admitting a
representation in which the family of intervals of each of $X$,
$Y$ is inclusion-free. Proper interval bigraphs are equivalent to
bipartite permutation graphs~\cite{H-H-circ2}.

\section{Algorithmic aspects of thinness}\label{sec:algo}

The recognition problem for $k$-thin graphs is
NP-complete~\cite{Shitov-thin}, and the recognition problem for
proper $k$-thin graphs is open. Both problems remain open when $k
\geq 2$ is fixed. Some related algorithmic problems were also
studied: partition into a minimum number of classes (strongly)
consistent with a given vertex
ordering~\cite{B-D-thinness,B-M-O-thin-tcs}, existence of a vertex
ordering (strongly) consistent with a given vertex
partition~\cite{B-D-thinness}.
 In this work, in
order to prove the intersection model characterizations, we will
deal with the problem of the existence of a vertex ordering
(strongly) consistent with a given vertex partition and that
extends a partial order of the vertices that is a total order when
restricted to each of the parts.

The problem of finding a partition into a minimum number of
classes (strongly) consistent with a given vertex ordering can be
solved by coloring a conflict graph, that is shown to belong to a
class in which the coloring problem is polynomial-time solvable.
Namely, let $G$ be a graph and ${<}$ an ordering of its vertices.
In~\cite{B-M-O-thin-tcs}, it was defined the graph $G_{<}$ having
$V(G)$ as vertex set, and $E(G_{<})$ is such that for $v < w$, $vw
\in E(G_{<})$ if and only if there is a vertex $z$ in $G$ such
that $w < z$, $zv \in E(G)$ and $zw \not \in E(G)$. Similarly,
in~\cite{B-D-thinness}, it was introduced the graph
$\tilde{G}_{<}$, which has $V(G)$ as vertex set, and
$E(\tilde{G}_{<})$ is such that for $v < w$, $vw \in
E(\tilde{G}_{<})$ if and only if either $vw \in E(G_{<})$ or there
is a vertex $x$ in $G$ such that $x < v$, $xw \in E(G)$ and $xv
\not \in E(G)$. An edge of $G_{<}$ (respectively $\tilde{G}_{<}$)
represents that its endpoints cannot belong to the same class in a
vertex partition that is consistent (respectively strongly
consistent) with the ordering ${<}$, and, as it was observed in
the respective works, such a partition is a coloring of the
corresponding graph.

In those works it was proved that $G_{<}$ and $\tilde{G}_{<}$ are
co-comparability graphs, thus perfect~\cite{Meyn-co-comp}. This
has two main implications. The first one is that the optimum
coloring can be computed in polynomial time~\cite{Go-comp}, and
thus the problem of finding a partition into a minimum number of
classes (strongly) consistent with a given vertex ordering can be
solved in polynomial time. The other one is that the chromatic
number equals the clique number, and the following corollary was
used to prove upper and lower bounds for the thinness and proper
thinness of a graph.

\begin{corollary}\emph{\cite{BGOSS-thin-oper}}\label{cor:thin-comp-order} Let $G$ be a graph, and $k$ a positive integer.
Then $\thin(G) \geq k$ (resp. $\pthin(G) \geq k$) if and only if,
for every ordering ${<}$ of $V(G)$, the graph $G_{<}$ (resp.
$\tilde{G}_{<}$) has a clique of size $k$.
\end{corollary}

In~\cite{BGOSS-thin-oper}, it was observed that also the problem
of finding a partition into a minimum number of independent
classes (strongly) consistent with a given vertex ordering $<$ can
be solved by coloring a conflict graph. Precisely, the graph
$G_{<}^{ind}$ (resp. $\tilde{G}_{<}^{ind}$) whose vertex set is
$V(G)$ and whose edge set is $E(G) \cup E(G_{<})$ (resp. $E(G)
\cup E(\tilde{G}_{<})$). We will prove next that these graphs are
co-comparability graphs, as well.

\begin{theorem}\label{thm:ind-cocomp}
 Given a graph $G$ and a vertex ordering $<$, the conflict graphs $G_{<}^{ind}$ (resp. $\tilde{G}_{<}^{ind}$) are co-comparability graphs.
\end{theorem}

\begin{proof}
Let us see first that $<$ is a co-comparability order for
$G_{<}^{ind}$. Let $x < y < z$ in $V(G)$ and suppose that $xz \in
E(G_{<}^{ind})$. We need to prove that at least one of $xy$, $yz$
is an edge of $G_{<}^{ind}$. If $xz \in E(G)$, then either $yz \in
E(G)$ or $xy$ in $E(G_{<})$. If $xz \in E(G_{<})$, then at least
one of $xy$, $yz$ is an edge of $E(G_{<})$, because $G_{<}$ is a
co-comparability graph. In all the cases, at least one of $xy$,
$yz$ is an edge of $G_{<}^{ind}$. The proof holds exactly the same
way for $\tilde{G}_{<}^{ind}$, replacing $G_{<}$ by
$\tilde{G}_{<}$.
\end{proof}

In particular, we have the following.

\begin{corollary}\label{cor:ind-cocomp}
The problem of finding a partition into a minimum number of
independent classes (strongly) consistent with a given vertex
ordering is polynomial-time solvable.
\end{corollary}

The problem about the existence of a vertex ordering (strongly)
consistent with a given vertex partition was shown to be
NP-complete~\cite{B-D-thinness}, but the complexity remains open
when the number of parts is fixed.


Let us solve now the following problem.\\

\noindent {\textsc{(Strongly) Consistent Extending Order -- (S)CEO}}\\
\textbf{Instance:} A graph $G$, a partition $\Pi =
\{V^1,\dots,V^k\}$ and a partial order ${<}$ of $V(G)$ that is
total and (strongly) consistent restricted to each $V^j$, $1 \leq j \leq k$. \\
\textbf{Question:} Does there exist a total ordering of $V(G)$ extending $<$ and (strongly) consistent with $\Pi$?\\

Given the input of the (S)CEO problem, we define a digraph
$D(G,\Pi,<)$ (resp. $\tilde{D}(G,\Pi,<)$) having $V(G)$ as vertex
set and such that an ordering of $V(G)$ is a solution to (S)CEO if
and only if it is a topological ordering of $D(G,\Pi,<)$ (resp.
$\tilde{D}(G,\Pi,<)$). The problem then reduces to the existence
of a topological order of a digraph, which is polynomial-time
solvable~\cite{Knuth1}. Given two vertices $v_i \in V^i$, $v_j \in
V^j$, $i \neq j$, we create the arc $v_iv_j$ in $D(G,\Pi,<)$ if
and only if $v_iv_j \not \in E(G)$ and there exists $v_j' \in V^j$
with $v_j' < v_j$ and $v_iv_j' \in E(G)$, and in
$\tilde{D}(G,\Pi,<)$ if and only if $v_iv_j \not \in E(G)$ and
either there exists $v_j' \in V^j$ with $v_j' < v_j$ and $v_iv_j'
\in E(G)$, or there exists $v_i' \in V^i$ with $v_i' > v_i$ and
$v_i'v_j \in E(G)$. Additionally, in order to ensure that a
topological ordering of the digraph extends $<$, we create in both
cases the arc $vv'$ for every pair of vertices $v < v'$.

\begin{lemma}\label{lem:D-thin} Let $G$ be a graph, $\Pi = \{V^1,\dots,V^k\}$ a partition and ${<}$ a partial order of $V(G)$ that is total and consistent restricted to each $V^j$, $1 \leq j \leq k$.
An ordering of $V(G)$ is consistent with the partition $\Pi$ and
extends the partial order $<$ if and only if it is a topological
ordering of $D(G,\Pi,<)$.
\end{lemma}

\begin{proof}
Suppose first $\prec$ is a total ordering of $V(G)$ consistent
with the partition $\Pi$ and that extends the partial order $<$.
Let $vv'$ be an arc such that $v < v'$. Then $v \prec v'$, since
$\prec$ extends $<$. Let $v_iv_j$ be an arc with $v_i \in V^i$,
$v_j \in V^j$, $i \neq j$, such that $v_iv_j \not \in E(G)$ and
there exists $v_j' \in V^j$ with $v_j' < v_j$ and $v_iv_j' \in
E(G)$. Suppose that $v_j \prec v_i$. Since $\prec$ extends $<$,
$v_j' \prec v_j \prec v_i$, and $v_iv_j' \in E(G)$ but $v_iv_j
\not \in E(G)$, a contradiction because $\prec$ is consistent with
the partition $\Pi$. Thus $v_i \prec v_j$ for every such arc, and
$\prec$ is a topological ordering for $D(G,\Pi,<)$.

Conversely, suppose $\prec$ is a topological ordering for
$D(G,\Pi,<)$. Let $v_j' \prec v_j \prec v_i$ such that $v_i \in
V^i$, $v_j, v_j' \in V^j$, $v_j'v_i \in E(G)$. If $i = j$, then
$v_j' < v_j < v_i$ because $<$ is total restricted to $V^i$ and
$\prec$ extends $<$, so $v_jv_i \in E(G)$ because $<$ is
consistent restricted to $V^i$. If $i \neq j$, then $v_j' < v_j$
because $<$ is total restricted to $V^j$ and $\prec$ extends $<$.
Since $\prec$ is a topological ordering for $D(G,\Pi,<)$, $v_iv_j$
is not an arc of $D(G,\Pi,<)$, thus $v_iv_j \in E(G)$.
\end{proof}

\begin{lemma}\label{lem:D-pthin} Let $G$ be a graph, $\Pi = \{V^1,\dots,V^k\}$ a partition and ${<}$ a partial order of $V(G)$ that is total and strongly consistent restricted to each $V^j$, $1 \leq j \leq k$.
An ordering of $V(G)$ is strongly consistent with the partition
$\Pi$ and extends the partial order $<$ if and only if it is a
topological ordering of $\tilde{D}(G,\Pi,<)$.
\end{lemma}

\begin{proof}
Suppose first $\prec$ is a total ordering of $V(G)$ strongly
consistent with the partition $\Pi$ and that extends the partial
order $<$. Let $vv'$ be an arc such that $v < v'$. Then $v \prec
v'$, since $\prec$ extends $<$. Let $v_iv_j$ be an arc with $v_i
\in V^i$, $v_j \in V^j$, $i \neq j$, such that $v_iv_j \not \in
E(G)$ and there exists $v_j' \in V^j$ with $v_j' < v_j$ and
$v_iv_j' \in E(G)$. Suppose that $v_j \prec v_i$. Since $\prec$
extends $<$, $v_j' \prec v_j \prec v_i$, and $v_iv_j' \in E(G)$
but $v_iv_j \not \in E(G)$, a contradiction because $\prec$ is
strongly consistent with the partition $\Pi$. Thus $v_i \prec v_j$
for every such arc. Analogously, let $v_iv_j$ be an arc with $v_i
\in V^i$, $v_j \in V^j$, $i \neq j$, such that $v_iv_j \not \in
E(G)$ and there exists $v_i' \in V^i$ with $v_i' > v_i$ and
$v_i'v_j \in E(G)$. Suppose that $v_j \prec v_i$. Since $\prec$
extends $<$, $v_j \prec v_i \prec v_i'$, and $v_i'v_j \in E(G)$
but $v_iv_j \not \in E(G)$, a contradiction because $\prec$ is
strongly consistent with the partition $\Pi$. Thus $v_i \prec v_j$
for every such arc. Therefore, $\prec$ is a topological ordering
for $\tilde{D}(G,\Pi,<)$.

Conversely, suppose $\prec$ is a topological ordering for
$\tilde{D}(G,\Pi,<)$. Let $v_i \prec v_i' \prec v_i''$ in $V^i$,
such that $v_iv_i'' \in E(G)$. Then $v_i < v_i' < v_i''$ because
$<$ is total restricted to $V^i$ and $\prec$ extends $<$, so
$v_iv_i', v_i'v_i'' \in E(G)$ because $<$ is strongly consistent
restricted to $V^i$. Let $v_j' \prec v_j \prec v_i$ such that $v_i
\in V^i$, $v_j, v_j' \in V^j$, $i \neq j$, and $v_j'v_i \in E(G)$.
Then $v_j' < v_j$ because $<$ is total restricted to $V^j$ and
$\prec$ extends $<$. Since $\prec$ is a topological ordering for
$\tilde{D}(G,\Pi,<)$, $v_iv_j$ is not an arc of
$\tilde{D}(G,\Pi,<)$, thus $v_iv_j \in E(G)$. Analogously, let
$v_j \prec v_i \prec v_i'$ such that $v_i,v_i' \in V^i$, $v_j \in
V^j$, $i \neq j$, and $v_jv_i' \in E(G)$. Then $v_i < v_i'$
because $<$ is total restricted to $V^i$ and $\prec$ extends $<$.
Since $\prec$ is a topological ordering for $\tilde{D}(G,\Pi,<)$,
$v_iv_j$ is not an arc of $\tilde{D}(G,\Pi,<)$, thus $v_iv_j \in
E(G)$.
\end{proof}

\begin{corollary}\label{cor:sceo}
\textsc{Consistent Extending Order} and \textsc{Strongly
Consistent Extending Order} are polynomial-time solvable.
\end{corollary}


\section{Rectangle intersection models for 2-thin and proper 2-thin
graphs}\label{sec:box}

We present in this section a rectangle intersection model for
2-thin graphs which is a slight modification of the model
in~\cite{C-M-O-thinness-man}. The common property is that the
upper-right corners of the boxes lie in two diagonals within the
second and fourth quadrants of the Cartesian plane. Indeed, the
upper-right corners of the boxes are the same as in their model.
Roughly speaking, the difference is that in their model, the boxes
in the upper diagonal ``go down'' enough to intersect the boxes
corresponding to all the neighbors in the lower diagonal, and the
boxes in the lower diagonal ``go left'' enough to intersect the
boxes corresponding to all the neighbors in the upper diagonal. In
our model, the boxes in the upper diagonal ``go down'' and stop
just before intersecting a box corresponding to a non-neighbor in
the lower diagonal, and the boxes in the lower diagonal ``go
left'' and stop just before intersecting a box corresponding to a
non-neighbor in the upper diagonal. This difference produces a
model that satisfies a property that we call \emph{blocking}.

One of the main results of this section is that, while some graphs
with thinness~3 admit a rectangle intersection model such that the
upper-right corners of the boxes lie in two diagonals within the
second and fourth quadrants of the Cartesian plane, when we add
the blocking property as a requirement, then every graph admitting
such a model is 2-thin.

\begin{figure}[t]
\noindent\begin{minipage}{\textwidth}
\begin{minipage}[t]{\dimexpr0.25\textwidth-0.5\Colsep\relax}
\begin{center}
\begin{tikzpicture}[scale=0.28]

\draw[fill=gray,draw=none,fill opacity=0.3] (-20,5) rectangle
(-10,7);

\draw[fill=gray,draw=none,fill opacity=0.3] (-16,-1) rectangle
(-12,9);

\draw[fill=white] (-16,5) [c1] rectangle (-12,7);

\node() at (-12, 7)  [n2] {};

\node[label=above:{\footnotesize $b$}] at (-14,4.6) {};
\node[label=below:{\footnotesize $P_Y(b)$}] at (-14,5) {};
  \node[label=right:{\footnotesize $P_X(b)$}] at (-20.5,6.1) {};

\end{tikzpicture}
    \end{center}
\end{minipage}\hfill
\begin{minipage}[t]{\dimexpr0.25\textwidth-0.5\Colsep\relax}
    \begin{center}
\begin{tikzpicture}[scale=0.28]


\draw (14,2) [c1] rectangle  (16,3); \draw (15,1) [c1] rectangle
(17,4); \draw (16,3.5) [c1] rectangle (18,5); \draw (16.5,3.5)
[c1] rectangle (19,6); \draw (14.5,6) [c1] rectangle (20,7); \draw
(18.8,2.5) [c1] rectangle (21,8);

\draw (16.5,0.5) [c1] rectangle (20,2); \draw (19.5,1.2) [c1]
rectangle (21,3); \draw (21.3,3.3) [c1] rectangle (22,4);

\node() at (16, 3)  [n2] {}; \node() at (17, 4)  [n2] {}; \node()
at (18, 5)  [n2] {}; \node() at (19, 6)  [n2] {}; \node() at (20,
7)  [n2] {}; \node() at (21, 8)  [n2] {};

\node() at (20, 2)  [n2] {}; \node() at (21, 3)  [n2] {}; \node()
at (22, 4)  [n2] {};

\end{tikzpicture}
    \end{center}
\end{minipage}\hfill
\begin{minipage}[t]{\dimexpr0.25\textwidth-0.5\Colsep\relax}
    \begin{center}
\begin{tikzpicture}[scale=0.28]


\draw[color=gray] (0,4) -- (10,4); \draw[color=gray] (5,-1) --
(5,9);

\draw (0,0.3) [c1] rectangle  (1,5); \draw (0.5,0) [c1] rectangle
(2,6); \draw (2.5,3.2) [c1] rectangle (3,7); \draw (0.3,1.8) [c1]
rectangle (4,8);

\draw (2.5,0) [c1] rectangle (6,1); \draw (2.2,0.5) [c1] rectangle
(7,2); \draw (0,1.2) [c1] rectangle (8,3);

\node() at (1, 5)  [n2] {}; \node() at (2, 6)  [n2] {}; \node() at
(3, 7)  [n2] {}; \node() at (4, 8)  [n2] {};

\node() at (6, 1)  [n2] {}; \node() at (7, 2)  [n2] {}; \node() at
(8, 3)  [n2] {};

\end{tikzpicture}
    \end{center}
\end{minipage}\hfill
\begin{minipage}[t]{\dimexpr0.25\textwidth-0.5\Colsep\relax}
    \begin{center}
\begin{tikzpicture}[scale=0.28]


\draw[color=gray] (12,4) -- (22,4); \draw[color=gray] (17,-1) --
(17,9);

\draw (12,0.3) [c1] rectangle  (13,5); \draw (12.5,0) [c1]
rectangle  (14,6); \draw (14.5,3.2) [c1,fill=gray, fill
opacity=0.5] rectangle (15,7); \draw (12.3,1.8) [c1] rectangle
(16,8);

\draw (15.5,0) [c1,fill=gray, fill opacity=0.5] rectangle (18,1);
\draw (14.2,0.5) [c1] rectangle (19,2); \draw (12,1.2) [c1]
rectangle (20,3);

\node() at (13, 5)  [n2] {}; \node() at (14, 6) [n2] {}; \node()
at (15, 7) [n2] {}; \node() at (16, 8) [n2] {};

\node() at (18, 1)  [n2] {}; \node() at (19, 2) [n2] {}; \node()
at (20, 3) [n2] {};

\end{tikzpicture}
\end{center}
\end{minipage}%
\end{minipage}
\caption{Upper-right corner, vertical and horizontal prolongations
(first), a weakly 2-diagonal model (second), a blocking 2-diagonal
model (third), a 2-diagonal model that is not blocking, where the
gray boxes do not satisfy the required property
(fourth).}\label{fig:blocking}
\end{figure}
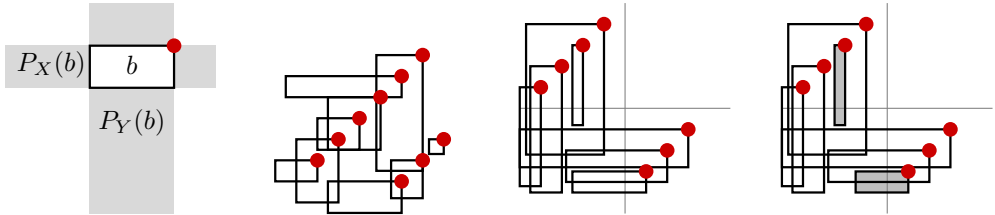

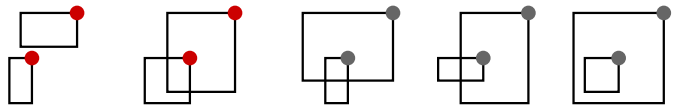
\begin{figure}[t]
\begin{center}
\begin{tikzpicture}[scale=0.3]

\draw (-7,0) [c1] rectangle  (-6,2); \draw (-6.5,2.5) [c1]
rectangle (-4,4);

\node() at (-6, 2)  [n2] {}; \node() at (-4, 4)  [n2] {};

\draw (-1,0) [c1] rectangle  (1,2); \draw (0,0.5) [c1] rectangle
(3,4);

\node() at (1, 2)  [n2] {}; \node() at (3, 4)  [n2] {};

\draw (7,0) [c1] rectangle (8,2); \draw (6,1) [c1] rectangle
(10,4);

\node() at (8, 2)  [n2,fill=colornododiag3] {}; \node() at (10, 4)
[n2,fill=colornododiag3] {};

\draw (12,1) [c1] rectangle (14,2); \draw (13,0) [c1] rectangle
(16,4);

\node() at (14, 2)  [n2,fill=colornododiag3] {}; \node() at (16,
4)  [n2,fill=colornododiag3] {};

\draw (18.5,0.5) [c1] rectangle (20,2); \draw (18,0) [c1]
rectangle (22,4);

\node() at (20, 2)  [n2,fill=colornododiag3] {}; \node() at (22,
4)  [n2,fill=colornododiag3] {};

\end{tikzpicture}
\end{center}

\caption{The first two situations are bi-semi-proper, the last
three are not.}\label{fig:semi-prop}
\end{figure}

\subsection{Formal definitions of (weakly) 2-diagonal, blocking, and bi-semi-proper}

We call \emph{box} a rectangle that is aligned with the Cartesian
axes in $\mathbb{R}^2$, i.e., the Cartesian product of two
segments $[x_1,x_2] \times [y_1,y_2]$. We say that the box $b$
\emph{is defined by} $x_1,x_2,y_1,y_2$, which will be denoted by
$X_1(b), X_2(b), Y_1(b)$, and $Y_2(b)$, respectively. The
\emph{upper-right corner} of $b$ is the point $(x_2,y_2)$. The
\emph{vertical} (resp. \emph{horizontal}) \emph{prolongation} is
the Cartesian product $P_Y(b) = [x_1,x_2] \times \mathbb{R}$
(resp. $P_X(b) = \mathbb{R} \times [y_1,y_2]$).

A set of boxes is \emph{2-diagonal} if their upper-right corners
are pairwise distinct and each of them lies, for some constant
values $d_1 < 0 < d_2$, either in the intersection of the diagonal
$y = x+d_1$ and the 4th quadrant of the Cartesian plane, or in the
intersection of the diagonal $y= x+d_2$ and the 2nd quadrant of
the Cartesian plane.

A set of boxes is \emph{weakly 2-diagonal} if their upper-right
corners are pairwise distinct and each of them lies, for some
constant values $d_1 < d_2$, either in the diagonal $y = x+d_1$ or
in the diagonal $y= x+d_2$. We will call $y = x+d_1$ the
\emph{lower diagonal} and $y= x+d_2$ the \emph{upper diagonal}.

A 2-diagonal model is \emph{blocking} if for every two
non-intersecting boxes $b_1$, $b_2$ in the upper and lower
diagonal, resp., either the vertical prolongation of $b_1$
intersects $b_2$ or the horizontal prolongation of $b_2$
intersects $b_1$ (see Figure~\ref{fig:blocking}).
%
%

A weakly 2-diagonal model is \emph{bi-semi-proper} if for any two
boxes $b$, $b'$, defined by $x_1,x_2,y_1,y_2$ and
$x_1',x_2',y_1',y_2'$ and such that $y_2 - x_2 = y_2' - x_2'$ and
$x_2 < x_2'$, it holds $x_1 \leq x_1'$ and $y_1 \leq y_1'$ (see
Figure~\ref{fig:semi-prop}).

\subsection{Definition and properties of the model $\mathcal{M}_1(G)$}

Let $G=(V,E)$ be a 2-thin graph, with partition $V^1, V^2$
consistent with an order $<$. Let $V^1 = v_1 < \dots < v_{n_1}$,
$V^2 = w_1 < \dots < w_{n_2}$.
Let $U(1,v_i) = i$ if $v_i$ has no neighbors smaller than $v_i$ in
$V^1$, or $\min \{j: v_j < v_i, v_jv_i \in E(G)\}$, otherwise. Let
$U(2,w_i) = i$ if $w_i$ has no neighbors smaller than $w_i$ in
$V^2$, or $\min \{j: w_j < w_i, w_jw_i \in E(G)\}$, otherwise.
Let $U(2,v_i) = 0$ if $v_i$ is adjacent to all the vertices of
$V^2$ which are smaller than $v_i$, or $\max \{j: w_j < v_i,
w_jv_i \not \in E(G)\}$, otherwise. Let $U(1,w_i) = 0$ if $w_i$ is
adjacent to all the vertices of $V^1$ which are smaller than
$w_i$, or $\max \{j: v_j < w_i, v_jw_i \not \in E(G)\}$,
otherwise.

We define the following model of $G$ as intersection of boxes in
the plane, which is a 2-diagonal model centered at $(n_2,n_1)$,
and that we will denote by $\mathcal{M}_1(G)$: the upper-right
corner of $v_i$ is $(i+n_2,i)$, for $1 \leq i \leq n_1$, and the
upper-right corner of $w_i$ is $(i,i+n_1)$, for $1 \leq i \leq
n_2$; the lower-left corner of $v_i$ is
$(U(2,v_i)+0.5,U(1,v_i)-0.5)$, for $1 \leq i \leq n_1$, and the
lower-left corner of $w_i$ is $(U(2,w_i)-0.5,U(1,w_i)+0.5)$, for
$1 \leq i \leq n_2$.
Intuitively, the boxes having the upper right corner in the higher
(resp. lower) diagonal ``go down'' (resp. left) and stop just to
avoid the greatest non-neighbor smaller than themselves in the
other class (if any), and ``go left'' (resp. down) enough to catch
all the neighbors smaller than themselves in their own class (if
any, and without intersecting a non-neighbor). Examples are
depicted in Figures~\ref{fig:M1} and~\ref{fig:M1-prop}.

\begin{figure}
\begin{center}
\begin{tikzpicture}[scale=0.35]


\node(v1) at (-14, 1.4)  [n2] {}; \node(v2) at (-14, 2.8)  [n2]
{}; \node(v3) at (-14, 5.6)  [n2] {}; \node(v4) at (-14, 7)  [n2]
{}; \node(v5) at (-14, 8.4)  [n2] {}; \node(v6) at (-14, 9.8)
[n2] {}; \node(v7) at (-14, 11.2)  [n2] {}; \node(v8) at (-14,
12.6)  [n2] {}; \node(v9) at (-14, 14)  [n2] {};

\node(w1) at (-7, 3.5)  [n2] {}; \node(w2) at (-7, 4.9)  [n2] {};

\node(w3) at (-7, 9.1)  [n2] {}; \node(w4) at (-7, 10.5)  [n2] {};
\node(w5) at (-7, 11.9)  [n2] {}; \node(w6) at (-7, 13.3)  [n2]
{};

\path (v4) edge [e1,bend right=0] (v3); \path (v4) edge [e1,bend
right=45] (v2); \path (v4) edge [e1,bend right=45] (v1);

\path (v8) edge [e1,bend right=0] (v7); \path (v8) edge [e1,bend
right=45] (v6); \path (v8) edge [e1,bend right=45] (v5);

\path (w6) edge [e1,bend left=0] (w5); \path (w6) edge [e1,bend
left=45] (w4); \path (w6) edge [e1,bend left=45] (w3);

\path (w1) edge [e1,bend left=0] (v2); \path (w1) edge [e1,bend
left=0] (v3); \path (w2) edge [e1,bend left=0] (v2); \path (w2)
edge [e1,bend left=0] (v3);

\path (v9) edge [e1,bend left=0] (w3); \path (v9) edge [e1,bend
left=0] (w4); \path (v9) edge [e1,bend left=0] (w5); \path (v9)
edge [e1,bend left=0] (w6);

\path (v9) edge [e1,bend left=0] (v8); \path (v9) edge [e1,bend
left=45] (v7); \path (v9) edge [e1,bend left=45] (v6); \path (v9)
edge [e1,bend left=45] (v5);

\path (v9) edge [e1,bend left=45] (v4); \path (v9) edge [e1,bend
left=45] (v3); \path (v9) edge [e1,bend left=45] (v2); \path (v9)
edge [e1,bend left=45] (v1);


\draw[color=gray] (0,9.5) -- (16,9.5);

\draw[color=gray] (6.5,0) -- (6.5,16);

\foreach \i / \x / \y in
{1/1/0,2/2/0.1,3/3/0.1,4/1.1/2,5/5/2.1,6/6/3,7/7/4,8/5.1/5,9/1.2/2.2}
    \draw (\y+0.5,\x-0.5) [c2] rectangle  (\i+6, \i);

\foreach \i / \x / \y  in {1/0.9/1, 2/0.9/2, 3/5/3, 4/6/4, 5/7/5,
6/8/3.1}
    \draw (\y-0.5,\x+0.5) [c2] rectangle  (\i, \i+9);

\foreach \i in {1,...,9}
    \node() at (\i+6, \i)  [n2] {};

\foreach \i in {1,...,6}
    \node() at (\i,\i+9) [n2] {};

\end{tikzpicture}
\end{center}

\caption{The model $\mathcal{M}_1$ for the 2-thin graph on the
left, whose representation is not proper (indeed, its proper
thinness is~3). In the graph, the vertices are ordered
increasingly by their $y$-coordinate, and the classes correspond
to the vertical lines.}\label{fig:M1}
\end{figure}
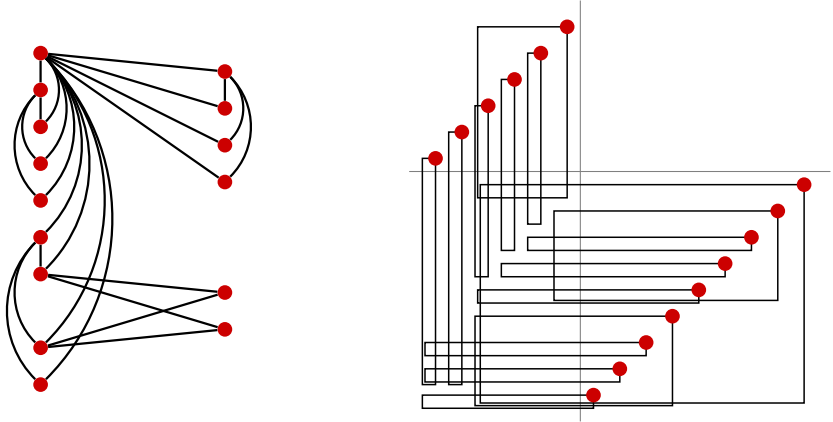

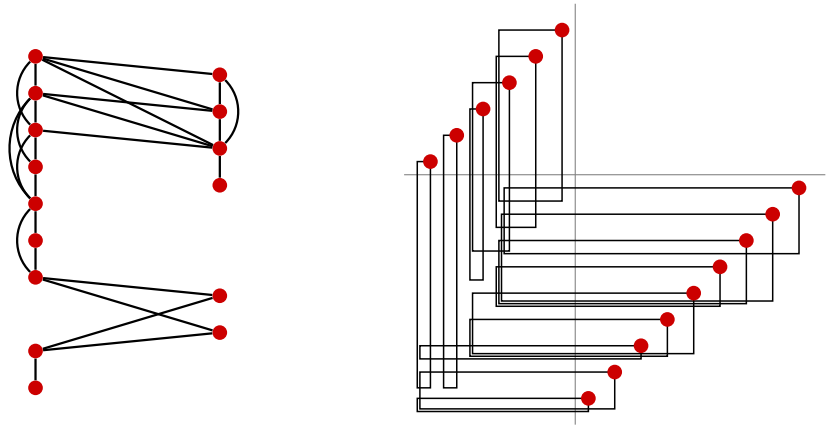
\begin{figure}
\begin{center}
\begin{tikzpicture}[scale=0.35]


\node(x1) at (4, 1.4)  [n2] {}; \node(x2) at (4, 2.8)  [n2] {};
\node(x3) at (4, 5.6)  [n2] {}; \node(x4) at (4, 7)  [n2] {};
\node(x5) at (4, 8.4)  [n2] {}; \node(x6) at (4, 9.8)  [n2] {};
\node(x7) at (4, 11.2)  [n2] {}; \node(x8) at (4, 12.6)  [n2] {};
\node(x9) at (4, 14)  [n2] {};

\node(z1) at (11, 3.5)  [n2] {}; \node(z2) at (11, 4.9)  [n2] {};

\node(z3) at (11, 9.1)  [n2] {}; \node(z4) at (11, 10.5)  [n2] {};
\node(z5) at (11, 11.9)  [n2] {}; \node(z6) at (11, 13.3)  [n2]
{};

\path (x4) edge [e1,bend right=0] (x3); \path (x2) edge [e1,bend
right=0] (x1); \path (x4) edge [e1,bend right=0] (x5); \path (x5)
edge [e1,bend right=45] (x3);

\path (x8) edge [e1,bend right=0] (x7); \path (x8) edge [e1,bend
right=45] (x6); \path (x8) edge [e1,bend right=45] (x5); \path
(x7) edge [e1,bend right=0] (x6); \path (x7) edge [e1,bend
right=45] (x5); \path (x6) edge [e1,bend right=0] (x5);

\path (z6) edge [e1,bend left=0] (z5); \path (z6) edge [e1,bend
left=45] (z4); \path (z5) edge [e1,bend left=0] (z4); \path (z3)
edge [e1,bend left=0] (z4);

\path (z1) edge [e1,bend left=0] (x2); \path (z1) edge [e1,bend
left=0] (x3); \path (z2) edge [e1,bend left=0] (x2); \path (z2)
edge [e1,bend left=0] (x3);

\path (x8) edge [e1,bend left=0] (z4); \path (x7) edge [e1,bend
left=0] (z4);

\path (x8) edge [e1,bend left=0] (z5); \path (x9) edge [e1,bend
left=0] (z4); \path (x9) edge [e1,bend left=0] (z5); \path (x9)
edge [e1,bend left=0] (z6);

\path (x9) edge [e1,bend left=0] (x8); \path (x9) edge [e1,bend
right=45] (x7);



\draw[color=gray] (18+0,9.5) -- (18+16,9.5);

\draw[color=gray] (18+6.5,0) -- (18+6.5,16);

\foreach \i / \x / \y in
{1/1/0,2/1.1/0.1,3/3/0.1,4/3.1/2,5/3.2/2.1,6/5/3,7/5.1/3.1,8/5.2/3.2,9/7/3.3}
    \draw (18+\y+0.5,\x-0.5) [c2] rectangle  (18+\i+6, \i);

\foreach \i / \x / \y in {1/0.9/1, 2/0.9/2, 3/5/3, 4/6.1/3.1,
5/7/4, 6/8/4.1}
    \draw (18+\y-0.5,\x+0.5) [c2] rectangle  (18+\i, \i+9);

\foreach \i in {1,...,9}
    \node() at (18+\i+6, \i)  [n2] {};

\foreach \i in {1,...,6}
    \node() at (18+\i,\i+9) [n2] {};

\end{tikzpicture}
\end{center}

\caption{The model $\mathcal{M}_1$ for the proper 2-thin graph on
the left. In the graph, the vertices are ordered increasingly by
their $y$-coordinate, and the classes correspond to the vertical
lines. }\label{fig:M1-prop}
\end{figure}

\begin{lemma}\label{lem:model1}
Let $G=(V,E)$ be a 2-thin graph, with partition $V^1, V^2$
consistent with an order $<$. Then $\mathcal{M}_1(G)$ is a
blocking 2-diagonal intersection model for $G$ that respects the
relative order on each class. Moreover, if the order and the
partition are strongly consistent, the model is bi-semi-proper.
\end{lemma}

\begin{proof} It is straightforward that the model is 2-diagonal, centered at $(n_2,n_1)$,
and respects the relative order on each class. Let us prove that
 $\mathcal{M}_1(G)$ is an intersection model for $G$ and that it is blocking.

\begin{itemize}

\item Let $v_i < v_j$ adjacent. Then $U(1,v_j) \leq i$ and
$U(2,v_j) \leq n_2$, so the boxes of $v_i$ and $v_j$ intersect.

\item Let $v_i < v_j$ not adjacent. Because of the consistency
between order and partition, $U(1,v_j)
> i$, thus the boxes of $v_i$ and $v_j$ do not intersect.

\item Let $w_i < w_j$ adjacent. Then $U(2,w_j) \leq i$ and
$U(1,w_j) \leq n_1$, so the boxes of $w_i$ and $w_j$ intersect.

\item Let $w_i < w_j$ not adjacent. Because of the consistency
between order and partition,  $U(2,w_j) > i$, thus the boxes of
$w_i$ and $w_j$ do not intersect.

\item Let $v_i < w_j$. Then $U(2,v_i) < j$. If they are adjacent,
by the consistency between order and partition, $U(1,w_j) < i$ and
the boxes of $v_i$ and $w_j$ intersect. If they are not adjacent,
$U(1,w_j) \geq i$ and the boxes of $v_i$ and $w_j$ do not
intersect. In this case, $P_Y(w_j) \cap v_i \neq \emptyset$.

\item Let $w_i < v_j$. Then $U(1,w_i) < j$. If they are adjacent,
by the consistency between order and partition, $U(2,v_j) < i$ and
the boxes of $v_i$ and $w_j$ intersect. If they are not adjacent,
$U(2,v_j) \geq i$ and the boxes of $v_i$ and $w_j$ do not
intersect. In this case, $P_X(v_j) \cap w_i \neq \emptyset$.
\end{itemize}

It remains to observe that if the order and the partition are
strongly consistent and $x_i < x_j$ are in the same class, then
$U(1,x_i) \leq U(1,x_j)$ and $U(2,x_i) \leq U(2,x_j)$, so the
model is bi-semi-proper.
\end{proof}

\subsection{Characterization of 2-thin graphs as rectangle intersection graphs}

\begin{theorem}\label{thm:char-2-thin} A graph is 2-thin if and only if it has a blocking 2-diagonal model.
Moreover, if a graph $G$ is 2-thin and the partition $V^1,V^2$ of
its vertices is consistent with an order $<$, then there exists a
blocking 2-diagonal model such that on each of the diagonals lie,
respectively, the upper-right corners of the vertices of $V^1$ and
$V^2$, in such a way that their order corresponds to $<$
restricted to the respective part. Conversely, if a graph $G$
admits a blocking 2-diagonal model, then there exists an order of
the vertices of $G$ that is consistent with the partition given by
the diagonals where the upper-right corners lie, and extends their
order on the respective diagonals.
\end{theorem}

\begin{proof}

$\Rightarrow$) It follows from Lemma~\ref{lem:model1}.\\

\noindent $\Leftarrow$) Let us consider a blocking 2-diagonal
model of $G$,
 and let $V^1$ and $V^2$ be the vertices corresponding to boxes
 whose upper-right corners lie in the lower and upper diagonal, respectively.
We will slightly abuse notation and use it indistinctly for a
vertex and the box representing it.

Let $\Pi = \{V^1,V^2\}$, $<$ be the order of $V^1$ and $V^2$
defined by the $X_2$ coordinates on each of the sets, and where a
vertex of $V^1$ and a vertex of $V^2$ are not comparable.

Let us first prove that $<$ is consistent restricted to $V^i$, $i
= 1,2$. Let $x < y < z$ in $V^1$ with $xz \in E(G)$ (the
definitions are symmetric with respect to both classes). Then
$X_2(x) < X_2(y) < X_2(z)$ and since $xz \in E(G)$, it holds
$X_1(z) < X_2(x) < X_2(y)$ and $Y_1(z) < Y_2(x) = X_2(x) + d_1 <
X_2(y) + d_1 = Y_2(y)$. Therefore, $yz \in E(G)$.

Let $D = D(G,\Pi,<)$. By the blocking property, given two vertices
$v_i \in V^i$, $v_{3-i} \in V^{3-i}$, if $v_iv_{3-i} \in A(D)$,
then the appropriate prolongation of $v_{3-i}$ intersects $v_i$.
As observed above, an ordering of $V(G)$ is consistent with the
partition $V^1,V^2$ and extends the partial order $<$ if and only
if it is a topological ordering of $D$.

Let us prove now that $D$ is acyclic, thus it admits a topological
ordering. Suppose it is not, and let us consider a shortest
directed cycle of $D$. Moreover, since the subdigraph induced by
each class is complete and acyclic, the cycle has at most two
vertices of each class, and necessarily an arc from $V^1$ to $V^2$
and another from $V^2$ to $V^1$.

\noindent \textbf{{Case 1:}} The cycle consists of two vertices,
$v_1 \in V^1$ and $v_2 \in V^2$.

In this case, $v_1$ and $v_2$ do not intersect but the horizontal
prolongation of $v_1$ intersects $v_2$ and the vertical
prolongation of $v_2$ intersects $v_1$, which is not possible.

\noindent \textbf{{Case 2:}} The cycle is $v_1 w_1 v_2$ such that
$v_1, w_1 \in V^1$ and $v_2 \in V^2$.

Since $v_1 w_1 \in D$, we have $X_2(v_1) < X_2(w_1)$ and therefore
$Y_2(v_1) < Y_2(w_1)$. The horizontal prolongation of $v_1$
intersects $v_2$, therefore $Y_1(v_2) < Y_2(v_1) < Y_2(w_1)$, and
the vertical prolongation of $v_2$ intersects $w_1$, therefore
$X_1(w_1) < X_2(v_2)$, contradicting that $v_2$ and $w_1$ do not
intersect because they are not adjacent.

\noindent \textbf{{Case 3:}} The cycle is $v_2 w_2 v_1$ such that
$v_1 \in V^1$ and $v_2, w_2 \in V^2$.

Since $v_2 w_2 \in D$, we have $X_2(v_2) < X_2(w_2)$. The vertical
prolongation of $v_2$ intersects $v_1$, therefore $X_1(v_1) <
X_2(v_2) < X_2(w_2)$, and the horizontal prolongation of $v_1$
intersects $w_2$, therefore $Y_1(w_2) < Y_2(v_1)$, contradicting
that $v_1$ and $w_2$ do not intersect because they are not
adjacent.

\noindent \textbf{{Case 4:}} The cycle is $v_1 w_1 v_2 w_2$ such
that $v_1, w_1 \in V^1$ and $v_2, w_2 \in V^2$.

Since $v_i w_i \in D$, for $i = 1,2$, we have $X_2(v_i) <
X_2(w_i)$ and therefore $Y_2(v_i) < Y_2(w_i)$.

The vertical prolongation of $v_2$ intersects $w_1$ and $v_2$ does
not intersect $w_1$, therefore $Y_2(v_1) < Y_2(w_1) < Y_1(v_2)$.
The horizontal prolongation of $v_1$ intersects $w_2$ and $v_1$
does not intersect $w_2$, therefore $X_2(v_2) < X_2(w_2) <
X_1(v_1)$. This contradicts for $v_1$ and $v_2$ the fact that the
model is blocking.
\end{proof}

\subsection{Necessity of the blocking property}

Propositions~\ref{prop:rep} and~\ref{prop:th3} show that the
blocking property is necessary for Theorem~\ref{thm:char-2-thin},
since there are graphs having a 2-diagonal model which are not
2-thin.

\begin{definition}\label{def-G-star}
Let $G^*$ be the graph defined in the following way: $V(G^*) = A
\cup B$, where $A = a_1, \dots, a_{36}$, $B = b_1, \dots, b_{36}$,
and $A = A_0 \cup A_1 \cup \dots \cup A_5$, $B = B_0 \cup B_1 \cup
\dots \cup B_5$, where $A_5 = \{a_{33},a_{34},a_{35},a_{36}\}$,
$A_0 = \{a_i : 1 \leq i \leq 32, i \mbox{ is odd}\}$, and, for $1
\leq k \leq 4$, $A_{k} = \{a_i : 8(k-1) < i \leq 8k \mbox{ and } i
\mbox{ is even}\}$; $B_j = \{b_i : a_i \in A_j\}$, for $0 \leq j
\leq 5$. The edges joining $A$ and $B$ are such that: for $1 \leq
j \leq 4$, $A_j$ is complete to $B_j$ and $B_5$, and anticomplete
to $B_0$ and $B_i$, $1 \leq i \leq 4$, $i \neq j$; $A_5$ is
anticomplete to $B_0$ and complete to $B \setminus B_0$. Besides,
$a_{2k-1}a_{2k} \in E(G^*)$ and $b_{2k-1}b_{2k} \in E(G^*)$, for
$1 \leq k \leq 16$, and these are the only internal edges of $A$
and $B$.
\end{definition}

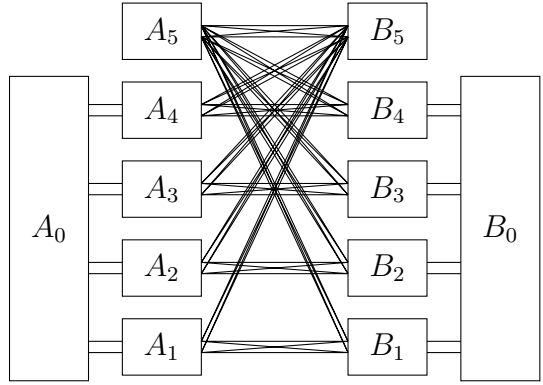
\begin{figure}
    \centering
    \begin{tikzpicture}[scale=1.5]

\foreach \x in {-1.05,-0.35,0.35,1.05,1.75} {

\draw (1-0.35,\x-0.25) -- (1+0.35,\x-0.25) -- (1+0.35,\x+0.25) --
(1-0.35,\x+0.25) -- cycle;

\draw (3-0.35,\x-0.25) -- (3+0.35,\x-0.25) -- (3+0.35,\x+0.25) --
(3-0.35,\x+0.25) -- cycle; }

\node (a1)  at (1,-1.05){$A_1$}; \node (a2)  at (1,-0.35){$A_2$};
\node (a3)  at (1,0.35){$A_3$}; \node (a4)  at (1,1.05) {$A_4$};
\node (a5)  at (1,1.75) {$A_5$};

\node (a0) at (0,0) {$A_0$}; \draw (-0.35,-1.35) -- (0.35,-1.35)
-- (0.35,1.35) -- (-0.35,1.35) -- cycle;

\node (b1)  at (3,-1.05){$B_1$}; \node (b2)  at (3,-0.35){$B_2$};
\node (b3)  at (3,0.35){$B_3$}; \node (b4)  at (3,1.05) {$B_4$};
\node (b5)  at (3,1.75) {$B_5$};

\node (b0) at (4,0) {$B_0$}; \draw (4-0.35,-1.35) -- (4.35,-1.35)
-- (4.35,1.35) -- (4-0.35,1.35) -- cycle;

\draw (0.35,-1) -- (1-0.35,-1); \draw (0.35,-1.1) --
(1-0.35,-1.1);

\draw (0.35,-0.3) -- (1-0.35,-0.3); \draw (0.35,-0.4) --
(1-0.35,-0.4);

\draw (0.35,0.3) -- (1-0.35,0.3); \draw (0.35,0.4) --
(1-0.35,0.4);

\draw (0.35,1) -- (1-0.35,1); \draw (0.35,1.1) -- (1-0.35,1.1);

\draw (4-0.35,-1) -- (3+0.35,-1); \draw (4-0.35,-1.1) --
(3+0.35,-1.1);

\draw (4-0.35,-0.3) -- (3+0.35,-0.3); \draw (4-0.35,-0.4) --
(3+0.35,-0.4);

\draw (4-0.35,0.3) -- (3+0.35,0.3); \draw (4-0.35,0.4) --
(3+0.35,0.4);

\draw (4-0.35,1) -- (3+0.35,1); \draw (4-0.35,1.1) --
(3+0.35,1.1);

\draw (1+0.35,-1) -- (3-0.35,-1) --
 (1+0.35,-1.1) -- (3-0.35,-1.1) -- cycle;

 \draw (1+0.35,-0.3) -- (3-0.35,-0.3) --
 (1+0.35,-0.4) -- (3-0.35,-0.4) -- cycle;

\draw (1+0.35,1) -- (3-0.35,1) --
 (1+0.35,1.1) -- (3-0.35,1.1) -- cycle;

 \draw (1+0.35,0.3) -- (3-0.35,0.3) --
 (1+0.35,0.4) -- (3-0.35,0.4) -- cycle;


  \draw (1+0.35,1.7) -- (3-0.35,1.7) --
 (1+0.35,1.8) -- (3-0.35,1.8) -- cycle;

\draw (1+0.35,1.7) -- (3-0.35,-1) --
 (1+0.35,1.8) -- (3-0.35,-1.1) -- cycle;

 \draw (1+0.35,1.7) -- (3-0.35,-0.3) --
 (1+0.35,1.8) -- (3-0.35,-0.4) -- cycle;

\draw (1+0.35,1.7) -- (3-0.35,1) --
 (1+0.35,1.8) -- (3-0.35,1.1) -- cycle;

 \draw (1+0.35,1.7) -- (3-0.35,0.3) --
 (1+0.35,1.8) -- (3-0.35,0.4) -- cycle;

\draw (1+0.35,-1) -- (3-0.35,1.7) --
 (1+0.35,-1.1) -- (3-0.35,1.8) -- cycle;

 \draw (1+0.35,-0.3) -- (3-0.35,1.7) --
 (1+0.35,-0.4) -- (3-0.35,1.8) -- cycle;

\draw (1+0.35,1) -- (3-0.35,1.7) --
 (1+0.35,1.1) -- (3-0.35,1.8) -- cycle;

 \draw (1+0.35,0.3) -- (3-0.35,1.7) --
 (1+0.35,0.4) -- (3-0.35,1.8) -- cycle;

    \end{tikzpicture}
    \caption{Sketch of the graph $G^*$ in Definition~\ref{def-G-star}. }
    \label{fig:G-star}
\end{figure}

\begin{proposition}\label{prop:rep}
The graph $G^*$ from Definition~\ref{def-G-star} has a
representation as intersection of boxes having the 2-diagonal
property.
\end{proposition}

\begin{proof}
The representation is as follows: the upper right corner of $a_i$
is $(i,i+36)$ for $1 \leq i \leq 36$; the lower left corners are,
for $a_i$ in $A_0$, $(i-0.5,i+35.5)$, for $a_i$ in $A_1$,
$(i-1.5,0)$, for $a_i$ in $A_2$, $(i-1.5,8.5)$, for $a_i$ in
$A_3$, $(i-1.5,16.5)$, for $a_i$ in $A_4$, $(i-1.5,24.5)$, for
$a_i$ in $A_5$, $(i-0.5,0)$. If the lower left and upper right
corners of $a_i$ are $(x,y)$ and $(w,z)$, respectively, then the
lower left and upper right corners of $b_i$ are $(y,x)$ and
$(z,w)$, respectively. It is not hard to verify that this is a
representation of $G^*$. The representation is drawn in
Figure~\ref{fig:contra-ej}.
\end{proof}

\begin{lemma}\label{lem:bip-comp}
Let $H$ be a complete bipartite graph with bipartition $(A,B)$. In
every 2-thin representation of $H$, except perhaps for the
greatest vertex of $A$ and the greatest vertex of $B$ (according
to the order associated with the representation), every vertex of
$A$ is in one class and every vertex of $B$ is in the other class.
\end{lemma}

\begin{proof}
Let $<$ be the order associated with a 2-thin representation of
$H$. Let $a_M$ (resp. $b_M$) be the greatest vertex of $A$ (resp.
$B$) according to $<$. By symmetry of the graph (since the sizes
of $A$ and $B$ are not specified in the statement), we may assume
without loss of generality $a_M > b_M$.

We will prove that $A \setminus \{a_M\}$ is complete to $B
\setminus \{b_M\}$ in $G_{<}$. Let $a \in A \setminus \{a_M\}$, $b
\in B$ such that $b < a$. Then, $b < a < a_M$, $a_M b \in E(G)$
and $a_M a \not \in E(G)$, therefore $ab \in E(G_{<})$. Now, let
$b \in B \setminus \{b_M\}$, $a \in A$ such that $a < b$. Then, $a
< b < b_M$, $b_M a \in E(G)$ and $b_M b \not \in E(G)$, thus $ab
\in E(G_{<})$. Hence, $A \setminus \{a_M\}$ is complete to $B
\setminus \{b_M\}$ in $G_{<}$.

In particular, $A \setminus \{a_M\}$ and $B \setminus \{b_M\}$ are
in different sets of the partition, and since there are only two
sets in the partition, the statement holds.
\end{proof}

\begin{proposition}\label{prop:th3}
The graph $G^*$ from Definition~\ref{def-G-star} has thinness~3.
\end{proposition}

\begin{proof}
Let $V^1 = A \setminus A_0$, $V^2 = B \setminus B_0$, $V^3 = A_0
\cup B_0$ and the order given by $a_1$, $a_2$, $b_1$, $b_2$,
$a_3$, $a_4$, $b_3$, $b_4$, $\dots a_{35}$, $a_{36}$, $b_{35}$,
$b_{36}$. Let us see that the order and the partition are
consistent. Let $x < y < z$, such that $x$ and $y$ belong to the
same class $V^{\ell}$ and $xz \in E(G^*)$. Since for each vertex
$x$ of $V^3$ its only neighbor is the vertex immediately after $x$
in the order, $\ell \neq 3$ and $z \not \in V^3$.

Suppose first that $\ell=1$, so $x=a_i$, $y = a_j$, and $z=b_k \in
V^2$, since $V^1$ is an independent set. If $z \in B_5$, then $zy
\in E(G^*)$, as required. Otherwise, the indices $i,j,k$ are even
and, by the definition of the order, $i < j \leq k$.  If $z \in
B_t$, $1 \leq t \leq 4$, then $x = a_i \in A_t$ because it is
adjacent to $z$, $a_k \in A_t$ because of the symmetric
definitions of $B_t$ and $A_t$, and $y = a_j \in A_t$ because $i <
j \leq k$, and all the indices are even. Thus $zy \in E(G^*)$, as
required.

Suppose now that $\ell=2$, so $x=b_i$, $y = b_j$, and $z=a_k \in
V^1$, since $V^2$ is an independent set. If $z \in A_5$, then $zy
\in E(G^*)$, as required. Otherwise, the indices $i,j,k$ are even
and, by the definition of the order, $i < j < k$. If $z \in A_t$,
$1 \leq t \leq 4$, then $x = b_i \in B_t$ because it is adjacent
to $z$, $b_k \in B_t$ because of the symmetric definitions of
$B_t$ and $A_t$, and $y = b_j \in B_t$ because $i < j < k$, and
all the indices are even. Thus $zy \in E(G^*)$, as required.
Therefore, $\thin(G^*) \leq 3$.

Now, suppose that $G^*$ admits a 2-thin representation
$(<,V^1,V^2)$. Notice that $(A_j \cup A_5, B_j \cup B_5)$ induce a
complete bipartite graph, for $j = 1, \dots, 4$. So, by
Lemma~\ref{lem:bip-comp} and transitivity, except perhaps for a
few vertices, the vertices of $A \setminus A_0$ are in one of the
sets of the partition, say $V^1$, and the vertices of $B \setminus
B_0$ are in the other, say $V^2$. Let us call $A'$ the vertices of
$A \setminus A_0$ that are in $V^2$ and $B'$ the vertices of $B
\setminus B_0$ that are in $V^1$. For $j = 1, \dots, 5$, let $A_j'
= A_j \setminus A'$ and $B_j' = B_j \setminus B'$.

For $1 \leq j \leq 5$, let $a^j_M, a^j_S$ and $b^j_M, b^j_S$ be
the greatest and smallest vertices of $A_j'$ and $B_j'$,
respectively.

Let $\{i, j, k, \ell\} = \{1,2,3,4\}$. Then, for every vertex $a
\in A_j' \cup A_k' \cup A_{\ell}'$, either $a
> b^i_M$ or $a < a'$ for every $a' \in A_i'$, and, analogously, for every vertex $b \in B_j' \cup
B_k' \cup B_{\ell}'$, either $b
> a^i_M$ or $b < b'$ for every $b' \in B_i'$.

By symmetry of the graph $G^*$, we may assume $a^1_M < a^2_M <
a^3_M < a^4_M$ and $a^4_M > \max\{b^1_M,b^2_M,b^3_M,b^4_M\}$. By
the observation above, for every $b \in B_3' \cup B_2' \cup B_1'$
and $b' \in B_4'$, it holds $b < b'$. In particular, $b^4_M
> b^4_S
> \max\{b^3_M,b^2_M,b^1_M\}$. Suppose $a^3_M
> b^4_S$. Then $a^3_M > b^4_S > b^3_M$, $a^3_M b^3_M \in E(G^*)$ and
$a^3_M b^4_S \not \in E(G^*)$, a contradiction. Then $a^3_M <
b^4_S < b^4_M$, and hence, for every $a \in A_3' \cup A_2' \cup
A_1'$ and $a' \in A_4'$, it holds $a < a'$, i.e., $a^4_S
> a^3_M$.

Suppose $b^3_M > a^4_S$. Then $b^3_M > a^4_S > a^3_M$, $a^3_M
b^3_M \in E(G^*)$ and $b^3_M a^4_S \not \in E(G^*)$, a
contradiction. Then $b^3_M < a^4_S$. Suppose now $b^2_M > a^3_M$.
Then $b^2_M
> a^3_M > a^2_M$, $b^2_M a^2_M \in E(G^*)$ and $b^2_M a^3_M \not \in
E(G^*)$, a contradiction. Then $b^2_M < a^3_M$, and hence, for
every $b \in B_2' \cup B_1'$ and $b' \in B_3'$, it holds $b < b'$,
i.e., $b^3_S > \max\{b^2_M,b^1_M\}$.

Next, suppose $a^2_M > b^3_S$. Then $a^2_M > b^3_S > b^2_M$,
$a^2_M b^2_M \in E(G^*)$ and $a^2_M b^3_S \not \in E(G^*)$, a
contradiction. Then $a^2_M < b^3_S < b^3_M$, and hence, for every
$a \in A_2' \cup A_1'$ and $a' \in A_3'$, it holds $a < a'$, i.e.,
$a^2_M < a^3_S$. Suppose now $b^2_M > a^3_S$. Then $b^2_M > a^3_S
> a^2_M$, $a^2_M b^2_M \in E(G^*)$ and $b^2_M a^3_S \not \in E(G^*)$,
a contradiction. Then $b^2_M < a^3_S$.

Suppose $b^1_M > a^2_M$. Then $b^1_M
> a^2_M > a^1_M$, $b^1_M a^1_M \in E(G^*)$ and $b^1_M a^2_M \not \in
E(G^*)$, a contradiction. Then $b^1_M < a^2_M$, and hence, for
every $b \in B_1'$ and $b' \in B_2'$, it holds $b < b'$, i.e.,
$b^2_S > b^1_M$. Suppose $b^1_M > a^2_S$. Then $b^1_M > a^2_S >
a^1_M$, $a^1_M b^1_M \in E(G^*)$ and $b^1_M a^2_S \not \in
E(G^*)$, a contradiction. Then $b^1_M < a^2_S$.

Finally, suppose $a^1_M > b^2_S$. Then $a^1_M > b^2_S > b^1_M$,
$a^1_M b^1_M \in E(G^*)$ and $a^1_M b^2_S \not \in E(G^*)$, a
contradiction. Then $a^1_M < b^2_S < b^2_M$, and hence, for every
$a \in A_1'$ and $a' \in A_2'$, it holds $a < a'$, i.e., $a^1_M <
a^2_S$. Suppose now $b^1_M > a^2_S$. Then $b^1_M > a^2_S
> a^1_M$, $a^1_M b^1_M \in E(G^*)$ and $b^1_M a^2_S \not \in E(G^*)$,
a contradiction. Then $b^1_M < a^2_S$.

So, $A_1' < A_2' < A_3' < A_4'$, $B_1' < B_2' < B_3' < B_4'$, and
$\max\{a^1_M,b^1_M\} < \min\{a^2_S,b^2_S\} < \max\{a^2_M,b^2_M\} <
\min\{a^3_S,b^3_S\} < \max\{a^3_M,b^3_M\} < \min\{a^4_S,b^4_S\} <
b^4_M < a^4_M$.

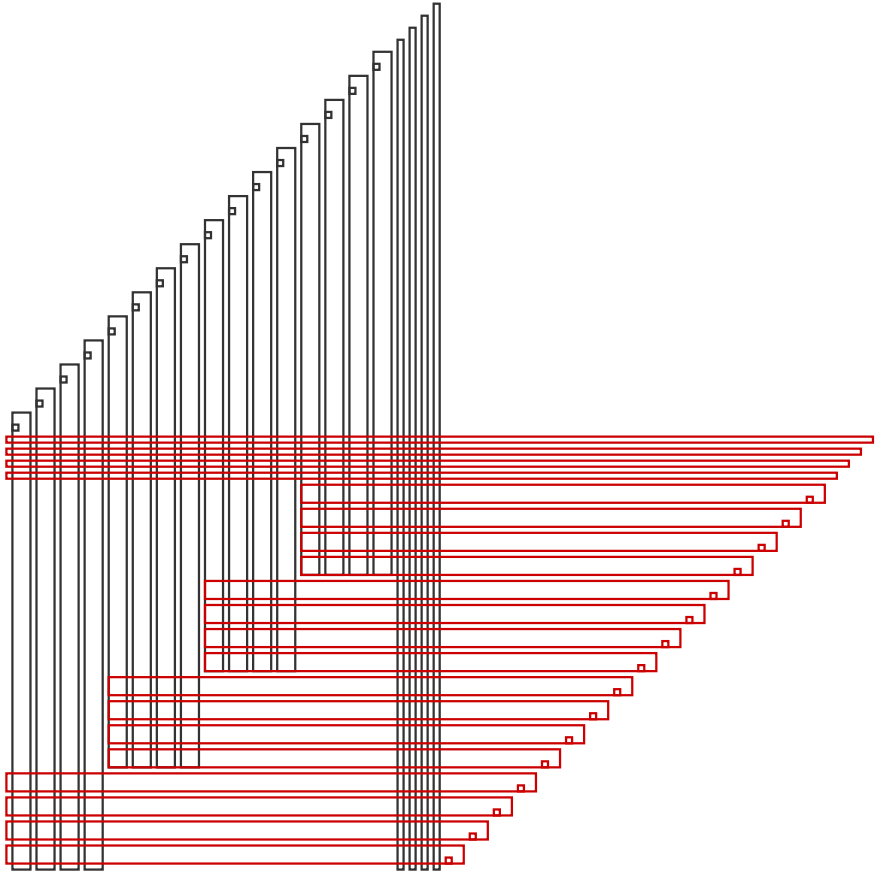
\begin{figure}[p]
\begin{center}
\begin{tikzpicture}[x=0.16cm,y=0.16cm]

\draw[line width=0.3mm,color=darkcyan] (35.5, 0)    rectangle (36, 72);  
\draw[line width=0.3mm,color=darkcyan] (34.5, 0)    rectangle (35, 71);  
\draw[line width=0.3mm,color=darkcyan] (33.5, 0)    rectangle (34, 70);  
\draw[line width=0.3mm,color=darkcyan] (32.5, 0)    rectangle (33, 69);  
\draw[line width=0.3mm,color=darkcyan] (30.5, 24.5) rectangle (32, 68);  
\draw[line width=0.3mm,color=darkcyan] (30.5, 66.5) rectangle (31, 67);  
\draw[line width=0.3mm,color=darkcyan] (28.5, 24.5) rectangle (30, 66);  
\draw[line width=0.3mm,color=darkcyan] (28.5, 64.5) rectangle (29, 65);  
\draw[line width=0.3mm,color=darkcyan] (26.5, 24.5) rectangle (28, 64);  
\draw[line width=0.3mm,color=darkcyan] (26.5, 62.5) rectangle (27, 63);  
\draw[line width=0.3mm,color=darkcyan] (24.5, 24.5) rectangle (26, 62);  
\draw[line width=0.3mm,color=darkcyan] (24.5, 60.5) rectangle (25, 61);  
\draw[line width=0.3mm,color=darkcyan] (22.5, 16.5) rectangle (24, 60);  
\draw[line width=0.3mm,color=darkcyan] (22.5, 58.5) rectangle (23, 59);  
\draw[line width=0.3mm,color=darkcyan] (20.5, 16.5) rectangle (22, 58);  
\draw[line width=0.3mm,color=darkcyan] (20.5, 56.5) rectangle (21, 57);  
\draw[line width=0.3mm,color=darkcyan] (18.5, 16.5) rectangle (20, 56);  
\draw[line width=0.3mm,color=darkcyan] (18.5, 54.5) rectangle (19, 55);  
\draw[line width=0.3mm,color=darkcyan] (16.5, 16.5) rectangle (18, 54);  
\draw[line width=0.3mm,color=darkcyan] (16.5, 52.5) rectangle (17, 53);  
\draw[line width=0.3mm,color=darkcyan] (14.5, 8.50) rectangle (16, 52);  
\draw[line width=0.3mm,color=darkcyan] (14.5, 50.5) rectangle (15, 51);  
\draw[line width=0.3mm,color=darkcyan] (12.5, 8.5)  rectangle (14, 50);  
\draw[line width=0.3mm,color=darkcyan] (12.5, 48.5) rectangle (13, 49);  
\draw[line width=0.3mm,color=darkcyan] (10.5, 8.5)  rectangle (12, 48);  
\draw[line width=0.3mm,color=darkcyan] (10.5, 46.5) rectangle (11, 47);  
\draw[line width=0.3mm,color=darkcyan] (8.5, 8.5)   rectangle (10, 46);  
\draw[line width=0.3mm,color=darkcyan] (8.5, 44.5)  rectangle (9, 45);   
\draw[line width=0.3mm,color=darkcyan] (6.5, 0)     rectangle (8, 44);   
\draw[line width=0.3mm,color=darkcyan] (6.5, 42.5)  rectangle (7, 43);   
\draw[line width=0.3mm,color=darkcyan] (4.5, 0)     rectangle (6, 42);   
\draw[line width=0.3mm,color=darkcyan] (4.5, 40.5)  rectangle (5, 41);   
\draw[line width=0.3mm,color=darkcyan] (2.5, 0)     rectangle (4, 40);   
\draw[line width=0.3mm,color=darkcyan] (2.5, 38.5)  rectangle (3, 39);   
\draw[line width=0.3mm,color=darkcyan] (0.5, 0)     rectangle (2, 38);   
\draw[line width=0.3mm,color=darkcyan] (0.5, 36.5)  rectangle (1, 37);   


\draw[line width=0.3mm,color=orange] (0, 35.5)    rectangle (72, 36);  
\draw[line width=0.3mm,color=orange] (0, 34.5)    rectangle (71, 35);  
\draw[line width=0.3mm,color=orange] (0, 33.5)    rectangle (70, 34);  
\draw[line width=0.3mm,color=orange] (0, 32.5)    rectangle (69, 33);  
\draw[line width=0.3mm,color=orange] (24.5, 30.5) rectangle (68, 32);  
\draw[line width=0.3mm,color=orange] (66.5, 30.5) rectangle (67, 31);  
\draw[line width=0.3mm,color=orange] (24.5, 28.5) rectangle (66, 30);  
\draw[line width=0.3mm,color=orange] (64.5, 28.5) rectangle (65, 29);  
\draw[line width=0.3mm,color=orange] (24.5, 26.5) rectangle (64, 28);  
\draw[line width=0.3mm,color=orange] (62.5, 26.5) rectangle (63, 27);  
\draw[line width=0.3mm,color=orange] (24.5, 24.5) rectangle (62, 26);  
\draw[line width=0.3mm,color=orange] (60.5, 24.5) rectangle (61, 25);  
\draw[line width=0.3mm,color=orange] (16.5, 22.5) rectangle (60, 24);  
\draw[line width=0.3mm,color=orange] (58.5, 22.5) rectangle (59, 23);  
\draw[line width=0.3mm,color=orange] (16.5, 20.5) rectangle (58, 22);  
\draw[line width=0.3mm,color=orange] (56.5, 20.5) rectangle (57, 21);  
\draw[line width=0.3mm,color=orange] (16.5, 18.5) rectangle (56, 20);  
\draw[line width=0.3mm,color=orange] (54.5, 18.5) rectangle (55, 19);  
\draw[line width=0.3mm,color=orange] (16.5, 16.5) rectangle (54, 18);  
\draw[line width=0.3mm,color=orange] (52.5, 16.5) rectangle (53, 17);  
\draw[line width=0.3mm,color=orange] (8.5, 14.5)  rectangle (52, 16);  
\draw[line width=0.3mm,color=orange] (50.5, 14.5) rectangle (51, 15);  
\draw[line width=0.3mm,color=orange] (8.5, 12.5)  rectangle (50, 14);  
\draw[line width=0.3mm,color=orange] (48.5, 12.5) rectangle (49, 13);  
\draw[line width=0.3mm,color=orange] (8.5, 10.5)  rectangle (48, 12);  
\draw[line width=0.3mm,color=orange] (46.5, 10.5) rectangle (47, 11);  
\draw[line width=0.3mm,color=orange] (8.5, 8.5)   rectangle (46, 10);  
\draw[line width=0.3mm,color=orange] (44.5, 8.5)  rectangle (45, 9);   
\draw[line width=0.3mm,color=orange] (0, 6.5)     rectangle (44, 8);   
\draw[line width=0.3mm,color=orange] (42.5, 6.5)  rectangle (43, 7);   
\draw[line width=0.3mm,color=orange] (0, 4.5)     rectangle (42, 6);   
\draw[line width=0.3mm,color=orange] (40.5, 4.5)  rectangle (41, 5);   
\draw[line width=0.3mm,color=orange] (0, 2.5)     rectangle (40, 4);   
\draw[line width=0.3mm,color=orange] (38.5, 2.5)  rectangle (39, 3);   
\draw[line width=0.3mm,color=orange] (0, 0.5)     rectangle (38, 2);   
\draw[line width=0.3mm,color=orange] (36.5, 0.5)  rectangle (37, 1);   

\end{tikzpicture}
\end{center}
\caption{The graph $G^*$ (Definition~\ref{def-G-star}) used to
prove that the blocking property is necessary for the
characterization of 2-thin graphs as boxicity 2 graphs
(Proposition~\ref{prop:rep}).}\label{fig:contra-ej}
\end{figure}

The vertices in $B_5'$ have to be greater than $b^3_M$, which is a
non-neighbor of $a^4_M$ and smaller than it. Similarly, the
vertices in $A_5'$ have to be greater than $a^3_M$.

Let $a^2_2, a^2_3$ be the second and third greatest vertices of
$A_2'$, respectively. Let $a_0$ be the neighbor of $a^2_2$ in
$A_0$.

Suppose first $a_0 \in V^1$. If $a_0 > a^2_2$, then $a_0 < a^2_M$,
because $a_0 a^2_M \not \in E(G^*)$. If $a_0 < a^2_2$, then $a_0 >
a^2_3$, because $a^2_2 a^2_3 \not \in E(G^*)$. Let $b^5 \in B_5'$.
Then $b^5 > b^3_M > a^2_M > a_0 > a^2_S$, but $b^5 a^2_S \in
E(G^*)$ and $b^5 a^0 \not \in E(G^*)$, a contradiction.

Suppose now $a^0 \in V^2$. If $a^0 > a^2_2$, then $a^0 < a^2_M$,
because $a^0 a^2_M \not \in E(G^*)$. If $a^0 < a^2_2$, then $a^0 >
b^1_M$, because $a^2_2 > b^1_M$ and $a^2_2 b^1_M \not \in E(G^*)$.
Let $a^5 \in A_5'$. Then $a^5 > a^3_M > a^2_2 > a^0 > b^1_M$, but
$a^5 b^1_M \in E(G^*)$ and $a^5 a^0 \not \in E(G^*)$, a
contradiction.
\end{proof}

\subsection{Characterization of proper 2-thin graphs as rectangle intersection graphs}

For proper 2-thin graphs, we can relax the 2-diagonal property and
do not require the blocking property, by requiring the model to be
bi-semi proper.

\begin{theorem}\label{thm:char-2-pthin} Let $G$ be a graph. The following statements are equivalent:

\begin{itemize}
\item[$(i)$] $G$ is a proper 2-thin graph.

\item[$(ii)$] $G$ has a bi-semi-proper blocking 2-diagonal model.

\item[$(iii)$] $G$ has a bi-semi-proper weakly 2-diagonal model.
\end{itemize}

Moreover, if $G$ is proper 2-thin and the partition $V^1,V^2$ of
its vertices is strongly consistent with an order $<$, then there
exists a bi-semi-proper blocking 2-diagonal model such that on
each of the diagonals lie, respectively, the upper-right corners
of the vertices of $V^1$ and $V^2$, in such a way that their order
corresponds to $<$ restricted to the respective part. Furthermore,
if $G$ admits a bi-semi-proper weakly 2-diagonal model, then there
exists an order of the vertices of $G$ that is consistent with the
partition given by the diagonals where the upper-right corners
lie, and extends their order on the respective diagonals.
\end{theorem}

\begin{proof} $(i) \Rightarrow (ii)$) It follows from Lemma~\ref{lem:model1}. \\

\noindent
$(ii) \Rightarrow (iii)$) This is straightforward. \\

\noindent $(iii) \Rightarrow (i)$) Let us consider a
bi-semi-proper weakly 2-diagonal model of $G$. We will slightly
abuse notation and use it indistinctly for a vertex and the box
representing it. Let $V^i$ be the set of vertices $v$ such that
$Y_2(v) - X_2(v) = d_i$, for $i = 1,2$. We may assume without loss
of generality that $d_1 < d_2$.

Let $\Pi = \{V^1,V^2\}$, $<$ be the order of $V^1 \cup V^2$
defined by the $X_2$ coordinates on each of the sets, and where a
vertex of $V^1$ and a vertex of $V^2$ are not comparable.

Let us first prove that $<$ is strongly consistent restricted to
$V^i$, $i = 1,2$. Let $x < y < z$ in $V^1$ with $xz \in E(G)$ (the
definitions are symmetric with respect to both classes). Then
$X_2(x) < X_2(y) < X_2(z)$ and since $xz \in E(G)$, it holds
$X_1(z) < X_2(x) < X_2(y)$ and $Y_1(z) < Y_2(x) = X_2(x) + d_1 <
X_2(y) + d_1 = Y_2(y)$. Therefore, $yz \in E(G)$. By the
bi-semi-proper property, $X_1(y) \leq X_1(z) < X_2(x)$ and $Y_1(y)
\leq Y_1(z) < Y_2(x)$, and therefore $xy \in E(G)$.

Let $D = \tilde{D}(G,\Pi,<)$. By Lemma~\ref{lem:D-pthin}, an
ordering of $V(G)$ is strongly consistent with the partition
$V^1,V^2$ and extends the partial order $<$ if and only if it is a  
topological ordering of $D$, thus let us prove that $D$ is
acyclic, Suppose it is not, and consider a shortest directed cycle
of $D$. Moreover, since the subdigraph induced by each class is
complete and acyclic, the cycle has at most two vertices of each
class, and necessarily an arc from $V^1$ to $V^2$ and another from
$V^2$ to $V^1$.

The possible types of arcs joining $v_1 \in V^1$ and $v_2 \in V^2$
are the following:
\begin{itemize}
\item Type $a$ if $X_2(v_1)
> X_2(v_2)$ and $Y_2(v_1) < Y_2(v_2)$;
\item Type $b$ if $X_2(v_1) < X_2(v_2)$ and $Y_2(v_1) < Y_2(v_2)$;
\item Type $c$ if $X_2(v_1) > X_2(v_2)$ and $Y_2(v_1) > Y_2(v_2)$;
\item Type $d$ if $X_2(v_1) = X_2(v_2)$ and $Y_2(v_1) < Y_2(v_2)$;
\item Type $e$ if $X_2(v_1) > X_2(v_2)$ and $Y_2(v_1) = Y_2(v_2)$,
\end{itemize}
and we use the subindex $ij$ if the orientation of the arc is from
$V^i$ to $V^j$.

The following properties hold because the model is bi-semi-proper.

\begin{itemize}
\item Type $a_{12}$: $Y_1(v_2) > Y_2(v_1)$ and $X_1(v_1) <
X_2(v_2)$;

\item Type $a_{21}$: $Y_1(v_2) < Y_2(v_1)$ and $X_1(v_1) >
X_2(v_2)$;

\item Type $b_{12}$: $Y_1(v_2) > Y_2(v_1)$ or $X_1(v_2) >
X_2(v_1)$;

\item Types $b_{21}$ and $c_{12}$: cannot exist;

\item Type $c_{21}$: $Y_1(v_1) > Y_2(v_2)$ or $X_1(v_1) >
X_2(v_2)$;

\item Type $d_{12}$: $Y_1(v_2) > Y_2(v_1)$;

\item Types $d_{21}$ and $e_{12}$: cannot exist;

\item Type $e_{21}$: $X_1(v_1) > X_2(v_2)$.
\end{itemize}

Let us see the possible cycles.

\noindent \textbf{{Case 1:}} The cycle consists of two vertices,
$v_1 \in V^1$ and $v_2 \in V^2$.

It cannot happen, since for every $t \in \{b,c,d,e\}$, one of
$t_{12}$ or $t_{21}$ cannot exist, and $a_{12}$ and $a_{21}$
cannot occur simultaneously.

\noindent \textbf{{Case 2:}} The cycle is $v_1 w_1 v_2$ such that
$v_1, w_1 \in V^1$ and $v_2 \in V^2$.

Since $v_1 w_1 \in D$, we have $X_2(v_1) < X_2(w_1)$ and therefore
 $Y_2(v_1) < Y_2(w_1)$. As the model is bi-semi-proper, $X_1(v_1) \leq X_1(w_1)$ and $Y_1(v_1) \leq
Y_1(w_1)$.

The arc $w_1 v_2$, as seen before, may be of type $a_{12}$, thus
$Y_1(v_2)
> Y_2(w_1)$ and $X_1(w_1) < X_2(v_2)$; or type $b_{12}$, thus
$Y_1(v_2) > Y_2(w_1)$ or $X_1(v_2)
> X_2(w_1)$; or type $d_{12}$, and in that case $Y_1(v_2) > Y_2(w_1)$.

Also, the arc $v_2 v_1$, as seen before, can be type $a_{21}$,
thus $Y_1(v_2) < Y_2(v_1)$ and $X_1(v_1) > X_2(v_2)$; or type
$c_{21}$, thus $Y_1(v_1) > Y_2(v_2)$ or $X_1(v_1) > X_2(v_2)$; or
type $e_{21}$, and in that case $X_1(v_1)
>
X_2(v_2)$.

The options $a_{12}$ and $a_{21}$ imply $Y_2(w_1) < Y_1(v_2) <
Y_2(v_1)$, but we know that $Y_2(v_1) < Y_2(w_1)$.

The options $a_{12}$ and $c_{21}$ imply either $Y_2(v_1) >
Y_1(v_1)
> Y_2(v_2) > Y_1(v_2) > Y_2(w_1)$, which contradicts $Y_2(v_1) <
Y_2(w_1)$, or $X_1(v_1)
> X_2(v_2) > X_1(w_1)$, which contradicts $X_1(v_1) \leq
X_1(w_1)$.

The options $a_{12}$ and $e_{21}$ imply $X_1(v_1) > X_2(v_2)
> X_1(w_1)$, which contradicts $X_1(v_1) \leq X_1(w_1)$.

The options $b_{12}$ and $a_{21}$ are incompatible because they
imply either $Y_2(v_1) > Y_1(v_2) > Y_2(w_1)$, that contradicts
$Y_2(v_1) < Y_2(w_1)$, or $X_2(v_1) > X_1(v_1) > X_2(v_2) >
X_1(v_2) > X_2(w_1)$, that contradicts $X_2(v_1) < X_2(w_1)$.

The options $b_{12}$ and $c_{21}$ are incompatible because
$Y_2(v_1) > Y_1(v_1) > Y_2(v_2) > Y_2(w_1)$ contradicts $Y_2(v_1)
< Y_2(w_1)$; $X_2(v_1) > X_1(v_1) > X_2(v_2) > X_2(w_1)$
contradicts $X_2(v_1) < X_2(w_1)$; $X_1(v_2) > X_2(w_1)$ and
$Y_1(v_1) > Y_2(v_2)$ implies $Y_2(v_2) = X_2(v_2) + d_2 >
X_1(v_2) + d_2 > X_2(w_1) + d_2 > X_2(v_1) + d_2 = Y_2(v_1) +
d_2-d_1 > Y_1(v_1) + d_2 - d_1 > Y_2(v_2) + d_2 - d_1 > Y_2(v_2)$,
a contradiction. Finally, suppose that $Y_1(v_2) > Y_2(w_1)$ and
$X_1(w_1) \geq X_1(v_1)
> X_2(v_2)$. The existence of the arc $w_1v_2$ implies that either
there exists $v_{2}' \in V^{2}$ with $X_2(v_2') < X_2(v_2)$ and
$w_1v_{2}' \in E(G)$, or there exists $v_1' \in V^1$ with
$X_2(v_1')
> X_2(w_1)$ and $v_1'v_{2} \in E(G)$. But $X_2(v_2') < X_2(v_2) <
X_1(w_1)$ contradicts $w_1v_{2}' \in E(G)$ and $X_2(v_1')
> X_2(w_1)$ implies $X_1(v_1') \geq X_1(w_1) > X_2(v_2)$ which
contradicts $v_1'v_{2} \in E(G)$.

The options $b_{12}$ and $e_{21}$ are incompatible because they
imply either $Y_2(v_1) = Y_1(v_2) > Y_2(w_1)$, a contradiction, or
$X_1(v_1) > X_2(v_2) > X_1(v_2)
> X_2(w_1)$, a contradiction too.

The options $d_{12}$ and $a_{21}$ imply $Y_2(v_1) > Y_1(v_2)
> Y_2(w_1)$, a contradiction.

The options $d_{12}$ and $c_{21}$ imply either $Y_2(v_1) >
Y_1(v_1)
> Y_2(v_2) > Y_2(w_1)$, a contradiction, or $X_2(v_1) > X_1(v_1) >
X_2(v_2) = X_2(w_1)$, a contradiction too.

The options $d_{12}$ and $e_{21}$ imply $Y_2(v_1) = Y_2(v_2)
> Y_1(v_2) > Y_2(w_1)$, a contradiction.

\noindent \textbf{{Case 3:}} The cycle is $v_2 w_2 v_1$ such that
$v_1 \in V^1$ and $v_2, w_2 \in V^2$.

Since $v_2 w_2 \in D$, we have $X_2(v_2) < X_2(w_2)$ and therefore
$Y_2(v_2) < Y_2(w_2)$. As the model is bi-semi-proper, $X_1(v_2)
\leq X_1(w_2)$ and $Y_1(v_2) \leq Y_1(w_2)$.

The arc $v_1 v_2$, as seen before, can be type $a_{12}$, thus
$Y_1(v_2)
> Y_2(v_1)$ and $X_1(v_1) < X_2(v_2)$; or type $b_{12}$, thus
$Y_1(v_2) > Y_2(v_1)$ or $X_1(v_2)
> X_2(v_1)$; or type $d_{12}$, and in that case $Y_1(v_2) > Y_2(v_1)$.

Also, the arc $w_2 v_1$, as seen before, can be type $a_{21}$,
thus $Y_1(w_2) < Y_2(v_1)$ and $X_1(v_1) > X_2(w_2)$; or type
$c_{21}$, thus $Y_1(v_1) > Y_2(w_2)$ or $X_1(v_1) > X_2(w_2)$; or
type $e_{21}$, and in that case $X_1(v_1)
>
X_2(w_2)$.

The options $a_{12}$ and $a_{21}$ imply $Y_1(v_2) > Y_2(v_1)
> Y_1(w_2)$, a contradiction.

The options $a_{12}$ and $c_{21}$ imply either $Y_2(v_2) >
Y_1(v_2)
> Y_2(v_1) > Y_1(v_1) > Y_2(w_2)$, a contradiction, or $X_2(w_2) <
X_1(v_1) < X_2(v_2)$, a contradiction.

The options $a_{12}$ and $e_{21}$ imply $X_2(w_2) < X_1(v_1) <
X_2(v_2)$, a contradiction.

The options $b_{12}$ and $a_{21}$ imply either $Y_1(v_2) >
Y_2(v_1)
> Y_1(w_2)$, a contradiction, or $X_2(v_2) > X_1(v_2)
> X_2(v_1) > X_1(v_1) >
X_2(w_2)$, a contradiction.

The options $b_{12}$ and $c_{21}$ are incompatible because
$Y_2(v_2) > Y_1(v_2) > Y_2(v_1) > Y_2(w_2)$ contradicts $Y_2(v_2)
< Y_2(w_2)$; $X_1(v_2) > X_2(v_1) > X_2(v_2)$ contradicts
$X_1(v_2) < X_2(v_2)$; $X_1(v_2) > X_2(v_1)$ and $Y_1(v_1) >
Y_2(w_2)$ implies $Y_2(v_2) = X_2(v_2) + d_2 > X_1(v_2) + d_2 >
X_2(v_1) + d_2 = Y_2(v_1) + d_2-d_1 > Y_1(v_1) + d_2 - d_1
> Y_2(w_2) + d_2 - d_1 > Y_2(w_2)$, a contradiction. Finally,
suppose that $Y_1(v_2) > Y_2(v_1)$ and $X_1(v_1)
> X_2(w_2)$. The existence of the arc $v_1v_2$ implies that either
there exists $v_{2}' \in V^{2}$ with $X_2(v_2') < X_2(v_2)$ and
$v_1v_{2}' \in E(G)$, or there exists $v_1' \in V^1$ with
$X_2(v_1')
> X_2(v_1)$ and $v_1'v_{2} \in E(G)$. But $X_2(v_2') < X_2(v_2) <
X_2(w_2) < X_1(v_1)$ contradicts $v_1v_{2}' \in E(G)$ and
$X_2(v_1') > X_2(v_1)$ implies $X_1(v_1') \geq X_1(v_1) > X_2(w_2)
> X_2(v_2)$ which contradicts $v_1'v_{2} \in E(G)$.

The options $b_{12}$ and $e_{21}$ imply $Y_2(v_2) > Y_1(v_2)
> Y_2(v_1) = Y_2(w_2)$, a contradiction.

The options $d_{12}$ and $a_{21}$ imply $Y_1(v_2) > Y_2(v_1)
> Y_1(w_2)$, a contradiction.

The options $d_{12}$ and $c_{21}$ imply either $Y_2(v_2) >
Y_1(v_2)
> Y_2(v_1) > Y_1(v_1) > Y_2(w_2)$, a contradiction, or $X_2(v_2) =
X_2(v_1) > X_1(v_1) > X_2(w_2)$, a contradiction.

The options $d_{12}$ and $e_{21}$ imply $Y_2(v_2) > Y_1(v_2)
> Y_2(v_1) = Y_2(w_2)$, a contradiction.

\noindent \textbf{{Case 4:}} The cycle is $v_1 w_1 v_2 w_2$ such
that $v_1, w_1 \in V^1$ and $v_2, w_2 \in V^2$.

Since $v_i w_i \in D$, for $i = 1,2$, we have $X_2(v_i) <
X_2(w_i)$ therefore $Y_2(v_i) < Y_2(w_i)$. As the model is
bi-semi-proper, $X_1(v_i) \leq X_1(w_i)$ and $Y_1(v_i) \leq
Y_1(w_i)$.

The arc $w_1 v_2$, as seen before, can be type $a_{12}$, thus
$Y_1(v_2)
> Y_2(w_1)$ and $X_1(w_1) < X_2(v_2)$; or type $b_{12}$, thus
$Y_1(v_2) > Y_2(w_1)$ or $X_1(v_2)
> X_2(w_1)$; or type $d_{12}$, and in that case $Y_1(v_2) > Y_2(w_1)$.

Also, the arc $w_2 v_1$, as seen before, can be type $a_{21}$,
thus $Y_1(w_2) < Y_2(v_1)$ and $X_1(v_1) > X_2(w_2)$; or type
$c_{21}$, thus $Y_1(v_1) > Y_2(w_2)$ or $X_1(v_1) > X_2(w_2)$; or
type $e_{21}$, and in that case $X_1(v_1)
>
X_2(w_2)$.

The options $a_{12}$ and $a_{21}$ imply $Y_2(v_1) > Y_1(w_2) \geq
Y_1(v_2) > Y_2(w_1)$, a contradiction.

The options $a_{12}$ and $c_{21}$ imply $Y_2(v_2) > Y_1(v_2)
> Y_2(w_1) > Y_2(v_1) > Y_1(v_1)
> Y_2(w_2)$, a contradiction.

The options $a_{12}$ and $e_{21}$ imply $X_1(w_1) < X_2(v_2) <
X_2(w_2) < X_1(v_1)$, a contradiction.

The options $b_{12}$ and $a_{21}$ imply either $Y_2(v_1) >
Y_1(w_2) \geq Y_1(v_2) > Y_2(w_1)$, a contradiction, or $X_2(v_2)
> X_1(v_2)
> X_2(w_1) > X_2(v_1) > X_1(v_1) >
X_2(w_2)$, a contradiction.


The options $b_{12}$ and $c_{21}$ are incompatible because
$Y_2(v_2) > Y_1(v_2) > Y_2(w_1) > Y_2(v_1) > Y_2(w_2)$ contradicts
$Y_2(v_2) < Y_2(w_2)$; $X_2(v_2) > X_1(v_2) > X_2(w_1) > X_2(v_1)
> X_2(w_2)$, contradicts $X_2(v_2) < X_2(w_2)$;
$X_1(v_2) > X_2(w_1)$ and $Y_1(v_1) > Y_2(w_2)$ implies $Y_2(v_2)
= X_2(v_2) + d_2 > X_1(v_2) + d_2 > X_2(w_1) + d_2 > X_2(v_1) +
d_2 = Y_2(v_1) + d_2-d_1 > Y_1(v_1) + d_2 - d_1
> Y_2(w_2) + d_2 - d_1 > Y_2(w_2)$, a contradiction. Finally,
suppose that $Y_1(v_2) > Y_2(w_1)$ and $X_1(v_1)
> X_2(w_2)$. The existence of the arc $w_1v_2$ implies that either
there exists $v_{2}' \in V^{2}$ with $X_2(v_2') < X_2(v_2)$ and
$w_1v_{2}' \in E(G)$, or there exists $v_1' \in V^1$ with
$X_2(v_1')
> X_2(w_1)$ and $v_1'v_{2} \in E(G)$. But $X_2(v_2') < X_2(v_2) <
X_2(w_2) < X_1(v_1) \leq X_1(w_1)$ contradicts $w_1v_{2}' \in
E(G)$, and $X_2(v_1')
> X_2(w_1)$ implies $X_1(v_1') \geq X_1(w_1) \geq X_1(v_1) > X_2(w_2) >
X_2(v_2)$ which contradicts $v_1'v_{2} \in E(G)$.

The options $b_{12}$ and $e_{21}$ imply $X_1(v_1) > X_2(w_2)
> X_2(v_2) > X_2(w_1)$, a contradiction.

The options $d_{12}$ and $a_{21}$ imply $Y_1(v_2) > Y_2(w_1)
> Y_2(v_1) > Y_1(w_2)$, a contradiction.

The options $d_{12}$ and $c_{21}$ imply $Y_2(v_2) > Y_1(v_2)
> Y_2(w_1) > Y_2(v_1) > Y_1(v_1)
> Y_2(w_2)$, a contradiction.

The options $d_{12}$ and $e_{21}$ imply $Y_2(v_2) > Y_1(v_2)
> Y_2(w_1) > Y_2(v_1) = Y_2(w_2)$, a contradiction.

\end{proof}

The bi-semi-proper requirement is necessary, otherwise we can
represent any interval graph (we place the intervals on a diagonal
line and make each of them the diagonal of a square box, see
Figure~\ref{fig:classes}), and interval graphs may have
arbitrarily large proper thinness~\cite{B-D-thinness}.

\subsection{Definition of the 2-grounded model $\mathcal{M}_2$}

The model $\mathcal{M}_1$ can be transformed into a model where
the rectangles lie within the 3rd quadrant of the Cartesian plane,
and each of the rectangles has either its top side or its right
side on a Cartesian axis. We call such a model \emph{2-grounded}.
Since every pair of intersecting rectangles has nonempty
intersection within the third quadrant, it is enough to define the
model $\mathcal{M}_2$ as the intersection of $\mathcal{M}_1$ and
the 3rd quadrant of the Cartesian plane (see
Figure~\ref{fig:M2-M3}). By adjusting the definitions of blocking
and bi-semi-proper for 2-grounded models, we can prove
characterizations analogous to those in
Theorems~\ref{thm:char-2-thin} and~\ref{thm:char-2-pthin} for
(proper) 2-thin graphs.


\section{Thin graphs as VPG graphs}\label{sec:vpg}

A graph is \emph{$B_k$-VPG} (resp. \emph{$B_k$-EPG}) if it is the
vertex (resp. edge) intersection graph of paths with at most
\emph{$k$ bends} in a grid~\cite{asinowski,Golumbic-epg}.

Every graph is an EPG graph~\cite{Golumbic-epg}, and $B_0$-EPG
graphs are exactly the interval graphs, or 1-thin graphs. However,
we have the following.

\begin{proposition}\label{prop:epg}
The bend number of proper independent 2-thin graphs as EPG graphs
is unbounded.
\end{proposition}

\begin{proof}
Complete bipartite graphs are proper independent 2-thin graphs,
and that class has unbounded EPG bend number~\cite{Suk-epg}.
\end{proof}

The definitions of EPG and VPG graphs are similar, but the classes
behave very differently. On the one hand, VPG graphs, without
bounds in the number of bends, are equivalent to \emph{string
graphs}~\cite{asinowski}, and not every graph is a string
graph~\cite{Sinden-K5}. On the other hand, $B_0$-VPG graphs
properly contain the class of interval graphs.

Indeed, we will first prove that $B_0$-VPG graphs have unbounded
thinness, and that not every 4-thin graph is a VPG graph. Then, we
will prove that graphs with thinness at most three have bounded
bend number as VPG graphs.

Let $G$ be a graph. Let $\Delta(G)$ be the maximum degree of a
vertex in $G$. A subgraph $H$ (not necessarily induced) of $G$ is
a \emph{spanning subgraph} if $V(H)=V(G)$. If $X \subseteq V(G)$,
denote by $N(X)$ the set of vertices of $G$ having at least one
neighbor in $X$. The \emph{vertex isoperimetric peak} of a graph
$G$, denoted as $b_v(G)$, is defined as $b_v(G) = \max_s
\min_{X\subset V, |X|=s} |N(X) \cap (V(G) \setminus X)|$, i.e.,
the maximum over $s$ of the lower bounds for the number of
boundary vertices (vertices outside the set with a neighbor in the
set) in sets of size $s$.

\begin{theorem}\emph{\cite{C-M-O-thinness-man}}\label{thm:peak}
For every graph $G$ with at least one edge, $\thin(G)\geq
b_v(G)/\Delta(G)$.
\end{theorem}

The following corollary is also useful.

\begin{corollary}\label{cor:peak}
Let $G$ be a graph such that $\Delta(G) \leq d$, and $H$ be a (not
necessarily induced) subgraph of $G$ with at least one edge and
such that $b_v(H) \geq b$. Then $\thin(G)\geq b/d$.
\end{corollary}

\begin{proof}
Let $G'$ be the subgraph of $G$ induced by $V(H)$. Then
$\Delta(G') \leq \Delta(G)$. Since $H$ is a spanning subgraph of
$G'$, then $b_v(G') \geq b_v(H)$, and since $G'$ is an induced
subgraph of $G$, then $thin(G) \geq thin(G')$. So, by
Theorem~\ref{thm:peak}, $\thin(G) \geq thin(G') \geq
b_v(G')/\Delta(G') \geq b_v(H)/\Delta(G) \geq b/d$.
\end{proof}

For a positive integer $r$, the \emph{$(r \times r)$-grid} $GR_r$
is the graph whose vertex set is $\{(i,j) : 1 \leq i, j \leq r\}$
and whose edge set is $\{(i,j)(k,l) : |i - k| + |j - l| = 1,
\mbox{ where } 1 \leq i,j,k, l \leq r \}$.

The thinness of the two dimensional $r \times r$ grid $GR_r$ was
lower bounded by $r/4$, by using Theorem~\ref{thm:peak} and the
following lemma.

\begin{lemma}\emph{\cite{Chv-grid}}\label{lem:grid}
For every $r \geq 2$, $b_v(GR_r) \geq r$.
\end{lemma}


We use these results to prove the unboundedness of the thinness of
$B_0$-VPG graphs.

\begin{proposition}\label{prop:b0-unbound}
The class of $B_0$-VPG graphs has unbounded thinness.
\end{proposition}

\begin{proof}
Let $r \geq 2$ and let $G_r$ be the intersection graph of the
following paths on a grid: $\{(i-0.1,j)$--$(i+1.1,j)_{0 \leq i,j
\leq r}\} \cup \{(i,j-1.1)$--$(i,j+0.1)_{1 \leq i,j \leq r}\}$.
The grid $GR_r$ is a subgraph of $G_r$, and $\Delta(G_r) = 6$ (see
Figure~\ref{fig:B0}). So, by Lemma~\ref{lem:grid} and
Corollary~\ref{cor:peak}, $\thin(G) \geq r/6$.
\end{proof}

\begin{figure}
    \begin{center}
    \begin{tikzpicture}[scale=0.7]

\foreach \y in {0,...,4} {    \draw (-8.1,\y) -- (-6.8,\y);
    \draw (-7.1,\y-0.1) -- (-5.8,\y-0.1);
    \draw (-6.1,\y) -- (-4.8,\y);
    \draw (-5.1,\y-0.1) -- (-3.8,\y-0.1);
    \draw (-4.1,\y) -- (-2.8,\y);
}

\foreach \x in {-7,...,-4} {    \draw (\x+0.1,-0.2) --
(\x+0.1,1.1);
    \draw (\x,0.8) -- (\x,2.1);
    \draw (\x+0.1,1.8) -- (\x+0.1,3.1);
    \draw (\x,2.8) -- (\x,4.1);
}

\foreach \x in {0,...,4}
    \foreach \y in {0,...,4}
        \vertex{\x}{\y}{v\x\y};

\foreach \x in {0,...,3}
    \foreach \y in {0,...,3}
        \vertex{\x+0.5}{\y+0.5}{w\x\y};

    \edge{v00}{v11};
    \edge{v01}{v12};
    \edge[ultra thick]{v02}{v13};
    \edge{v03}{v14};
    \edge{v01}{v10};
    \edge[ultra thick]{v02}{v11};
    \edge{v03}{v12};
    \edge{v04}{v13};

    \edge{v10}{v21};
    \edge[ultra thick]{v11}{v22};
    \edge{v12}{v23};
    \edge[ultra thick]{v13}{v24};
    \edge[ultra thick]{v11}{v20};
    \edge{v12}{v21};
    \edge[ultra thick]{v13}{v22};
    \edge{v14}{v23};

    \edge[ultra thick]{v20}{v31};
    \edge{v21}{v32};
    \edge[ultra thick]{v22}{v33};
    \edge{v23}{v34};
    \edge{v21}{v30};
    \edge[ultra thick]{v22}{v31};
    \edge{v23}{v32};
    \edge[ultra thick]{v24}{v33};

    \edge{v30}{v41};
    \edge[ultra thick]{v31}{v42};
    \edge{v32}{v43};
    \edge{v33}{v44};
    \edge{v31}{v40};
    \edge{v32}{v41};
    \edge[ultra thick]{v33}{v42};
    \edge{v34}{v43};

    \edge{v00}{v40};
    \edge{v01}{v41};
    \edge{v02}{v42};
    \edge{v03}{v43};
    \edge{v04}{v44};

    \edge[ultra thick]{w01}{w23};
    \edge[ultra thick]{w10}{w32};
    \edge[ultra thick]{w02}{w20};
    \edge[ultra thick]{w13}{w31};

    \edge{w00}{w03};
    \edge{w10}{w13};
    \edge{w20}{w23};
    \edge{w30}{w33};

\foreach \x in {0,...,4}
    \foreach \y in {0,...,4}
        \vertex[orange]{\x}{\y}{};

\foreach \x in {0,...,3}
    \foreach \y in {0,...,3}
        \vertex[orange]{\x+0.5}{\y+0.5}{};

    \end{tikzpicture}
    \end{center}
    \caption{The $r \times r$ grid is a subgraph of this bounded degree $B_0$-VPG graph.}\label{fig:B0}
\end{figure}
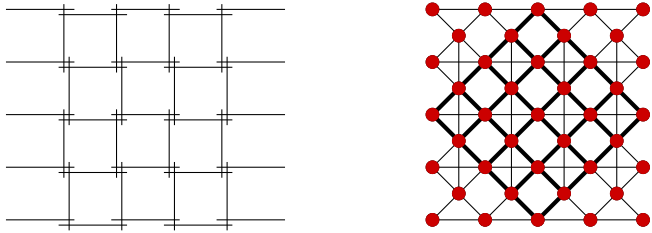

\begin{proposition}\label{prop:4t-not-vpg}
Not every 4-thin graph is a VPG graph.
\end{proposition}

\begin{proof}
The edge subdivision of the complete graph $K_5$ is 4-thin (see
Figure~\ref{fig:4-thin-no-vpg} for a representation).
Nevertheless, it is not a string graph~\cite{Sinden-K5}, and
string graphs are equivalent to VPG graphs~\cite{asinowski}.
\end{proof}

\begin{figure}
\begin{center}
\begin{tikzpicture}[scale=0.4,line cap=round,line join=round,>=triangle 45,x=1cm,y=1cm]
\fill[line width=1pt,color=black,fill=white,fill opacity=0]
(-25,-5) -- (-20,-5) -- (-18.454915028125264,-0.2447174185242338)
-- (-22.5,2.6942088429381323) --
(-26.545084971874736,-0.24471741852423157) -- cycle; \draw [line
width=1pt] (-8,-2)-- (-4,6); \draw [line width=1pt] (-4,6)--
(-6,3); \draw [line width=1pt] (-6,3)-- (-10,-1); \draw [line
width=1pt] (-10,-1)-- (-4,-8); \draw [line width=1pt] (-4,-8)--
(-10,-4.5); \draw [line width=1pt] (-10,-4.5)-- (-8,-4); \draw
[line width=1pt] (-8,-4)-- (-10,-5.5); \draw [line width=1pt]
(-10,-5.5)-- (-6,-6); \draw [line width=1pt] (-6,-6)-- (-4,1.15);
\draw [line width=1pt] (-4,1.15)-- (-8,-2); \draw [line width=1pt]
(-8,-2)-- (-4,5); \draw [line width=1pt] (-4,5)-- (-8,-4); \draw
[line width=1pt] (-8,-4)-- (-10,-3); \draw [line width=1pt]
(-10,-3)-- (-6,3); \draw [line width=1pt] (-6,3)-- (-4,4); \draw
[line width=1pt] (-4,4)-- (-6,-6); \draw [line width=1pt]
(-6,-6)-- (-10,-6.5); \draw [line width=1pt] (-10,-6.5)-- (-4,-8);
\draw [line width=1pt] (-4,-8)-- (-4,0); \draw [line width=1pt]
(-4,0)-- (-8,-2); \draw [line width=1pt] (-25,-5)-- (-20,-5);
\draw [line width=1pt] (-20,-5)--
(-18.454915028125264,-0.2447174185242338); \draw [line width=1pt]
(-18.454915028125264,-0.2447174185242338)--
(-22.5,2.6942088429381323); \draw [line width=1pt]
(-22.5,2.6942088429381323)--
(-26.545084971874736,-0.24471741852423157); \draw [line width=1pt]
(-26.545084971874736,-0.24471741852423157)-- (-25,-5); \draw [line
width=1pt] (-25,-5)-- (-18.454915028125264,-0.2447174185242338);
\draw [line width=1pt] (-18.454915028125264,-0.2447174185242338)--
(-26.545084971874736,-0.24471741852423157); \draw [line width=1pt]
(-26.545084971874736,-0.24471741852423157)-- (-20,-5); \draw [line
width=1pt] (-20,-5)-- (-22.5,2.6942088429381323); \draw [line
width=1pt] (-22.5,2.6942088429381323)-- (-25,-5); \draw
[fill=orange,color=orange] (-8,-2) circle (7pt); \draw
[fill=orange,color=orange] (-4,6) circle (7pt); \draw
[fill=orange,color=orange] (-6,3) circle (7pt); \draw
[fill=orange,color=orange] (-10,-1) circle (7pt); \draw
[fill=orange,color=orange] (-4,-8) circle (7pt); \draw
[fill=orange,color=orange] (-10,-4.5) circle (7pt); \draw
[fill=orange,color=orange] (-8,-4) circle (7pt); \draw
[fill=orange,color=orange] (-10,-5.5) circle (7pt); \draw
[fill=orange,color=orange] (-6,-6) circle (7pt); \draw
[fill=orange,color=orange] (-4,1.15) circle (7pt); \draw
[fill=orange,color=orange] (-4,5) circle (7pt); \draw
[fill=orange,color=orange] (-10,-3) circle (7pt); \draw
[fill=orange,color=orange] (-4,4) circle (7pt); \draw
[fill=orange,color=orange] (-10,-6.5) circle (7pt); \draw
[fill=orange,color=orange] (-4,0) circle (7pt); \draw
[fill=orange,color=orange] (-25,-5) circle (7pt); \draw
[fill=orange,color=orange] (-20,-5) circle (7pt); \draw
[fill=orange,color=orange]
(-18.454915028125264,-0.2447174185242338) circle (7pt); \draw
[fill=orange,color=orange] (-22.5,2.6942088429381323) circle
(7pt); \draw [fill=orange,color=orange]
(-26.545084971874736,-0.24471741852423157) circle (7pt); \draw
[fill=orange,color=orange]
(-24.522542485937368,1.2247457122069503) circle (7pt); \draw
[fill=orange,color=orange]
(-25.772542485937368,-2.622358709262116) circle (7pt); \draw
[fill=orange,color=orange] (-22.5,-5) circle (7pt); \draw
[fill=orange,color=orange]
(-19.227457514062632,-2.622358709262117) circle (7pt); \draw
[fill=orange,color=orange]
(-20.477457514062632,1.2247457122069494) circle (7pt); \draw
[fill=orange,color=orange] (-21.25,-1.1528955785309338) circle
(7pt); \draw [fill=orange,color=orange]
(-21.727457514062632,-2.622358709262117) circle (7pt); \draw
[fill=orange,color=orange]
(-23.272542485937368,-2.622358709262116) circle (7pt); \draw
[fill=orange,color=orange] (-22.5,-0.24471741852423268) circle
(7pt); \draw [fill=orange,color=orange]
(-23.75,-1.1528955785309338) circle (7pt);
\end{tikzpicture}
\end{center}
\caption{A 4-thin representation of the edge subdivision of
$K_5$.}\label{fig:4-thin-no-vpg}
\end{figure}
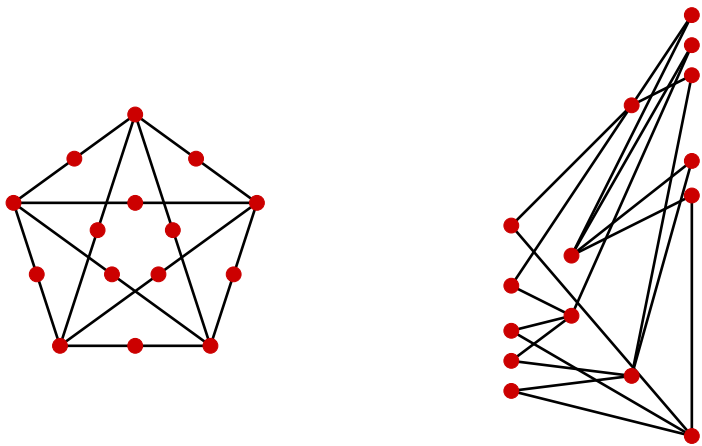

To prove that graphs with thinness at most three have bounded bend
number, we start by defining an intersection model obtained from
$\mathcal{M}_2$ by keeping from each rectangle the path formed by
the top and right sides (all the paths have the shape \Lxy).
Notice that two rectangles that are grounded to the $x$-axis
intersect if and only if their top sides intersect, and two
rectangles that are grounded to the $y$-axis intersect if and only
if their right sides intersect. Furthermore, a rectangle $X$
grounded to the $x$-axis intersects a rectangle $Y$ grounded to
the $y$-axis if and only if the right side of $X$ intersects the
top side of $Y$. So, both intersection models produce the same
graph. We can then reflect vertically and horizontally the model
in order to obtain the \Lr-model $\mathcal{M}_3$ (see
Figure~\ref{fig:M2-M3}).

\begin{proposition}\label{prop:b1-b0}
Every 2-thin graph is an \Lr-graph, thus a $B_1$-VPG graph.
Moreover, every  independent 2-thin graph is a $B_0$-VPG graph.
\end{proposition}

\begin{proof}
Every 2-thin graph admits the intersection model $\mathcal{M}_1$
(Lemma~\ref{lem:model1}). We have observed that the model
$\mathcal{M}_1$ can be modified to a grounded model
$\mathcal{M}_2$ and then to an \Lr-model $\mathcal{M}_3$
representing the same graph, so 2-thin graphs are \Lr-graphs, and
in particular $B_1$-VPG graphs. When the graph is independent
2-thin it is enough to keep the horizontal part for the paths that
are grounded to the $y$-axis and the vertical part for the paths
that are grounded to the $x$-axis, so independent 2-thin graphs
are $B_0$-VPG.
\end{proof}

Again, by adjusting the definitions of blocking, bi-semi-proper to
these 2-grounded models, we can prove characterizations analogous
to those in Theorems~\ref{thm:char-2-thin}
and~\ref{thm:char-2-pthin} for (independent) (proper) 2-thin
graphs. We will formalize one such characterization in
Theorem~\ref{thm:forb-pat-2-thin}.

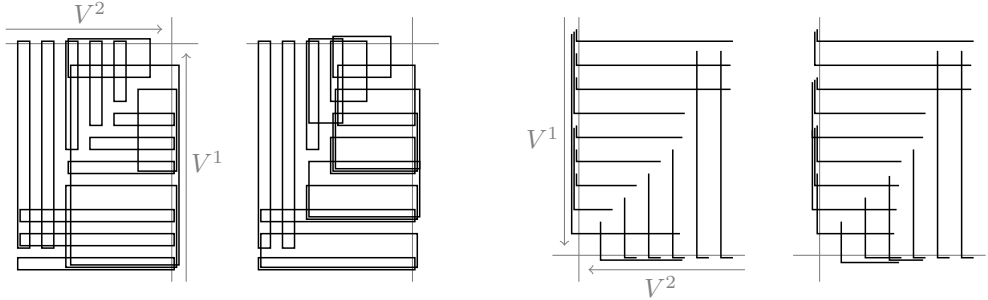
\begin{figure}

\noindent\begin{minipage}{\textwidth}
\begin{minipage}[t]{\dimexpr0.5\textwidth-0.5\Colsep\relax}
\begin{center}
\begin{tikzpicture}[scale=0.32]

\draw[color=gray,->] (0,10.5) -- (6.5,10.5);

\draw[color=gray,->] (7.5,0) -- (7.5,9.5);

\node at (3.5,11.2) {\footnotesize \textcolor{gray}{$V^2$}};

\node at (8.4,5) {\footnotesize \textcolor{gray}{$V^1$}};

\draw[color=white] (0,-2) -- (1,-2);

\draw[color=white] (0,13) -- (1,13);

\draw[color=gray] (0,9.9) -- (8,9.9);

\draw[color=gray] (6.9,0) -- (6.9,11);

\foreach \i / \x / \y / \d in
{1/1/0/0,2/2/0.1/0,3/3/0.1/0,4/1.1/2/1,5/5/2.1/0,6/6/3/0,7/7/4/0,8/5.1/5/1,9/1.2/2.2/2}
    \draw (\y+0.5,\x-0.5) [c2] rectangle  (7+0.1*\d,\i);

\foreach \i / \x / \y / \d in {1/0.9/1/0, 2/0.9/2/0, 3/5/3/0,
4/6/4/0, 5/7/5/0, 6/8/3.1/1}
    \draw (\y-0.5,\x+0.5) [c2] rectangle  (\i,10+0.1*\d);




\draw[color=gray] (10,9.9) -- (18,9.9);

\draw[color=gray] (16.9,0) -- (16.9,11);

\foreach \i / \x / \y / \d in
{1/1/0/0,2/1.1/0.1/1,3/3/0.1/0,4/3.1/2/1,5/3.2/2.1/2,6/5/3/0,7/5.1/3.1/1,8/5.2/3.2/2,9/7/3.3/0}
    \draw (10+\y+0.5,\x-0.5) [c2] rectangle  (17+0.1*\d, \i);

\foreach \i / \x / \y / \d in {1/0.9/1/0, 2/0.9/2/0, 3/5/3/0,
4/6.1/3.1/1, 5/7/4/0, 6/8/4.1/2}
    \draw (10+\y-0.5,\x+0.5) [c2] rectangle  (10+\i,10+0.1*\d);



\end{tikzpicture}
\end{center}
\end{minipage}\hfill
\begin{minipage}[t]{\dimexpr0.5\textwidth-0.5\Colsep\relax}
\begin{center}
\begin{tikzpicture}[scale=0.32,xscale=-1,yscale=-1]

\draw[color=gray,->] (10,10.5) -- (16.5,10.5);

\draw[color=gray,->] (17.5,0) -- (17.5,9.5);

\node at (13.5,11.2) {\footnotesize \textcolor{gray}{$V^2$}};

\node at (18.4,5) {\footnotesize \textcolor{gray}{$V^1$}};

\draw[color=white] (0,-2) -- (1,-2);

\draw[color=white] (0,13) -- (1,13);

\draw[color=gray] (0,9.9) -- (8,9.9);

\draw[color=gray] (6.9,0) -- (6.9,11);

\draw[color=gray] (10,9.9) -- (18,9.9);

\draw[color=gray] (16.9,0) -- (16.9,11);

\foreach \i / \x / \y / \d in
{1/1/0/0,2/2/0.1/0,3/3/0.1/0,4/1.1/2/1,5/5/2.1/0,6/6/3/0,7/7/4/0,8/5.1/5/1,9/1.2/2.2/2}
    \draw (10+\y+0.5,\i) [c2] --  (17+0.1*\d,\i) -- (17+0.1*\d,\x-0.5);

\foreach \i / \x / \y / \d in {1/0.9/1/0, 2/0.9/2/0, 3/5/3/0,
4/6/4/0, 5/7/5/0, 6/8/3.1/1}
    \draw (10+\i,\x+0.5) [c2] --  (10+\i,10+0.1*\d) -- (10+\y-0.5,10+0.1*\d);




\foreach \i / \x / \y / \d in
{1/1/0/0,2/1.1/0.1/1,3/3/0.1/0,4/3.1/2/1,5/3.2/2.1/2,6/5/3/0,7/5.1/3.1/1,8/5.2/3.2/2,9/7/3.3/0}
    \draw (\y+0.5,\i) [c2] --  (7+0.1*\d, \i) -- (7+0.1*\d,\x-0.5);

\foreach \i / \x / \y / \d in {1/0.9/1/0, 2/0.9/2/0, 3/5/3/0,
4/6.1/3.1/1, 5/7/4/0, 6/8/4.1/2}
    \draw (\i,\x+0.5) [c2] -- (\i,10+0.1*\d) -- (\y-0.5,10+0.1*\d);



\end{tikzpicture}
\end{center}
\end{minipage}%
\end{minipage}

\caption{The model $\mathcal{M}_2$ as intersection of grounded
rectangles (left) and the \Lr-model $\mathcal{M}_3$ (right) for
the graphs in Figures~\ref{fig:M1} and~\ref{fig:M1-prop}.
}\label{fig:M2-M3}
\end{figure}

For 3-thin graphs, we use the previous ideas for the intersections
within each class and between each pair of classes.

\begin{proposition}\label{prop:b3-b1}
Every 3-thin graph is a $B_3$-VPG graph. Moreover, every
independent 3-thin graph is a $B_1$-VPG graph.
\end{proposition}

\begin{proof}
Let $G$ be a 3-thin graph with partition $V^1,V^2,V^3$. We start
by constructing the \Lr-model $\mathcal{M}_3$ for $G[V^1 \cup
V^2]$. Then we extend the paths corresponding to the vertices of
$V^1$ and the paths corresponding to the vertices of $V^2$ in
order to allow the respective intersections with the paths
corresponding to the vertices of $V^3$ in two different sectors of
the plane. These intersections also follow the ideas of the model
$\mathcal{M}_3$ for $G[V^1 \cup V^3]$ and $G[V^2 \cup V^3]$,
respectively. In the case of independent 3-thin graphs, less bends
per path are necessary because there are no internal intersections
on each class. A sketch of a $B_3$-VPG model for 3-thin graphs and
a $B_1$-VPG model for independent 3-thin graphs can be found in
Figure~\ref{fig:3-thin}.
\end{proof}

\begin{proposition}\label{prop:not-b0}
There are 2-thin graphs and independent 3-thin graphs that are not
$B_0$-VPG.
\end{proposition}

\begin{proof}
The 4-wheel $W_4$ (obtained from a 4-cycle by adding a universal
vertex) is 2-thin and proper independent 3-thin but not $B_0$-VPG,
since in a $B_0$-VPG graph, for every vertex $v$, $N[v]$ induces
an interval graph~\cite{asinowski}.
\end{proof}

So, the bound of the bend number in the cases of 2-thin and
independent 3-thin graphs is tight. We conjecture that also the 3
bends bound is tight for 3-thin graphs, but we are missing an
example. We can prove, however, the following.

\begin{proposition}\label{prop:not-L}
There are independent 3-thin graphs that are not monotone
\Lr-graphs.
\end{proposition}

\begin{proof}
The octahedron $\overline{3K_2}$ is an example of a graph of
independent thinness~3 which is not a p-box since it has
boxicity~3~\cite{Rob-box}.
\end{proof}

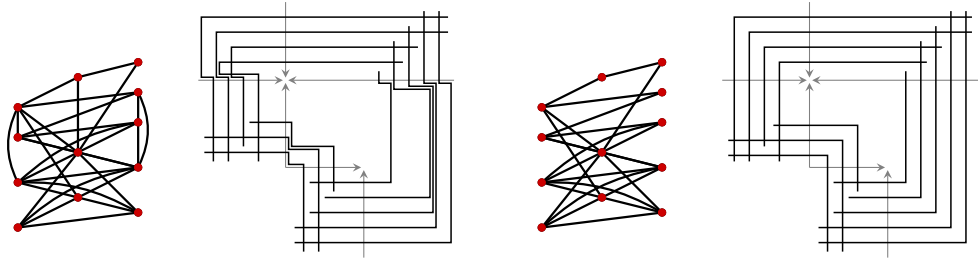
\begin{figure}
\noindent\begin{minipage}{\textwidth}
\begin{minipage}[t]{\dimexpr0.5\textwidth-0.5\Colsep\relax}
\begin{center}
    \begin{tikzpicture}[scale=0.4]

\foreach \x in {0.5} {
    \vertex{-6}{0+4}{v1};
    \vertex{-6}{3*\x+4}{v2};
    \vertex{-6}{6*\x+4}{v3};
    \vertex{-6}{8*\x+4}{v4};

    \vertex{-4}{2*\x+4}{w1};
    \vertex{-4}{5*\x+4}{w2};
    \vertex{-4}{10*\x+4}{w3};

    \vertex{-2}{1*\x+4}{z1};
    \vertex{-2}{4*\x+4}{z2};
    \vertex{-2}{7*\x+4}{z3};
    \vertex{-2}{9*\x+4}{z4};
    \vertex{-2}{11*\x+4}{z5};
}

\path (v4) edge [e1,bend right=25] (v2); \path (v4) edge [e1]
(v3);

\path (w3) edge [e1] (w2);

\path (z4) edge [e1,bend left=25] (z2); \path (z4) edge [e1] (z3);
\path (z3) edge [e1] (z2);

\path (v2) edge [e1] (w1); \path (v3) edge [e1] (w2); \path (v4)
edge [e1] (w1); \path (v4) edge [e1] (w2); \path (v1) edge [e1]
(w2); \path (v4) edge [e1] (w3); \path (v2) edge [e1] (w2); \path
(v1) edge [e1] (w1);

\path (v1) edge [e1] (z1); \path (v2) edge [e1] (z2); \path (v2)
edge [e1,bend left=15] (z1); \path (v1) edge [e1,bend left=15]
(z2); \path (v2) edge [e1,bend left=15] (z3); \path (v4) edge [e1]
(z4); \path (v3) edge [e1] (z4); \path (v3) edge [e1] (z2); \path
(v3) edge [e1] (z3);

\path (z2) edge [e1] (w1); \path (z3) edge [e1] (w2); \path (z5)
edge [e1] (w3); \path (z5) edge [e1] (w2); \path (z2) edge [e1]
(w2); \path (z1) edge [e1] (w1); \path (z1) edge [e1] (w2);

\foreach \x in {0.5} {
    \vertex[orange]{-6}{0+4}{};
    \vertex[orange]{-6}{3*\x+4}{};
    \vertex[orange]{-6}{6*\x+4}{};
    \vertex[orange]{-6}{8*\x+4}{};

    \vertex[orange]{-4}{2*\x+4}{};
    \vertex[orange]{-4}{5*\x+4}{};
    \vertex[orange]{-4}{10*\x+4}{};

    \vertex[orange]{-2}{1*\x+4}{};
    \vertex[orange]{-2}{4*\x+4}{};
    \vertex[orange]{-2}{7*\x+4}{};
    \vertex[orange]{-2}{9*\x+4}{};
    \vertex[orange]{-2}{11*\x+4}{};
}

\draw[color=gray,->,> = stealth, shorten > = 1pt, auto] (0,8.9) ->
(2.9,8.9); \draw[color=gray,->,> = stealth, shorten > = 1pt, auto]
(8.5,8.9) -> (2.9,8.9); \draw[color=gray,->,> = stealth, shorten >
= 1pt, auto] (2.9,11.5) -> (2.9,8.9); \draw[color=gray,->,> =
stealth, shorten > = 1pt, auto] (2.9,6) -> (2.9,8.9);
\draw[color=gray,->,> = stealth, shorten > = 1pt, auto] (2.9,6) ->
(5.5,6);
\draw[color=gray,->,> = stealth, shorten > = 1pt, auto] (5.5,3) ->
(5.5,6);

\foreach \h / \p / \c in {12/3/0} \foreach \i / \x / \y / \z / \d
in {1/1/0/0/0,2/2/0/0/0,3/3/1/2/0,4/2/0/3/1}
    \draw (0.5*\i,\h-6+0.5*\y+0.2) [c2] --  (0.5*\i,\h-3+0.1*\d) -- (0.5*\x-0.4+0.1*\d,\h-3+0.1*\d) -- (0.5*\x-0.4+0.1*\d,\h-0.5-0.5*\i) -- (8.5-0.5*\z-0.2,\h-0.5-0.5*\i);

\foreach \h / \p / \c in {3/3/6} \foreach \i / \x / \y / \z / \d
in {1/0/1/0/0, 2/0/2/0/0, 3/3/2/4/1}
    \draw (0.5*\x+0.2,\c+0.5*\i) [c2] --  (\h+0.1*\d,\c+0.5*\i) -- (\h+0.1*\d,\c+0.5*\y-0.4+0.1*\d) -- (\h+0.5*\i,\c+0.5*\y-0.4+0.1*\d) -- (\h+0.5*\i,\p+0.5*\z+0.2);

\foreach \h / \p / \c in {11.8/3/0} \foreach \i / \x / \y / \z /
\d in {1/0/0/1/0,2/0/0/2/0,3/1/1/2/1,4/2/2/2/2,5/4/1/5/0}
    \draw (8.5-0.5*\i,\h-1-0.5*\x+0.4) [c2] --  (8.5-0.5*\i,\h-3-0.1*\d) -- (8.5-0.5*\z+0.4-0.1*\d,\h-3-0.1*\d)  -- (8.5-0.5*\z+0.4-0.1*\d,3+0.5*\i) -- (3+0.5*\y+0.2,3+0.5*\i) ;

    \end{tikzpicture}
    \end{center}
\end{minipage}\hfill
\begin{minipage}[t]{\dimexpr0.5\textwidth-0.5\Colsep\relax}
    \begin{center}
    \begin{tikzpicture}[scale=0.4]

\foreach \x in {0.5} {
    \vertex{-6}{0+4}{v1};
    \vertex{-6}{3*\x+4}{v2};
    \vertex{-6}{6*\x+4}{v3};
    \vertex{-6}{8*\x+4}{v4};

    \vertex{-4}{2*\x+4}{w1};
    \vertex{-4}{5*\x+4}{w2};
    \vertex{-4}{10*\x+4}{w3};

    \vertex{-2}{1*\x+4}{z1};
    \vertex{-2}{4*\x+4}{z2};
    \vertex{-2}{7*\x+4}{z3};
    \vertex{-2}{9*\x+4}{z4};
    \vertex{-2}{11*\x+4}{z5};
}




\path (v2) edge [e1] (w1); \path (v3) edge [e1] (w2); \path (v4)
edge [e1] (w1); \path (v4) edge [e1] (w2); \path (v1) edge [e1]
(w2); \path (v4) edge [e1] (w3); \path (v2) edge [e1] (w2); \path
(v1) edge [e1] (w1);

\path (v1) edge [e1] (z1); \path (v2) edge [e1] (z2); \path (v2)
edge [e1,bend left=15] (z1); \path (v1) edge [e1,bend left=15]
(z2); \path (v2) edge [e1,bend left=15] (z3); \path (v4) edge [e1]
(z4); \path (v3) edge [e1] (z4); \path (v3) edge [e1] (z2); \path
(v3) edge [e1] (z3);

\path (z2) edge [e1] (w1); \path (z3) edge [e1] (w2); \path (z5)
edge [e1] (w3); \path (z5) edge [e1] (w2); \path (z2) edge [e1]
(w2); \path (z1) edge [e1] (w1); \path (z1) edge [e1] (w2);

\foreach \x in {0.5} {
    \vertex[orange]{-6}{0+4}{};
    \vertex[orange]{-6}{3*\x+4}{};
    \vertex[orange]{-6}{6*\x+4}{};
    \vertex[orange]{-6}{8*\x+4}{};

    \vertex[orange]{-4}{2*\x+4}{};
    \vertex[orange]{-4}{5*\x+4}{};
    \vertex[orange]{-4}{10*\x+4}{};

    \vertex[orange]{-2}{1*\x+4}{};
    \vertex[orange]{-2}{4*\x+4}{};
    \vertex[orange]{-2}{7*\x+4}{};
    \vertex[orange]{-2}{9*\x+4}{};
    \vertex[orange]{-2}{11*\x+4}{};
}

\draw[color=gray,->,> = stealth, shorten > = 1pt, auto] (0,8.9) ->
(2.9,8.9); \draw[color=gray,->,> = stealth, shorten > = 1pt, auto]
(8.5,8.9) -> (2.9,8.9); \draw[color=gray,->,> = stealth, shorten >
= 1pt, auto] (2.9,11.5) -> (2.9,8.9); \draw[color=gray,->,> =
stealth, shorten > = 1pt, auto] (2.9,6) -> (2.9,8.9);
\draw[color=gray,->,> = stealth, shorten > = 1pt, auto] (2.9,6) ->
(5.5,6);
\draw[color=gray,->,> = stealth, shorten > = 1pt, auto] (5.5,3) ->
(5.5,6);

\foreach \h / \p / \c in {12/3/0} \foreach \i / \x / \y / \z / \d
in {1/1/0/0/0,2/2/0/0/0,3/3/1/2/0,4/2/0/3/1} {
    \draw (0.5*\i-0.1,\h-6+0.5*\y+0.2) [c2]  -- (0.5*\i-0.1,\h-0.5-0.5*\i) -- (8.5-0.5*\z-0.2,\h-0.5-0.5*\i);
}

\foreach \h / \p / \c in {3/3/6} \foreach \i / \x / \y / \z / \d
in {1/0/1/0/0, 2/0/2/0/0, 3/3/2/4/1} {
    \draw (0.5*\x+0.2,\c+0.5*\i-0.1) [c2] -- (\h+0.5*\i,\c+0.5*\i-0.1) -- (\h+0.5*\i,\p+0.5*\z+0.2);
}

\foreach \h / \p / \c in {11.8/3/0} \foreach \i / \x / \y / \z /
\d in {1/0/0/1/0,2/0/0/2/0,3/1/1/2/1,4/2/2/2/2,5/4/1/5/0} {
    \draw (8.5-0.5*\i+0.1,\h-1-0.5*\x+0.4) [c2] -- (8.5-0.5*\i+0.1,3+0.5*\i) -- (3+0.5*\y+0.2,3+0.5*\i) ;

}

    \end{tikzpicture}
\end{center}
\end{minipage}%
\end{minipage}

    \caption{A $B_3$-VPG representation of a 3-thin graph (left) and a $B_1$-VPG representation of an independent 3-thin graph (right). }\label{fig:3-thin}
\end{figure}

Further modifying the model $\mathcal{M}_3$ by translating and
extending vertically the paths that are grounded to the $x$-axis
and horizontally the paths that are grounded to the $y$-axis in
order to have the bends lying in the inverted diagonal $y = -x$,
we can obtain a monotone \Lr-model $\mathcal{M}_4$ for every
2-thin graph. Concretely, the path $(x,y)$--$(x,0)$--$(z,0)$ ($0 <
x < z$, $0 < y$) becomes $(x,y)$--$(x,-x)$--$(z,-x)$, and the path
$(0,z)$--$(0,y)$--$(x,y)$ ($0 < y < z$, $0 < x$) becomes
$(-y,z)$--$(-y,y)$--$(x,y)$ (see Figure~\ref{fig:M3-M4}). In this
way, the intersections within the first quadrant are maintained;
two paths $(x,y)$--$(x,-x)$--$(z,-x)$ and
$(x',y')$--$(x',-x')$--$(z',-x')$ with $0 < x < x'$ intersect (at
$(x',-x)$) if and only if $x' < z$, if and only if the original
paths $(x,y)$--$(x,0)$--$(z,0)$ and $(x',y')$--$(x',0)$--$(z',0)$
intersect; two paths $(-y,z)$--$(-y,y)$--$(x,y)$ and
$(-y',z')$--$(-y',y')$--$(x',y')$ with $0 < y < y'$ intersect (at
$(-y,y')$) if and only if $y' < z$, if and only if the original
paths $(0,z)$--$(0,y)$--$(x,y)$ and $(0,z')$--$(0,y')$--$(x',y')$
intersect. So, both intersection models produce the same graph.

We will prove in the next section that 2-thin graphs are exactly
the blocking monotone \Lr-graphs. An \Lr-model is \emph{blocking}
if for every two non-intersecting \Lr's, either the vertical or
the horizontal prolongation of one of them intersects the other.

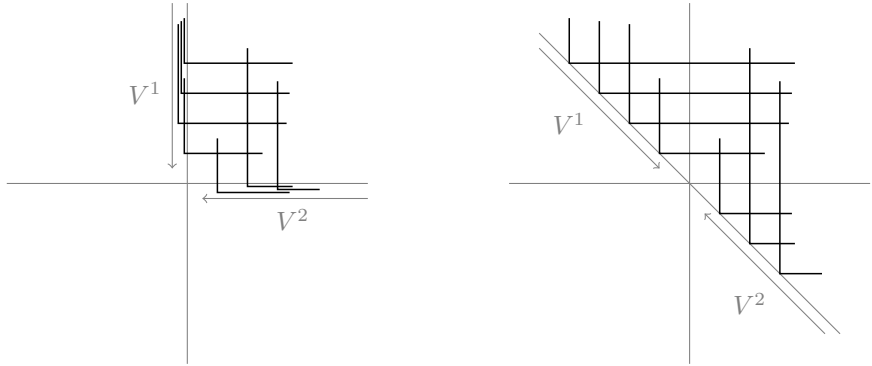
\begin{figure}

\noindent\begin{minipage}{\textwidth}
\begin{minipage}[t]{\dimexpr0.5\textwidth-0.5\Colsep\relax}
\begin{center}
\begin{tikzpicture}[scale=0.4,xscale=-1,yscale=-1]

\draw[color=gray,->] (1,10.5) -- (6.5,10.5);

\draw[color=gray,->] (7.5,4) -- (7.5,9.5);

\node at (3.5,11.2) {\footnotesize \textcolor{gray}{$V^2$}};

\node at (8.4,7) {\footnotesize \textcolor{gray}{$V^1$}};

\draw[color=gray] (1,10) -- (13,10);

\draw[color=gray] (7,4) -- (7,16);

\foreach \i / \x / \y / \d in
{6/5/3/1,7/5.1/3.1/2,8/5.2/3.2/3,9/7/4/1}
    \draw (\y+0.5,\i) [c2] --  (7+0.1*\d, \i) -- (7+0.1*\d,\x-0.5);

\foreach \i / \x / \y / \d in {4/6.1/3.1/2, 5/5/4/1, 6/8/4.1/3}
    \draw (\i,\x+0.5) [c2] -- (\i,10+0.1*\d) -- (\y-0.5,10+0.1*\d);

\end{tikzpicture}
\end{center}
\end{minipage}%
\begin{minipage}[t]{\dimexpr0.5\textwidth-0.5\Colsep\relax}
\begin{center}
\begin{tikzpicture}[scale=0.4]

\draw[color=gray,->] (-5,4.5) -- (-1,0.5);

\draw[color=gray,->] (4.5,-5) -- (0.5,-1);

\node at (2,-4) {\footnotesize \textcolor{gray}{$V^2$}};

\node at (-4,2) {\footnotesize \textcolor{gray}{$V^1$}};

\draw[color=gray] (-6,0) -- (6,0);

\draw[color=gray] (0,-6) -- (0,6);

\draw[color=gray] (-5,5) -- (5,-5);

\foreach \i / \x / \y / \d in
{6/5/3/0,7/5.1/3.1/1,8/5.2/3.2/2,9/7/4/0}
    \draw (7-\y-0.5,10-\i) [c2] --  (\i-10,10-\i) -- (\i-10,10-\x+0.5);

\foreach \i / \x / \y / \d in {4/6.1/3.1/1, 5/5/4/0, 6/8/4.1/2}
    \draw (7-\i,10-\x-0.5) [c2] -- (7-\i,\i-7) -- (7-\y+0.5,\i-7);

\end{tikzpicture}
\end{center}
\end{minipage}\hfill
\end{minipage}

\caption{The \Lr-model $\mathcal{M}_3$ (left) and the monotone
\Lr-model $\mathcal{M}_4$ (right) for the same graph.
}\label{fig:M3-M4}
\end{figure}

\section{Characterization by forbidden patterns}\label{sec:pat}

A \emph{trigraph} $T$ is a 4-tuple $(V(T),E(T),N(T),U(T))$ where
$V(T)$ is the vertex set and every unordered pair of vertices
belongs to one of the three disjoint sets $E(T)$, $N(T)$, and
$U(T)$ called respectively \emph{edges}, \emph{non-edges}, and
\emph{undecided edges}. A graph $G = (V(G),E(G))$ is a
\emph{realization} of a trigraph $T$ if $V(G) = V(T)$ and $E(G) =
E(T) \cup U'$, where $U' \subseteq U(T)$. When representing a
trigraph, we will draw solid lines for edges, dotted lines for non
edges, and nothing for undecided edges. As $(E(T),N(T),U(T))$ is a
partition of the unordered pairs, it is enough to give any two of
these sets to define the trigraph, and we will often define a
trigraph by giving only $E$ and $N$.

An \emph{ordered graph} is a graph given with a linear ordering of
its vertices. We define the same for a trigraph, and call it a
\emph{pattern}. We say that an ordered graph is a
\emph{realization} of a pattern if they share the same set of
vertices and linear ordering and the graph is a realization of the
trigraph. When, in an ordered graph, no ordered subgraph is the
realization of given pattern, we say that the ordered graph
\emph{avoids} the pattern. The \emph{mirror} or \emph{reverse} of
a pattern is the same pattern, except the ordering, which is
reversed.

Given a family of patterns $\mathcal{F}$, the class
\textsc{Ord}$(\mathcal{F})$ is the set of graphs that have the
following property: there exists an ordering of the nodes, such
that none of the ordered subgraphs is a realization of a pattern
in $\mathcal{F}$, i.e., the ordered graph avoids all the patterns
in $\mathcal{F}$.

Many natural graph classes can be described as
\textsc{Ord}$(\mathcal{F})$ for sets $\mathcal{F}$ of small
patterns~\cite{Ch-perf-ord,Chap-MPT,Damas-order,Habib-patterns,H-M-R-orders}.
For instance, for $Z_1 = (\{1,2,3\},\{13\},\{23\})$ and its mirror
$Z_2 = (\{1,2,3\},\{13\},\{12\})$, the class of interval graphs is
\textsc{Ord}$(\{Z_1\})$~\cite{Ola-interval} and the class of
proper interval graphs is
\textsc{Ord}$(\{Z_1,Z_2\})$~\cite{Rob-uig}, and for $Z_3 =
(\{1,2,3\},\{13\},\{12,23\})$ and $Z_4 =
(\{1,2,3\},\{12,13,23\},\emptyset)$, the class of bipartite
permutation graphs is
\textsc{Ord}$(\{Z_3,Z_4\})$~\cite{Habib-patterns}. For $Z_6 =
(\{1,2,3,4\}, \{13,24\}, \{23\})$, the class of monotone
\Lr-graphs is \textsc{Ord}$(\{Z_6\})$~\cite{Chap-MPT}. See
Figure~\ref{fig:patterns} for a graphical representation of the
aforementioned patterns. Moreover, a recent paper studies
systematically forbidden pattern characterizations of graph
classes defined by grounded intersection
models~\cite{Habib-patterns-ground-arx}.

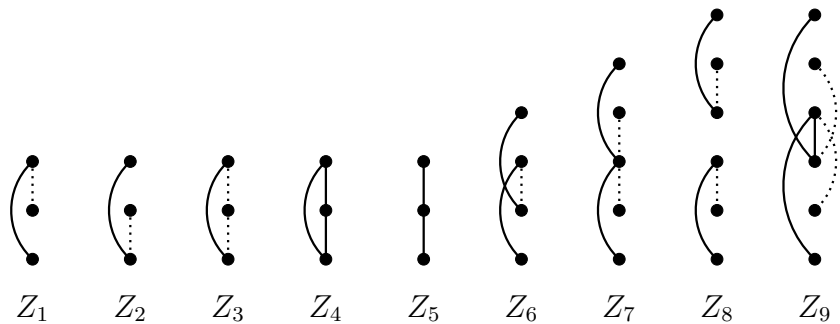
\begin{figure}
    \begin{center}
    \begin{tikzpicture}[scale=0.65]

    \vertex{0}{0}{v0};
    \vertex{0}{1}{v1};
    \vertex{0}{2}{v2};

    \vertex{2}{0}{w0};
    \vertex{2}{1}{w1};
    \vertex{2}{2}{w2};

    \vertex{4}{0}{x0};
    \vertex{4}{1}{x1};
    \vertex{4}{2}{x2};

    \vertex{6}{0}{y0};
    \vertex{6}{1}{y1};
    \vertex{6}{2}{y2};

    \vertex{8}{0}{z0};
    \vertex{8}{1}{z1};
    \vertex{8}{2}{z2};

\path (v0) edge [e1,bend left=45] (v2); \path (w0) edge [e1,bend
left=45] (w2);

\path (v1) edge [e1,dotted] (v2); \path (w0) edge [e1,dotted]
(w1);

\path (x0) edge [e1,bend left=45] (x2); \path (y0) edge [e1,bend
left=45] (y2);

\path (x1) edge [e1,dotted] (x2); \path (x0) edge [e1,dotted]
(x1); \path (y0) edge [e1] (y1); \path (y1) edge [e1] (y2); \path
(z0) edge [e1] (z1); \path (z1) edge [e1] (z2);

    \node at (0,-1) {$Z_1$};
    \node at (2,-1) {$Z_2$};
    \node at (4,-1) {$Z_3$};
    \node at (6,-1) {$Z_4$};
    \node at (8,-1) {$Z_5$};

    \vertex{10}{0}{a0};
    \vertex{10}{1}{a1};
    \vertex{10}{2}{a2};
    \vertex{10}{3}{a3}:

    \vertex{12}{0}{b1};
    \vertex{12}{1}{b2};
    \vertex{12}{2}{b3};
    \vertex{12}{3}{b4};
    \vertex{12}{4}{b5};

    \vertex{14}{0}{c1};
    \vertex{14}{1}{c2};
    \vertex{14}{2}{c3};
    \vertex{14}{3}{c4};
    \vertex{14}{4}{c5};
    \vertex{14}{5}{c6};

    \vertex{16}{0}{d1};
    \vertex{16}{1}{d2};
    \vertex{16}{2}{d3};
    \vertex{16}{3}{d4};
    \vertex{16}{4}{d5};
    \vertex{16}{5}{d6};

\path (a0) edge [e1,bend left=45] (a2); \path (a1) edge [e1,bend
left=45] (a3); \path (a1) edge [e1,dotted] (a2);

\path (b1) edge [e1,bend left=45] (b3); \path (b3) edge [e1,bend
left=45] (b5);

\path (b2) edge [e1,dotted] (b3); \path (b3) edge [e1,dotted]
(b4);

\path (c1) edge [e1,bend left=45] (c3); \path (c4) edge [e1,bend
left=45] (c6);

\path (c2) edge [e1,dotted] (c3); \path (c4) edge [e1,dotted]
(c5);

\path (d1) edge [e1,bend left=45] (d4); \path (d3) edge [e1,bend
left=45] (d6); \path (d3) edge [e1] (d4);

\path (d2) edge [e1,dotted,bend right=45] (d4); \path (d3) edge
[e1,dotted,bend right=45] (d5);

    \node at (10,-1) {$Z_6$};
    \node at (12,-1) {$Z_7$};
    \node at (14,-1) {$Z_8$};
    \node at (16,-1) {$Z_9$};

    \end{tikzpicture}
    \end{center}
    \caption{The patterns used in the characterizations of this section. The solid lines denote compulsory edges and the dotted lines are
compulsory non-edges in the pattern.}\label{fig:patterns}
\end{figure}

The class of bipartite graphs can be characterized as
\textsc{Ord}$(\{Z_5\})$, where $Z_5 =
(\{1,2,3\},\{12,23\},\emptyset)$~\cite{Habib-patterns}. In the
literature, there are two additional ways of defining patterns for
subclasses of bipartite graphs. The first one~\cite{H-H-bigraphs}
involves a total order of the vertices, that are colored black or
white, a set of compulsory edges, and a set of compulsory
non-edges. We will call such a structure a \emph{bicolored
pattern}, and we will describe it as a 4-tuple containing a set of
ordered vertices, the subset of white vertices, the set of edges
and the set of non-edges. For instance, the paths in
Figure~\ref{fig:bicol-patterns} are described as $Q_1 =
(\{1,2,3\}, \{3\}, \{13\}, \{23\})$, $Q_2 = (\{1,2,3\}, \{1,2\},
\{13\}, \{23\})$, $Q_3 = (\{1,2,3\}, \{1\}, \{13\}, \{12\})$, $Q_4
= (\{1,2,3\}, \{2,3\}, \{13\}, \{12\})$.

We say that a bipartite graph $H$ belongs to
\textsc{BicolOrd}$(\mathcal{F})$, for a fixed family of bicolored
patterns $\mathcal{F}$, if $H$ admits a bipartition $V(H) = A \cup
B$ and an ordering of $A \cup B$ that avoids the patterns from
$\mathcal{F}$. It was proved in~\cite{H-H-bigraphs} that interval
bigraphs are exactly \textsc{BicolOrd}$(\{Q_1,Q_2\})$.

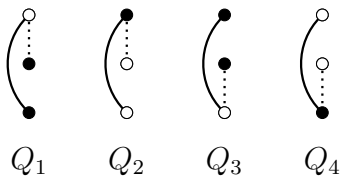
\begin{figure}
    \begin{center}
    \begin{tikzpicture}[scale=0.65]

    \vertex{0}{0}{v0};
    \vertex{0}{1}{v1};
    \vertex{0}{2}{v2};

    \vertex{2}{0}{w0};
    \vertex{2}{1}{w1};
    \vertex{2}{2}{w2};

    \vertex{4}{0}{x0};
    \vertex{4}{1}{x1};
    \vertex{4}{2}{x2};

    \vertex{6}{0}{y0};
    \vertex{6}{1}{y1};
    \vertex{6}{2}{y2};

\path (v0) edge [e1,bend left=45] (v2); \path (w0) edge [e1,bend
left=45] (w2); \path (x0) edge [e1,bend left=45] (x2); \path (y0)
edge [e1,bend left=45] (y2);

\path (v1) edge [e1,dotted] (v2); \path (w1) edge [e1,dotted]
(w2);

\path (x0) edge [e1,dotted] (x1); \path (y0) edge [e1,dotted]
(y1);

    \wvertex{0}{2}{};

    \wvertex{2}{0}{};
    \wvertex{2}{1}{};

    \wvertex{4}{0}{};

    \wvertex{6}{1}{};
    \wvertex{6}{2}{};

    \node at (0,-1) {$Q_1$};
    \node at (2,-1) {$Q_2$};
    \node at (4,-1) {$Q_3$};
    \node at (6,-1) {$Q_4$};
    \end{tikzpicture}
    \end{center}
    \caption{Examples of bicolored patterns.}\label{fig:bicol-patterns}
\end{figure}

The other (slightly different) way~\cite{H-M-R-orders} is the
following. A \emph{bipartite pattern} is a bipartite trigraph
whose vertices in each part of the bipartition are linearly
ordered. We will denote such pattern as a 4-tuple containing two
disjoint sets of ordered vertices, the set of edges and the set of
non-edges. For instance, the paths in Figure~\ref{fig:bip-pat} are
described as $R_1 = (\{1,2\}, \{1',2'\}, \{12',21'\}, \{11'\})$,
$R_2 = (\{1,2\}, \{1',2'\}, \{12',21'\}, \{22'\})$, $R_3 =
(\{1,2,3\}, \{1',2',3'\}, \{13',31',33'\}, \{23',32'\})$,
$R_4=(\{1,2,3\}, \{1'\}, \{11',31'\}, \{21'\})$, $R_4'=(\{1\},
\{1',2',3'\}, \{11',13'\}, \{12'\})$.

We say that a bipartite graph $H$ belongs to
\textsc{BiOrd}$(\mathcal{F})$, for a fixed family of bipartite
patterns $\mathcal{F}$, if $H$ admits a bipartition $V(H) = A \cup
B$ and an ordering of $A$ and of $B$ so that no pattern from
$\mathcal{F}$ occurs. Several known bipartite graph classes can be
characterized as \textsc{BiOrd}$(\mathcal{F})$. For instance,
bipartite convex graphs are \textsc{BiOrd}$(\{R_4\})$, and proper
interval bigraphs are
\textsc{BiOrd}$(\{R_1,R_2\})$~\cite{H-M-R-orders}.

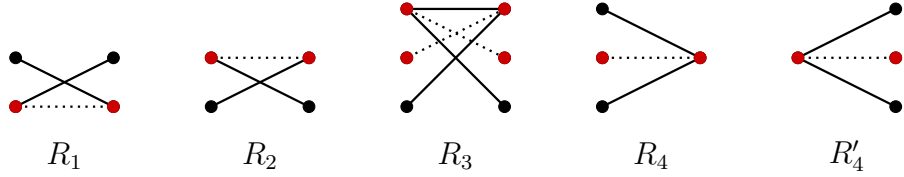
\begin{figure}
    \begin{center}
    \begin{tikzpicture}[scale=0.65]

\foreach \x in {0,...,6,9}
    \foreach \y in {0,1}
        \vertex{2*\x}{\y}{v\x\y};

\foreach \x in {7,8}
    \foreach \y in {1}
        \vertex{2*\x}{\y}{v\x\y};

\foreach \x in {4,5,6,9}
    \foreach \y in {2}
        \vertex{2*\x}{\y}{v\x\y};

\path (v00) edge [e1] (v11); \path (v01) edge [e1] (v10); \path
(v00) edge [e1,dotted] (v10);

\path (v20) edge [e1] (v31); \path (v21) edge [e1] (v30); \path
(v21) edge [e1,dotted] (v31);

\path (v40) edge [e1] (v52); \path (v42) edge [e1] (v50); \path
(v42) edge [e1] (v52); \path (v41) edge [e1,dotted] (v52); \path
(v42) edge [e1,dotted] (v51);

\path (v60) edge [e1] (v71); \path (v62) edge [e1] (v71); \path
(v61) edge [e1,dotted] (v71);

\path (v90) edge [e1] (v81); \path (v92) edge [e1] (v81); \path
(v91) edge [e1,dotted] (v81);

    \vertex[orange]{0}{0}{};
    \vertex[orange]{2}{0}{};
    \vertex[orange]{4}{1}{};
    \vertex[orange]{6}{1}{};
    \vertex[orange]{8}{1}{};
    \vertex[orange]{8}{2}{};
    \vertex[orange]{10}{1}{};
    \vertex[orange]{10}{2}{};
    \vertex[orange]{12}{1}{};
    \vertex[orange]{14}{1}{};
    \vertex[orange]{16}{1}{};
    \vertex[orange]{18}{1}{};

    \node at (1,-1) {$R_1$};
    \node at (5,-1) {$R_2$};
    \node at (9,-1) {$R_3$};
    \node at (13,-1) {$R_4$};
    \node at (17,-1) {$R_4'$};

    \end{tikzpicture}
    \end{center}
    \caption{The bipartite patterns in the characterizations of Lemma~\ref{lem:D-pat}.
The solid lines denote compulsory edges and the dotted lines are
compulsory non-edges in the pattern. Red vertices form a directed
cycle in $\tilde{D}(G,\Pi,<)$ or $D(G,\Pi,<)$, respectively, and
black vertices are the witnesses.}\label{fig:bip-pat}
\end{figure}

We will state now a lemma that is necessary to prove the forbidden
pattern characterizations of 2-thin graphs and (proper)
independent 2-thin graphs.

\begin{lemma}\label{lem:D-pat}
Let $G$ be a graph, $\{V^1,V^2\}$ a partition of $V(G)$, and $<$ a
partial order of $V(G)$ that is total when restricted to each of
$V^1, V^2$. Then $D(G,\{V^1,V^2\},<)$ is acyclic if and only if
$G[V^1,V^2]$ ordered according to $<$ avoids the bipartite
patterns $R_2$ and $R_3$, and $\tilde{D}(G,\{V^1,V^2\},<)$ is
acyclic if and only if $G[V^1,V^2]$ ordered according to $<$
avoids the bipartite patterns $R_1$, $R_2$, $R_4$ and $R_4'$.
Furthermore, if $G[V^1,V^2]$ has no isolated vertices, then
$\tilde{D}(G,\{V^1,V^2\},<)$ is acyclic if and only if
$G[V^1,V^2]$ ordered according to $<$ avoids the bipartite
patterns $R_1$ and $R_2$.
\end{lemma}

\begin{proof}
If the pattern $R_2$ occurs, then the vertices $2,2'$ form a
directed cycle both in $D(G,\{V^1,V^2\},<)$ and
$\tilde{D}(G,\{V^1,V^2\},<)$. If the pattern $R_1$ occurs, then
the vertices $1,1'$ form a directed cycle in
$\tilde{D}(G,\{V^1,V^2\},<)$. If the pattern $R_3$ occurs, then
the vertices $2,3,2',3'$ form a directed cycle in
$D(G,\{V^1,V^2\},<)$. If the pattern $R_4$ (resp. $R_4'$) occurs,
then the vertices $2,1'$ (resp. $1,2'$) form a directed cycle in
$\tilde{D}(G,\{V^1,V^2\},<)$.

In order to prove the converse, suppose that there is a directed
cycle in $D(G,\{V^1,V^2\},<)$. As in the proof of
Theorem~\ref{thm:char-2-thin}, we can reduce up to symmetry (since
the definition of the digraph and the patterns are symmetric) to
the following three cases.

\noindent \textbf{{Case 1:}} The cycle consists of two vertices,
$v_1 \in V^1$ and $v_2 \in V^2$.

In this case, $v_1$ and $v_2$ are not adjacent and, by definition
of the digraph, there exist $v_1' \in V^1, v_2' \in V^2$ such that
$v_1' < v_1$, $v_2' < v_2$, and $v_1'v_2, v_1v_2' \in E(G)$. So,
the pattern $R_2$ occurs in $G[V^1,V^2]$ ordered according to $<$.

\noindent \textbf{{Case 2:}} The cycle is $v_1 w_1 v_2$ such that
$v_1, w_1 \in V^1$ and $v_2 \in V^2$.

By definition of the digraph, we have $v_1 < w_1$, $v_1v_2$,
$w_1v_2 \not \in E(G)$, there exists in $V^2$ a vertex $v_2' <
v_2$ such that $w_1v_2' \in E(G)$, and there exists in $V^1$ a
vertex $v_1' < v_1$ such that $v_1'v_2 \in E(G)$. So,
$\{v_1',w_1,v_2',v_2\}$ form the pattern $R_2$ in $G[V^1,V^2]$
ordered according to $<$.

\noindent \textbf{{Case 3:}} The cycle is $v_1 w_1 v_2 w_2$ such
that $v_1, w_1 \in V^1$ and $v_2, w_2 \in V^2$.

By definition of the digraph, we have $v_1 < w_1$, $v_2 < w_2$,
$w_1v_2$, $v_1w_2 \not \in E(G)$, there exists in $V^2$ a vertex
$v_2' < v_2$ such that $w_1v_2' \in E(G)$, and there exists in
$V^1$ a vertex $v_1' < v_1$ such that $v_1'w_2 \in E(G)$. If
$w_1w_2 \not \in E(G)$, then $\{v_1',w_1,v_2',w_2\}$ form the
pattern $R_2$ in $G[V^1,V^2]$ ordered according to $<$. If,
otherwise, $w_1w_2 \in E(G)$, then $\{v_1',v_1,w_1,v_2',v_2,w_2\}$
form the pattern $R_3$ in $G[V^1,V^2]$ ordered according to $<$.

Suppose now that there is a directed cycle in
$\tilde{D}(G,\{V^1,V^2\},<)$. As in the proof of
Theorem~\ref{thm:char-2-pthin}, we can reduce up to symmetry
(since the definition of the digraph and the patterns are
symmetric) to the following three cases.

\noindent \textbf{{Case 1:}} The cycle consists of two vertices,
$v_1 \in V^1$ and $v_2 \in V^2$.

In this case, by definition of the digraph, $v_1$ and $v_2$ are
not adjacent and, on the one hand, either there exists in $V^2$ a
vertex $v_2' < v_2$ with $v_1v_2' \in E(G)$ or there exists in
$V^1$ a vertex $v_1'' > v_1$ with $v_1''v_2 \in E(G)$ and, on the
other hand, either there exists in $V^1$ a vertex $v_1' < v_1$
with $v_1'v_2 \in E(G)$ or there exists in $V^2$ a vertex $v_2''
> v_2$ with $v_1v_2'' \in E(G)$.

If the existent vertices are $v_1'$ and $v_2'$, then
$\{v_1',v_1,v_2',v_2\}$ form the pattern $R_2$ in $G[V^1,V^2]$
ordered according to $<$. If this is the case for $v_1''$ and
$v_2''$, then the pattern $R_1$ is formed by
$\{v_1,v_1'',v_2,v_2''\}$. If the existent vertices are $v_1'$ and
$v_1''$, then $\{v_1',v_1,v_1'',v_2\}$ form the pattern $R_4$.

\noindent \textbf{{Case 2:}} The cycle is $v_1 w_1 v_2$ such that
$v_1, w_1 \in V^1$ and $v_2 \in V^2$.

By definition of the digraph, we have $v_1 < w_1$, $v_1v_2$,
$w_1v_2 \not \in E(G)$, and, on the one hand, either there exists
in $V^2$ a vertex $v_2' < v_2$ with $w_1v_2' \in E(G)$ or there
exists in $V^1$ a vertex $w_1'' > w_1$ with $w_1''v_2 \in E(G)$
and, on the other hand, either there exists in $V^1$ a vertex
$v_1' < v_1$ with $v_1'v_2 \in E(G)$ or there exists in $V^2$ a
vertex $v_2''
> v_2$ with $v_1v_2'' \in E(G)$.

If the existent vertices are $v_1'$ and $v_2'$, then
$\{v_1',w_1,v_2',v_2\}$ form the pattern $R_2$ in $G[V^1,V^2]$
ordered according to $<$. If this is the case for $w_1''$ and
$v_2''$, then the pattern $R_1$ is formed by
$\{v_1,w_1'',v_2,v_2''\}$. In the case of $v_2'$ and $v_2''$, if
at least one of $v_1v_2', w_1v_2''$ is not an edge, then either
$R_1$ or $R_2$ is formed by $\{v_1,w_1,v_2',v_2''\}$. If $v_1v_2'$
and $w_1v_2''$ are edges, then $\{v_1,v_2',v_2,v_2''\}$ form
$R_4'$. In the case of $v_1'$ and $w_1''$, the vertices
$\{v_1',v_1,w_1'',v_2\}$ form $R_4$.

\noindent \textbf{{Case 3:}} The cycle is $v_1 w_1 v_2 w_2$ such
that $v_1, w_1 \in V^1$ and $v_2, w_2 \in V^2$.

By definition of the digraph, we have $v_1 < w_1$, $v_2 < w_2$,
$w_1v_2$, $v_1w_2 \not \in E(G)$, and, on the one hand, either
there exists in $V^2$ a vertex $v_2' < v_2$ with $w_1v_2' \in
E(G)$ or there exists in $V^1$ a vertex $w_1'' > w_1$ with
$w_1''v_2 \in E(G)$ and, on the other hand, either there exists in
$V^1$ a vertex $v_1' < v_1$ with $v_1'w_2 \in E(G)$ or there
exists in $V^2$ a vertex $w_2''
> w_2$ with $v_1w_2'' \in E(G)$.

Suppose the existent vertices are $v_1'$ and $v_2'$. If $w_1w_2
\not \in E(G)$, then $R_2$ is formed by $\{v_1',w_1,v_2',w_2\}$,
otherwise, $R_4$ is formed by $\{v_1',v_1,w_1,w_2\}$. Similarly,
suppose that the existent vertices are $w_1''$ and $w_2''$. If
$v_1v_2 \not \in E(G)$, then $R_1$ is formed by
$\{v_1,w_1'',v_2,w_2''\}$, otherwise, $R_4$ is formed by
$\{v_1,w_1,w_1'',v_2\}$. In the case of $v_2'$ and $w_2''$, if
$w_1w_2'' \not \in E(G)$, then $R_2$ is formed by
$\{v_1,w_1,v_2',w_2''\}$, otherwise, $R_4'$ is formed by
$\{w_1,v_2',v_2,w_2''\}$. The last case is symmetric.

To conclude, notice that, in a bipartite graph with no isolated
vertices, the patterns $R_4$ or $R_4'$ imply either $R_1$ or
$R_2$. Indeed, in the case of $R_4 = (\{v_1,v_2,v_3,v_1'\},
\{v_1v_1',v_3v_1'\}, \{v_2v_1'\})$, since $v_2$ is not an isolated
vertex, there is a vertex $v'$ with $v_2v' \in E(G)$, and either
$R_1$ or $R_2$ occurs, when $v' > v_1'$ or $v' < v_1'$,
respectively. The case of $R_4'$ is symmetric.
\end{proof}

While a characterization of $k$-thin and proper $k$-thin graphs by
forbidden induced subgraphs is open for $k \geq 2$, they may be
defined by means of forbidden patterns, due to
Corollary~\ref{cor:thin-comp-order}, i.e., we forbid the patterns
for an order $<$ that produce a clique of size $k+1$ in $G_{<}$
(resp. $\tilde{G}_{<}$). This approach leads to a high number of
forbidden patterns. 

However, the model of 2-thin graphs as monotone \Lr-graphs leads
to the following forbidden pattern characterization for the class,
with only four symmetric patterns.

Recall that an \Lr-model is \emph{blocking} if for every two
non-intersecting \Lr's, either the vertical or the horizontal
prolongation of one of them intersects the other.

\begin{theorem}\label{thm:forb-pat-2-thin}
Let us define the patterns $Z_6 =$ $(\{1,2,3,4\},$ $\{13,24\},$
$\{23\})$, $Z_7 =$ $(\{1,2,3,4,5\},$ $\{13,35\},$ $\{23,34\})$,
$Z_8 =$ $(\{1,2,3,4,5,6\},$ $\{13,46\},$ $\{23,45\})$, and $Z_9 =$
$(\{1,2,3,4,5,6\},$ $\{14,34,36\},$ $\{24,35\})$ (see
Figure~\ref{fig:patterns}). Let $G$ be a graph. The following
statements are equivalent:
\begin{itemize}
\item[$(i)$] $G$ is a 2-thin graph. \item[$(ii)$] $G$ has a
blocking monotone \Lr-model. \item[$(iii)$] $G \in$
\textsc{Ord}$(\{Z_6,Z_7,Z_8,Z_9\})$.
\end{itemize}
\end{theorem}

\begin{proof}
$(i) \Rightarrow (ii)$) Let $G$ be a 2-thin graph with partition
$V^1, V^2$ and a consistent ordering $<$. Consider the model
$\mathcal{M}_4(G)$. We have proved in Section~\ref{sec:vpg} that
it is an intersection model for $G$, given that $\mathcal{M}_1(G)$
is. We will see that it is blocking. Let us consider two
non-adjacent vertices $v_1$, $v_2$. If $v_1 \in V^1$ and $v_2 \in
V^2$, then either the horizontal prolongation of $v_1$ intersects
$v_2$ or the vertical prolongation of $v_2$ intersects $v_1$,
because the model $\mathcal{M}_1(G)$ is blocking
(Lemma~\ref{lem:model1}). If $v_1 < v_2$ and both belong to $V^1$
(resp. $V^2$), then the vertical (resp. horizontal) prolongation
of $v_2$ intersects $v_1$ (see
Figure~\ref{fig:M3-M4}). \\

\noindent $(ii) \Rightarrow (iii)$) Let $G$ be a graph admitting a
blocking monotone \Lr-model and consider the ordering of the
vertices according to the \Lr\ corners along the inverted
diagonal. It is known that the pattern $Z_6$ is not possible in an
\Lr-model for that vertex ordering~\cite{Chap-MPT}. Let us see
that the other patterns are not possible when the model is
blocking. Figure~\ref{fig:L-pat} shows schematic representations
of patterns $Z_7$, $Z_8$, and $Z_9$. The light parts are optional,
according to the undecided edges of the trigraph. In each of the
cases,
vertices labeled as $x$ and $y$ violate the blocking property. \\

\noindent $(iii) \Rightarrow (i)$) Let $G \in$
\textsc{Ord}$(\{Z_6,Z_7,Z_8,Z_9\})$ and let $v_1, \dots, v_n$ be
an ordering of the vertices avoiding the patterns $Z_6$, $Z_7$,
$Z_8$, and $Z_9$. If the order avoids $Z_1$, then $G$ is an
interval graph, in particular 2-thin. Otherwise, let $n_1$ be such
that $v_1, \dots, v_{n_1}$ avoids $Z_1$ but there exist $1 \leq i
< j < n_1+1$ such that $v_iv_{n_1+1} \in E(G)$ and $v_jv_{n_1+1}
\not \in E(G)$. Let $V^1 = \{v_1, \dots, v_{n_1}\}$ and $V^2 =
\{v_{n_1+1}, \dots, v_n\}$. Consider the ordering $<$ such that
$V^1$ is ordered increasingly according to the vertex indices and
$V^2$ is ordered decreasingly according to the vertex indices. The
graph $G[V^2]$ ordered by $<$ avoids $Z_1$, since otherwise either
$Z_7$ or $Z_8$ occurs in the original ordering. It remains to
prove that $D(G,\{V^1,V^2\},<)$ is acyclic. By
Lemma~\ref{lem:D-pat}, this is so if and only if $G[V^1,V^2]$ with
the sets ordered according to $<$ avoids the bipartite patterns
$R_2$ and $R_3$. But this holds because the original ordering
avoids $Z_6$ and $Z_9$, respectively.
\end{proof}

Notice that the blocking property is crucial, since every tree is
a monotone \Lr-graph~\cite{S-T-p-box} and trees may have
arbitrarily large thinness~\cite{B-D-thinness}.

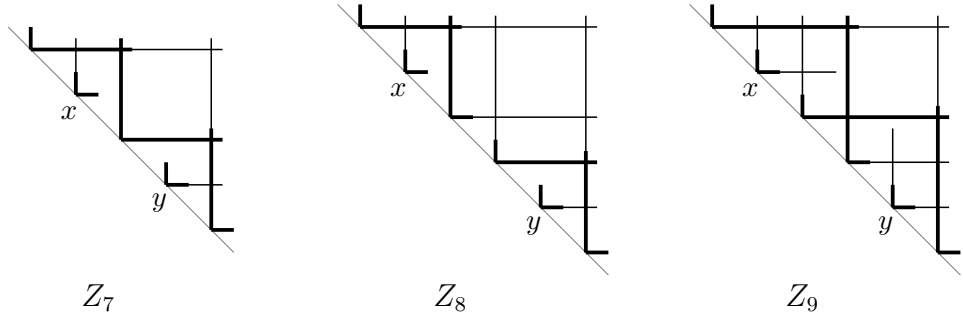
\begin{figure}[t]
\noindent\begin{minipage}{\textwidth}
\begin{minipage}[t]{\dimexpr0.33\textwidth-0.5\Colsep\relax}
\begin{center}
\begin{tikzpicture}[scale=0.3]

\draw[color=gray] (0,6) -- (10,-4);

\draw[line width=0.5mm] (1,5) -- (5.5,5); \draw[line width=0.5mm]
(1,6) -- (1,5); \draw[line width=0.2mm] (1,5) -- (9.5,5);

\draw[line width=0.5mm] (3,3) -- (4,3); \draw[line width=0.5mm]
(3,4) -- (3,3); \draw[line width=0.2mm] (3,5.5) -- (3,3);

\draw[line width=0.5mm] (5,1) -- (9.5,1); \draw[line width=0.5mm]
(5,5.5) -- (5,1);

\draw[line width=0.5mm] (7,-1) -- (8,-1); \draw[line width=0.5mm]
(7,0) -- (7,-1); \draw[line width=0.2mm] (7,-1) -- (9.5,-1);

\draw[line width=0.5mm] (9,-3) -- (10,-3); \draw[line width=0.5mm]
(9,1.5) -- (9,-3); \draw[line width=0.2mm] (9,-3) -- (9,5.5);


\node() at (2.7,2.2) {\small $x$}; \node() at (6.7,-1.8) {\small
$y$};

\node() at (4,-6) {$Z_7$};

\end{tikzpicture}
    \end{center}
\end{minipage}\hfill
\begin{minipage}[t]{\dimexpr0.33\textwidth-0.5\Colsep\relax}
    \begin{center}
\begin{tikzpicture}[scale=0.3]

\draw[color=gray] (0,6) -- (12,-6);

\draw[line width=0.5mm] (1,5) -- (5.5,5); \draw[line width=0.5mm]
(1,6) -- (1,5); \draw[line width=0.2mm] (1,5) -- (11.5,5);

\draw[line width=0.5mm] (3,3) -- (4,3); \draw[line width=0.5mm]
(3,4) -- (3,3); \draw[line width=0.2mm] (3,5.5) -- (3,3);

\draw[line width=0.5mm] (5,1) -- (6,1); \draw[line width=0.5mm]
(5,5.5) -- (5,1); \draw[line width=0.2mm] (5,1) -- (11.5,1);

\draw[line width=0.5mm] (7,-1) -- (11.5,-1); \draw[line
width=0.5mm] (7,0) -- (7,-1); \draw[line width=0.2mm] (7,5.5) --
(7,-1);

\draw[line width=0.5mm] (9,-3) -- (10,-3); \draw[line width=0.5mm]
(9,-2) -- (9,-3); \draw[line width=0.2mm] (9,-3) -- (11.5,-3);

\draw[line width=0.5mm] (11,-5) -- (12,-5); \draw[line
width=0.5mm] (11,-0.5) -- (11,-5); \draw[line width=0.2mm]
(11,5.5) -- (11,-5);


\node() at (2.7,2.2) {\small $x$}; \node() at (8.7,-3.8) {\small
$y$};

\node() at (5,-7) {$Z_8$};

\end{tikzpicture}
    \end{center}
\end{minipage}\hfill
\begin{minipage}[t]{\dimexpr0.33\textwidth-0.5\Colsep\relax}
    \begin{center}
\begin{tikzpicture}[scale=0.3]

\draw[color=gray] (0,6) -- (12,-6);

\draw[line width=0.5mm] (1,5) -- (7.5,5); \draw[line width=0.5mm]
(1,6) -- (1,5); \draw[line width=0.2mm] (1,5) -- (11.5,5);

\draw[line width=0.5mm] (3,3) -- (4,3); \draw[line width=0.5mm]
(3,4) -- (3,3); \draw[line width=0.2mm] (3,5.5) -- (3,3);
\draw[line width=0.2mm] (3,3) -- (6.5,3);

\draw[line width=0.5mm] (5,1) -- (11.5,1); \draw[line width=0.5mm]
(5,2) -- (5,1); \draw[line width=0.2mm] (5,5.5) -- (5,1);

\draw[line width=0.5mm] (7,-1) -- (8,-1); \draw[line width=0.5mm]
(7,5.5) -- (7,-1); \draw[line width=0.2mm] (7,-1) -- (11.5,-1);

\draw[line width=0.5mm] (9,-3) -- (10,-3); \draw[line width=0.5mm]
(9,-2) -- (9,-3); \draw[line width=0.2mm] (9,-3) -- (11.5,-3);
\draw[line width=0.2mm] (9,0.5) -- (9,-3);

\draw[line width=0.5mm] (11,-5) -- (12,-5); \draw[line
width=0.5mm] (11,1.5) -- (11,-5); \draw[line width=0.2mm] (11,5.5)
-- (11,-5);


\node() at (2.7,2.2) {\small $x$}; \node() at (8.7,-3.8) {\small
$y$};

\node() at (5,-7) {$Z_9$};

\end{tikzpicture}
    \end{center}
\end{minipage}%
\end{minipage}
\caption{Schematic \Lr-models of patterns $Z_7$, $Z_8$, and $Z_9$.
The light parts are optional, according to the undecided edges of
the trigraph. In each of the cases, vertices labeled as $x$ and
$y$ violate the blocking property.}\label{fig:L-pat}
\end{figure}

By combining results
from~\cite{Habib-patterns,H-H-circ2,H-M-R-orders}, we have the
following two characterization theorems. They show, among other
equivalences, that (proper) independent 2-thin graphs are
equivalent to (proper) interval bigraphs, respectively.

\begin{theorem}\label{thm:char-ind-2-thin}
Let $G$ be a graph. The following statements are equivalent:
\begin{itemize}
\item[$(i)$] $G$ is an independent 2-thin graph. \item[$(ii)$] $G$
is an interval bigraph. \item[$(iii)$] $G$ is bipartite and $G
\in$ \textsc{BicolOrd}$(\{Q_1,Q_2\})$. \item[$(iv)$] $G$ is
bipartite and $G \in$ \textsc{BiOrd}$(\{R_2,R_3\})$. \item[$(v)$]
$G \in$ \textsc{Ord}$(\{Z_5,Z_6,Z_9\})$.
\end{itemize}
\end{theorem}

\begin{proof}
$(i) \Leftrightarrow (iii)$) It is straightforward from the definition of independent thinness. \\

\noindent
$(ii) \Leftrightarrow (iii)$) It is proved in~\cite{H-H-circ2}. \\

\noindent $(i) \Rightarrow (iv)$) Let $\{V^1,V^2\}$ be a partition
of $V(G)$ into independent sets and $<$ an ordering of $V(G)$ that
is consistent with the partition $\{V^1,V^2\}$. By
Lemma~\ref{lem:D-thin}, $D(G,\{V^1,V^2\},<)$ is acyclic.
By Lemma~\ref{lem:D-pat}, $<$ avoids $R_2$ and $R_3$, so $G \in$ \textsc{BiOrd}$(\{R_2,R_3\})$. \\

\noindent $(iv) \Rightarrow (i)$) Let $\{V^1,V^2\}$ be a
bipartition of $V(G)$ and let $<$ be an order of $V^1$ and of
$V^2$ that avoids $R_2$ and $R_3$. By Lemma~\ref{lem:D-pat},
$D(G,\{V^1,V^2\},<)$ is acyclic. By Lemma~\ref{lem:D-thin}, there
is an ordering of $V(G)$ that is consistent with the partition
$\{V^1,V^2\}$, so $G$ is independent 2-thin.

\noindent $(iv) \Rightarrow (v)$) Let $\{V^1,V^2\}$ be a
bipartition of $G$ and $<$ an ordering of $V^1$ and of $V^2$ that
avoids $R_2$ and $R_3$. Consider the order of $V(G)$ such that
every vertex of $V^1$ precedes every vertex of $V^2$, $V^1$ is
ordered according to $<$ and $V^2$ is ordered according to the
reverse of $<$. This order avoids $Z_5$ because every edge has an
endpoint in $V^1$ and the other in $V^2$. It also avoids $Z_6$ and
$Z_9$, because otherwise, by the way of defining the ordering of
$V(G)$ and by the edges in the patterns, the first two (resp.
three) vertices of $Z_6$ (resp. $Z_9$) belong to $V^1$, and the
last two (resp. three) to $V^2$. Thus, with the order $<$ of $V^1$
and of $V^2$ the pattern $R_2$ (resp. $R_3$) occurs, which is
a contradiction.\\

\noindent $(v) \Rightarrow (i)$) Since a graph is independent
2-thin if and only if each of its connected components is (see,
for example,~\cite{BGOSS-thin-oper}), we may assume $G$ is
connected and non-trivial. Let $\{V^1,V^2\}$ be the bipartition of
$V(G)$ and let $<$ be an order of $V^1 \cup V^2$ that avoids
$Z_5$, $Z_6$, and $Z_9$. Since the graph is connected and the
order avoids $Z_5$, either every vertex of $V^1$ precedes every
vertex of $V^2$, or every vertex of $V^2$ precedes every vertex of
$V^1$. We may assume the first case. Consider $V^1$ ordered
according to $<$ and $V^2$ ordered according to the reverse of
$<$. We will call this partial order $<'$. Since $<$ avoids $Z_6$
and $Z_9$, $G$ ordered according to $<'$ avoids $R_2$ and $R_3$.
By Lemma~\ref{lem:D-pat}, $D(G,\{V^1,V^2\},<')$ is acyclic. By
Lemma~\ref{lem:D-thin}, there is an ordering of $V(G)$ that is
consistent with the partition $\{V^1,V^2\}$, so $G$ is independent
2-thin.
\end{proof}

Even cycles of length at least 6 are bipartite and 2-thin but not
interval bigraphs~\cite{Muller-dig}, so not independent 2-thin
graphs.

Since interval bigraphs can be recognized in polynomial
time~\cite{Muller-dig,Raf-rec-int-big}, we have the following.

\begin{corollary}\label{cor:indep-2-thin-rec}
Independent 2-thin graphs can be recognized in polynomial time.
\end{corollary}

\begin{theorem}\label{thm:char-prop-ind-2-thin}
Let $G$ be a graph. The following statements are equivalent:
\begin{itemize}
\item[$(i)$] $G$ is a proper independent 2-thin graph.
\item[$(ii)$] $G$ is a proper interval bigraph. \item[$(iii)$] $G$
is a bipartite permutation graph. \item[$(iv)$] $G$ is bipartite
and $G \in$ \textsc{BicolOrd}$(\{Q_1,Q_2,Q_3,Q_4\})$. \item[$(v)$]
$G$ is bipartite and $G \in$ \textsc{BiOrd}$(\{R_1,R_2\})$.
\item[$(vi)$] $G \in$ \textsc{Ord}$(\{Z_3,Z_4\})$.
\end{itemize}
\end{theorem}

\begin{proof}
$(i) \Leftrightarrow (iv)$) It is straightforward from the definition of proper independent thinness. \\

\noindent
$(ii) \Leftrightarrow (iii)$) It is proved in~\cite{H-H-circ2}. \\

\noindent $(ii) \Leftrightarrow (v)$) It is proved in~\cite{H-M-R-orders}.\\

\noindent
$(iii) \Leftrightarrow (vi)$) It is proved in~\cite{Habib-patterns}. \\

\noindent $(i) \Rightarrow (v)$) Let $\{V^1,V^2\}$ be a partition
of $V(G)$ into independent sets and $<$ an ordering of $V(G)$ that
is strongly consistent with the partition $\{V^1,V^2\}$. By
Lemma~\ref{lem:D-pthin}, $\tilde{D}(G,\{V^1,V^2\},<)$ is acyclic.
By Lemma~\ref{lem:D-pat}, $<$ avoids $R_1$ and $R_2$, so $G \in$ \textsc{BiOrd}$(\{R_1,R_2\})$. \\

\noindent $(v) \Rightarrow (i)$) Since a graph is proper
independent 2-thin if and only if each of its connected components
is (see, for example,~\cite{BGOSS-thin-oper}), we may assume $G$
is connected and non-trivial. Let $\{V^1,V^2\}$ be the bipartition
of $V(G)$ and let $<$ be an order of $V^1$ and of $V^2$ that
avoids $R_1$ and $R_2$. By Lemma~\ref{lem:D-pat},
$\tilde{D}(G,\{V^1,V^2\},<)$ is acyclic. By
Lemma~\ref{lem:D-pthin}, there is an ordering of $V(G)$ that is
strongly consistent with the partition $\{V^1,V^2\}$, so $G$ is
proper independent 2-thin.
\end{proof}

The bipartite claw (the subdivision of $K_{1,3}$) is bipartite and
proper 2-thin but not bipartite permutation~\cite{Koehler-thesis},
so not proper independent 2-thin.

Since bipartite permutation graphs can be recognized in linear
time~\cite{H-H-bigraphs,S-B-S-bip-perm}, we have the following.

\begin{corollary}\label{cor:indep-2-pthin-rec}
Proper independent 2-thin graphs can be recognized in linear time.
\end{corollary}

Theorems~\ref{thm:char-ind-2-thin}
and~\ref{thm:char-prop-ind-2-thin} show that, for a (proper)
$k$-thin graph with a partition $V^1,\dots,V^k$ consistent with
some ordering, not only $G[V^i]$ is a (proper) interval graph for
every $1 \leq i \leq k$, but also $G[V^i,V^j]$ is a (proper)
interval bigraph for every $1 \leq i,j \leq k$.


\section{Thinness and other width parameters}\label{sec:width}

Given a graph $G$, the \emph{pathwidth} $\pw(G)$ (resp.
\emph{bandwidth} $\bw(G)$) may be defined as one less than the
maximum clique size in an interval (resp. \emph{proper} interval)
supergraph of $G$, chosen to minimize its clique
size~\cite{K-S-bw}. It was implicitly proved
in~\cite{M-O-R-C-thinness} that
$$\indthin(G) \leq \pw(G) +1.$$ We will reproduce in Theorem~\ref{thm:ind-bw} the
proof, emphasizing the independence of the classes defined. A
characterization in~\cite{K-S-bw} of the bandwidth as a
\emph{proper pathwidth} allows to mimic the proof
in~\cite{M-O-R-C-thinness} and prove that $$\indpthin(G) \leq
\bw(G) +1.$$ This bound can be further improved for proper
thinness. We use a third equivalent definition of bandwidth,
namely $$\bw(G) = \min_f \max\{\,|f(v_{i})-f(v_{j})|:v_{i}v_{j}\in
E\,\}$$  for $f: V(G) \to \mathbb{Z}$ an injective labeling, and
Corollary~\ref{cor:thin-comp-order} to prove the following.

\begin{theorem}\label{thm:bw}
Let $G$ be a graph. Then $\pthin(G) \leq \max\{1,\bw(G)\}$.
\end{theorem}

\begin{proof}
Suppose on the contrary that $\bw(G) \geq 1$, i.e., $G$ has at
least one edge, and $\pthin(G) > \bw(G)$. By
Corollary~\ref{cor:thin-comp-order}, for every vertex order $<$ of
$G$, there is a clique of size $\bw(G)+1$ in $\tilde{G}_<$. Let
$f$ be a labeling of $V(G)$ realizing the bandwidth and $<$ be the
order induced by $f$. Suppose $v_1 < v_2 < \dots < v_b < v_{b+1}$
is a clique of $\tilde{G}_<$, where $b = \bw(G)$. As $v_1v_{b+1}
\in E(\tilde{G}_<)$, there exists either $v_0$ such that $v_0 <
v_1$, $v_0v_{b+1} \in E(G)$, and $v_0v_1\not \in E(G)$, or
$v_{b+2}$ such that $v_{b+2}
> v_{b+1}$, $v_1v_{b+2} \in E(G)$, and $v_{b+1}v_{b+2}\not
\in E(G)$. In the first case, $|f(v_{b+1}) - f(v_{0})| =
f(v_{b+1}) - f(v_{0}) \geq b+1$ and, in the second case,
$|f(v_{b+2}) - f(v_{1})| = f(v_{b+2}) - f(v_{1}) \geq b+2-1 =
b+1$. In either case, it is a contradiction with the fact that $f$
realizes the bandwidth $b$. So, $\pthin(G) \leq \bw(G)$.
\end{proof}

This bound can be arbitrarily bad, for example, for the complete
bipartite graphs $K_{n,n}$, that are proper 2-thin and have
unbounded bandwidth. However, it is tight (up to a constant
factor) for
grids~\cite{BBOSSS-thin-coloring-arxiv,C-M-O-thinness-man,Chv-grid}.

As a consequence of Theorem~\ref{thm:bw}, we have the following.

\begin{corollary}
Let $G$ be a connected graph. Then $\pthin(G) \leq |V(G)| -
\diam(G)$, where $\diam(G)$ denotes the diameter of $G$.
\end{corollary}

\begin{proof} It holds easily for graphs with only one vertex. For
graphs with at least two vertices, it holds since for a connected
graph, $\bw(G) \leq |V(G)| - \diam(G)$~\cite{C-C-D-G-bw}.
\end{proof}

A \emph{path decomposition}~\cite{R-S-minors1-pw} of a graph $G =
(V,E)$ is a sequence of subsets of vertices $(X_1,X_2, \dots,
X_r)$ such that

\begin{itemize}
\item[(1.)] $X_1 \cup \dots \cup X_r = V$.

\item[(2.)] For each edge $vw \in E$, there exists $i \in \{1,
\dots, r\}$, such that both $v$ and $w$ belong to $X_i$.

\item[(3.)] For each $v \in V$ there exist $s(v), e(v) \in \{1,
\dots, r\}$, such that $s(v) \leq e(v)$ and $v \in X_j$ if and
only if $j \in \{s(v),s(v)+1, \dots, e(v)\}$.
\end{itemize}

The width of a path decomposition $(X_1,X_2, \dots, X_r)$ is
defined as $\max_i |X_i| - 1$. The \emph{pathwidth} of a graph $G$
is the minimum possible width over all possible path
decompositions of $G$.

A \emph{proper path decomposition}~\cite{K-S-bw} of a graph $G =
(V,E)$ is a path decomposition that additionally satisfies

\begin{itemize}
\item[(4.)] For every $u, v \in V$, $\{s(u),s(u)+1, \dots, e(u)\}
\not \subset \{s(v),s(v)+1, \dots, e(v)\}$.
\end{itemize}

The \emph{proper pathwidth} of a graph $G$ is the minimum possible
width over all possible proper path decompositions of $G$. Kaplan
and Shamir~\cite{K-S-bw} proved that the proper pathwidth of a
graph equals its bandwidth.

\begin{theorem}\label{thm:ind-bw} For a graph $G$,
$\indthin(G) \leq \pw(G)+1$ and $\indpthin(G) \leq \bw(G)+1$.
\end{theorem}

\begin{proof} (slight modification of the one in~\cite{M-O-R-C-thinness}) Consider an optimal
(proper) path decomposition $(X_1,X_2, \dots, X_r)$ of width $q$.
Let $k = q + 1$ be the cardinality of the biggest set. We
demonstrate that the graph is (proper) independent $k$-thin. We
first describe an ordering and then give a description of how we
can assign the vertices to $k$ classes, in order to get a
partition into $k$ independent sets which is (strongly) consistent
with the ordering.

We order the vertices $v$ according to $s(v)$, breaking ties
arbitrarily. Notice that, in the proper case, by~(4.), if
$s(v)=s(w)$ then also $e(v)=e(w)$, and if $s(v) < s(w)$ then also
$e(v) < e(w)$.

We assign the vertices in each $X_i$ to pairwise distinct classes,
for $1 \leq i \leq r$. We can do it for $i=1$, since $|X_1| \leq
k$. On each step, for $i > 1$, we assign the vertices of $X_{i}
\setminus X_{i-1}$ to pairwise distinct classes that are not used
by the vertices in $X_{i} \cap X_{i-1}$. This can be done because
$|X_i| \leq k$.

No vertex is assigned to more than one class because of
condition~(3.). The classes are independent sets because of
condition~(2.).

Let $v \in V(G)$, with $s(v)=i$. Then all neighbors $u$ of $v$
smaller than $v$ are also present in $X_i$, i.e., $e(u) \geq i$.
Suppose $u < z < v$, $uv \in E(G)$, $u, z$ in the same class. Note
that if $z$ and $u$ are in the same class, there is no subset
$X_k$ such that $z, u$ are both in $X_k$. Thus $u < z$ tells us
that $s(u) < s(z)$. Also, $z < v$ tells us that $s(z) \leq s(v)$.
But by the claim above, $e(u) \geq s(v)$, since $u$ is a neighbor
of $v$ smaller than $v$. Thus, by~(3.), $u \in X_{s(z)}$, a
contradiction with the fact that no set contains two vertices of
the same class. Thus, the partition and the ordering of $V(G)$ we
obtained are consistent. It follows that $\indthin(G) \leq \pw(G)
+ 1$.

For a (proper) path decomposition, the reverse of the order
defined is a decreasing order by $e(v)$. So the same argument
proves that, in that case, the partition and the ordering of
$V(G)$ we obtained are strongly consistent. It follows that
$\indpthin(G) \leq \bw(G) + 1$.
\end{proof}

These bounds are tight, respectively, for interval and proper
interval graphs, where both the (proper) independent thinness and
the (proper) pathwidth plus one equal the clique number of the
graph, but can be arbitrarily bad, for example, for complete
bipartite graphs, that are proper independent 2-thin but have
unbounded pathwidth.

A recent result shows that the proper independent thinness is also
bounded above by a function of the pathwidth.

\begin{theorem}\cite{B-B-M-P-convex-wads}\label{thm:ind-p-pw} For a graph $G$,
$\indpthin(G) \leq 2^{\pw(G)}(\pw(G)+1)$.
\end{theorem}

\section*{Acknowledgements}

We are deeply grateful to the anonymous referees, who helped us to
significantly improve the presentation of our work. This work was
partially supported by CONICET (PIP 11220200100084CO), ANPCyT
(PICT-2021-I-A-00755), UBACyT (20020170100495BA and
20020160100095BA), and MATHAMSUD Regional Program MATH190013.


\end{document}